\documentclass[11pt]{article}
\usepackage[margin=1in]{geometry}
\usepackage{caption}
\usepackage{float}
\RequirePackage[authoryear]{natbib}
\usepackage{bbm}
\usepackage{mathrsfs}
\usepackage{color,xcolor}
\usepackage{graphicx}
\usepackage{lmodern}
\usepackage{anyfontsize}
\usepackage{algorithm}
\usepackage{algpseudocode}
\usepackage{comment}
\usepackage{longtable}
\usepackage{booktabs}
\usepackage{multicol}
\usepackage{enumitem}
\usepackage{hyperref}
\usepackage{amsfonts}
\usepackage{amsmath}
\usepackage{amssymb}
\usepackage{amsthm}
\usepackage{algorithm}
\usepackage{algpseudocode}
\usepackage{bm}
\usepackage{subfigure}
\hypersetup{
  colorlinks,
  linkcolor={red!50!black},
  citecolor={blue},
  urlcolor={blue!80!black}
}
\allowdisplaybreaks

\usepackage{setspace}
\usepackage{authblk}

\theoremstyle{plain}
\newtheorem{theorem}{Theorem}
\newtheorem{proposition}{Proposition}

\newtheorem{lemma}{Lemma}
\newtheorem{corollary}{Corollary}
\newtheorem{example}{Example}

\theoremstyle{definition}
\newtheorem{assumption}{Assumption}
\newtheorem{definition}{Definition}

\newcommand\blfootnote[1]{%
  \begingroup
  \renewcommand\thefootnote{}\footnote{#1}%
  \addtocounter{footnote}{-1}%
  \endgroup
}

\def\spacingset#1{\renewcommand{\baselinestretch}%
{#1}\small\normalsize}

\spacingset{1}
\usepackage{authblk} 
\title{Covariate-assisted Grade of Membership Models via Shared Latent Geometry}
\author{Zhiyu Xu}
\author{Yuqi Gu}

\affil{Department of Statistics, Columbia University}

\date{}

\newcommand{\di}{\mathcal{P}_{\text{diag}}}
\newcommand{\od}{\mathcal{P}_{\text{off-diag}}}

\newcommand{\cR}{\mathcal{R}}
\newcommand{\cG}{\mathcal{G}}
\newcommand{\cX}{\mathcal{X}}
\newcommand{\cD}{\mathcal{D}}
\newcommand{\boldPi}{\boldsymbol{\Pi}}
\newcommand{\boldTheta}{\boldsymbol{\Theta}}
\newcommand{\boldM}{\mathbf{M}}
\newcommand{\boldX}{\mathbf{X}}
\newcommand{\boldR}{\mathbf{R}}
\newcommand{\boldU}{\mathbf{U}}
\newcommand{\boldG}{\mathbf{G}}

\newcommand{\boldSigma}{\boldsymbol{\Sigma}}

\newcommand{\E}{{E}}
\newcommand{\rP}{\text{pr}}

\newcommand{\Um}{\mathbf{U}^{(m)}}
\newcommand{\UmT}{\mathbf{U}^{(m)\top}}
\newcommand{\Hm}{\mathbf{H}^{(m)}}

\newcommand{\cGm}{\mathbf{G}^{(m)}}

\newcommand{\boldH}{\mathbf{H}}
\newcommand{\boldA}{\mathbf{A}}
\newcommand{\boldE}{\mathbf{E}}
\newcommand{\Pm}{\mathcal{P}_{-m,:}}
\newcommand{\Uml}{\mathbf{U}^{(m,l)}}
\newcommand{\Hml}{\mathbf{H}^{(m,l)}}
\newcommand{\UmlT}{\mathbf{U}^{(m,l)\top}}
\newcommand{\zetaop}{\zeta_{\text{op}}}
\newcommand{\bu}{\boldsymbol{u}}
\newcommand{\boldL}{\mathbf{L}}
\newcommand{\cL}{\mathcal{L}}
\def\hat{\widehat}
\def\tilde{\widetilde}

\usepackage{xcolor}

\begin{document}

\maketitle
\begin{abstract}
The grade of membership model is a flexible latent variable model for analyzing multivariate categorical data through individual-level mixed membership scores. In many modern applications, auxiliary covariates are collected alongside responses and encode information about the same latent structure. Traditional approaches to incorporating such covariates typically rely on fully specified joint likelihoods, which are computationally intensive and sensitive to misspecification. We introduce a covariate-assisted grade of membership model that integrates response and covariate information by exploiting their shared low-rank simplex geometry, rather than modeling their joint distribution. We propose a likelihood-free spectral estimation procedure that combines heterogeneous data sources through a balance parameter controlling their relative contribution. To accommodate high-dimensional and heteroskedastic noise, we employ heteroskedastic principal component analysis before performing simplex-based geometric recovery. Our theoretical analysis establishes weaker identifiability conditions than those required in the covariate-free model, and further derives finite-sample, entrywise error bounds for both mixed membership scores and item parameters. These results demonstrate that auxiliary covariates can provably improve latent structure recovery, yielding faster convergence rates in high-dimensional regimes. Simulation studies and an application to educational assessment data illustrate the computational efficiency, statistical accuracy, and interpretability gains of the proposed method. The code for reproducing these results is open-source and available at \texttt{\url{https://github.com/Toby-X/Covariate-Assisted-GoM}}.
\end{abstract}

\blfootnote{Emails: \texttt{xu.zhiyu@columbia.edu}; \texttt{yuqi.gu@columbia.edu}.}

\section{Introduction}
Multivariate categorical data arise in a wide range of scientific applications, including social science surveys \citep{Erosheva2007Describing}, genomic studies \citep{Dey2017Visualizing}, educational and psychological assessments \citep{Shang2021Partial-Mastery}, and topic modelling \citep{Blei2003}. A central inferential goal in such settings is to uncover interpretable latent structure from observed responses. The grade of membership (GoM) model \citep{Woodbury1978Mathematical, erosheva2002grade, Airoldi2014Handbook} provides a flexible and widely used framework for this purpose, allowing each individual to exhibit partial membership across multiple latent extreme profiles, rather than belonging exclusively to a single class. At the individual level, each subject has a vector of latent mixed membership scores, which are nonnegative and sum to one. At the population level, each latent extreme profile is characterized by its unique distribution over the multivariate categorical responses. By combining a small number of prototypical latent profiles with individual-specific mixed membership scores, the GoM model offers a flexible latent variable framework with both strong expressive power and clear interpretability.

The flexibility of the GoM model comes at the cost of estimation and computation challenges. Due to the need to integrate out the mixed membership scores, computationally intensive Markov chain Monte Carlo algorithms are dominant approaches to estimating the model and its variants in the literature \citep{erosheva2002grade, Erosheva2007Describing, Gormley2009grade, Gu2023DimensionGrouped}. Similarly, the marginal maximum likelihood estimation approach requires advanced optimization to tackle the intractable integral, leading to inefficient computation. To circumvent the troublesome integral, joint maximum likelihood estimation \citep{erosheva2002grade} treats the mixed membership scores as fixed parameters, instead of latent random variables, and jointly updates all parameters until convergence. However, joint MLE remains sensitive to initialization and does not scale well to high-dimensional settings.
These challenges have motivated recent interest in likelihood-free alternatives. In particular, \citet{chen2024spectral} exploited the low-rank structure of the population response matrix to develop a spectral estimator, which is computationally efficient and consistent under mild identifiability conditions.

In many modern data collection settings, response data are accompanied by auxiliary covariates measured on the same individuals. For example, in large-scale educational assessments such as trends in international mathematics and science study \citep{mullis2017timss}, students' social economic background is gathered. In genomics, recent technological advances allow for the simultaneous collection of the spatial information alongside gene expression data \citep{staahl2016visualization}. In biodiversity studies \citep{piirainen2023species}, species count for every location is supplemented with habitat information of each location.
Such covariates often encode information about the \emph{same latent structure} that governs the primary responses. Traditional approaches to incorporating covariates in mixed membership models typically rely on fully specified joint likelihoods \citep{Airoldi2014Handbook}, which can be difficult to justify for mixed-type covariates and are sensitive to misspecification. Moreover, these approaches do not directly address how auxiliary information may affect identifiability or finite-sample estimation error.

In this paper, we show that covariates can be incorporated into grade of membership analysis in a clean likelihood-free way. Rather than modelling the joint distribution of responses and covariates, we exploit the fact that both data sources provide noisy observations of the same underlying latent mixed membership structure. Both the response matrix and the covariate matrix admit low-rank decompositions driven by the same individual-level latent membership scores. This shared geometric structure allows covariates to be integrated through their contribution to the latent singular subspace, without imposing additional distributional assumptions.

We introduce a covariate-assisted grade of membership model and propose a spectral estimation procedure that jointly leverages response and covariate information. Our approach constructs a weighted combination of carefully modified Gram matrices obtained from the responses and covariates, with a balance parameter controlling their relative influence. 
Statistically, the balance parameter $\alpha$ weights two noisy data matrices reflecting the same latent singular subspace; its role is to control the contribution of response- and covariate-derived information according to their relative signal-to-noise ratios.
High-dimensionality and heteroskedastic noise are intrinsic features of multivariate binary responses and mixed-type covariates. To accommodate these features, we employ heteroskedastic principal component analysis \citep{Zhang2021} to estimate the shared singular subspace.
Based on the singular subspace estimate, we then exploit the simplex geometry of mixed membership scores to recover individual-level memberships and population-level parameters.

Our theoretical analysis establishes several new results. First, we show that incorporating covariates weakens the identifiability requirements relative to the covariate-free GoM model, enlarging the class of models for which the latent structure is uniquely recoverable. Second, we derive finite-sample, entrywise error bounds for both mixed membership scores and other parameters. These bounds reveal that auxiliary covariates can lead to strictly faster convergence rates, particularly in high-dimensional regimes where the number of items or covariates grows with the sample size. Third, we analyse the role of the balance parameter and demonstrate that naively stacking response and covariate matrices can be suboptimal when their signal-to-noise ratios differ.

To make the proposed method fully practical, we develop a data-driven procedure for selecting the balance parameter based on cross-validated prediction error. Extensive simulation studies demonstrate that the covariate-assisted estimator is both more accurate and more computationally efficient than existing likelihood-based and spectral alternatives. We further illustrate the method using data from a large-scale educational assessment, where incorporating student-level covariates yields more interpretable latent proficiency profiles and improved predictive performance.

The rest of this paper is organized as follows. Section \ref{sec:method} introduces the model formulation and the covariate-assisted spectral estimation procedure. Section \ref{sec:theory} presents identifiability results and finite-sample error bounds. Section \ref{sec:simulation} presents simulation studies to assess the performance of the proposed method. Section \ref{sec:application} illustrates the proposed method on an educational assessment dataset.
Section \ref{sec:conclusion} concludes the paper and discusses future directions. The proofs of all theoretical results are in the supplementary materials.

\section{Methodology}
\label{sec:method}
\subsection{Notation}
We write $A_n\lesssim B_n$ if there exists some constant $c$ such that $A_n\leq cB_n$ for sufficiently large $n$, write $A_n=\tilde{O}(B_n)$ if $A_n\lesssim B_nf(\log n)$, where $f$ is a polynomial function of some finite degree. For any matrix $\mathbf{A}$, let $\|\mathbf{A} \|$ denote its spectral norm, 
and $\|\mathbf{A} \|_{F}$ the Frobenius norm. Denote the two-to-infinity norm by $\|\mathbf{A}\|_{2,\infty} :=\max_{i}\|\mathbf{A}_{i,:}\|_2 $, and the infinity norm by $\|\cdot\|_{\infty}:=\max_{i,j}|A_{ij}| $. For any random vector $\boldsymbol{u}$, let $\|\cdot\|_2$ denote its $L_2$ norm. For any random variable $Y$, define its sub-Gaussian norm as $\|Y\|_{\psi_2}:= \inf\{t>0: {E}[\exp(Y^2/t^2)]\leq 2\}$. Its equivalent definitions can be referred to \citet{vershynin2018high}. Let $\mathbf{I}(\cdot)$ denote the indicator function. Let $\Delta^r$ denote the $r$ dimension simplex, i.e. $\forall \pi \in \Delta^r$, $\sum_{k=1}^r\pi_k=1$ and $\pi_k\geq 0$. Denote the rank-$r$ identity matrix by $\mathbf{I}_r$, and denote the all-one vector by $\boldsymbol{1}_r:=(1,\dots,1)\in\mathbb{R}^r$.

For any symmetric matrix $\mathbf{S} \in \mathbb{R}^{d_1\times d_1}$ of rank $r$, denote its $i$th largest eigenvalue to be $\lambda_{i}(\mathbf{S})$. For any matrix $\mathbf{A} \in \mathbb{R}^{d_1\times d_2}$ of rank $r$, let $\sigma_i(\mathbf{A})$ denote the $i$th largest singular value of $\mathbf{A}$, and $\sigma_i(\mathbf{A})=\sqrt{\lambda_{i}(\mathbf{A}^\top \mathbf{A})}$.
Define the condition number for any matrix of $\mathbf{A}$ as $\kappa(\mathbf{A}) = \sigma_1(\mathbf{A})/\sigma_r(\mathbf{A})$. 
Define the incoherence parameter \citep{Chen2021}
for $\mathbf{A}$ to be $\mu = \max\bigl\{{d_1d_2\|\mathbf{A}\|_{\infty}^2}/{\|\mathbf{A}\|_F^2}, {d_1\|\mathbf{U}\|_{2,\infty}^2}/{r}, {d_2\|\mathbf{V}\|_{2,\infty}^2}/{r}\bigr\}$, where $\boldU$ and $\mathbf{V}$ are respectively the top-$K$ left and right singular vectors of $\mathbf{A}$. 
For a square matrix $\mathbf{A}=(a_{ij})\in \mathbb{R}^{d\times d}$,  define $\di(\mathbf{A})=\text{diag}\left( a_{11},a_{22},\dots,a_{dd} \right)$. 
The hollowing operator is defined as $\od(\mathbf{A}):=\mathbf{A}-\di(\mathbf{A})$. Define $\mathcal{O}^{r\times r}$ as the set of all $r\times r$ orthonormal matrices.

\subsection{Model set up and motivation}

In this work, we focus on multivariate binary responses, which are frequently encountered in social and biomedical sciences, such as yes/no responses in social survey items \citep{Erosheva2007Describing}, correct/wrong answers to questions in educational assessments \citep{Shang2021Partial-Mastery}, and presence/absence of certain gene expressions in single-cell sequencing data \citep{BravoGonzalez-Blas2019cisTopic}. We will briefly discuss how to extend the proposed method to multivariate categorical data in the discussion section.

We first introduce the notation. Denote the data matrix by $\boldR = \{R_{ij}\}\in \{0,1\}^{N\times J}$, collecting the responses of $N$ subjects to $J$ items.
For each subject $i \in[N]$, we denote the mixed membership score as $\bm{\pi}_i\in \Delta^{K-1}$. Define the collection of mixed membership scores for every object to be $\boldPi=[ \bm{\pi}_1\mid \cdots \mid \bm{\pi}_N ]^{\top}\in [0,1]^{N\times K}$.
At the population level, define
\begin{equation*}
    \theta_{jk} = \rP(R_{ij}=1 \mid i\text{ in the }k\text{th extreme profile})
\end{equation*}
In the latent class model or grade of membership literature, the $\theta_{jk}$ are also called the \emph{item parameters}, which characterize the properties of the $J$ items.
The conditional distribution of $R_{ij}$ given $\bm{\pi}_i$ and $\bm{\theta}_j$ is
    $\rP(R_{ij}=1\mid \bm{\pi}_i,\bm{\theta}_j) = \sum_{k=1}^K\pi_{ik}\theta_{jk},$
which is a convex combination of the item parameters $\theta_{jk}$ with weights defined by the mixed membership scores $\pi_{ik}$.

In the literature of grade of membership models, there are two typical methods to estimate the mixed membership scores $\bm{\pi}_i$, adopting the random-effect and the fixed-effect perspective, respectively  \citep{erosheva2002grade}. The random effect perspective is traditionally more popular, which assumes $\bm{\pi}_i$ is random variable following some distribution on the simplex $\Delta^K$.
A common choice is the Dirichlet distribution, denoted as $D_{\bm{\alpha}}(\bm{\pi}_i)$, where $\bm{\alpha}=(\alpha_1,\ldots,\alpha_K)$ is the parameter vector for the Dirichlet distribution. Then, the marginal likelihood of population parameters is
\begin{equation}
    \label{eqn:marginal}
    L(\boldTheta, \bm{\alpha}\mid \boldR) = \prod_{i=1}^N\int\prod_{j=1}^J \left(\sum_{k=1}^K\pi_{ik}\theta_{jk} \right)^{R_{ij}} \left(1- \sum_{k=1}^K\pi_{ik}\theta_{jk} \right)^{1-R_{ij}}dD_{\bm{\alpha}}(\bm{\pi}_i).
\end{equation}
Given such complicated integral of the mixed memberships, Bayesian MCMC methods \citep{Erosheva2007Describing, Gu2023DimensionGrouped, bhattacharya2012simplex} are predominant estimation paradigms.

From the fixed-effect perspective, mixed membership scores are directly viewed as parameters. The joint likelihood function for both $\boldPi$ and $\boldTheta$ is
\begin{equation}
    \label{eqn:joint}
    L(\boldPi, \boldTheta \mid \boldR) = \prod_{i=1}^N\prod_{j=1}^J \left(\sum_{k=1}^K\pi_{ik}\theta_{jk} \right)^{R_{ij}} \left(1- \sum_{k=1}^K\pi_{ik}\theta_{jk} \right)^{1-R_{ij}}.
\end{equation}
This avoids the complicated integral and enables a more efficient joint MLE algorithm that maximizes (\ref{eqn:joint}) to obtain estimates \citep{erosheva2002grade}. Still, the likelihood is highly nonconvex, and thus the algorithm suffers severely from sensitivity to initialization and slow convergence.

Departing from the likelihood perspective,
a recent study \citet{chen2024spectral} made the key observation that the binary response matrix can be written as
\begin{equation}
    \label{eqn:R decom}
    \boldR = \underbrace{\boldPi\boldTheta^\top}_{\text{``signal''}} + \underbrace{\boldE^R}_{\text{``noise''}},
\end{equation}
where $\boldE^R$ is the mean zero noise matrix with entries $E^R_{ij} = R_{ij} - E[R_{ij}]$. For notational simplicity, from now on, we will denote $\cR=E[\boldR]$. The low-rank factorization implies fast spectral
methods may be applied to estimate the parameters for the grade of membership models. Based on an interesting geometric structure, 
\citet{chen2024spectral} propose an efficient and accurate singular value decomposition based algorithm for estimation.

Building on this key observation, we assume that the covariates share the same mixed membership structure as the responses at the population level.
In many applications, this is a reasonable and common assumption. Specifically, covariate-assisted graph learning \citep{Binkiewicz2017, Yan2018} usually assumes the network exhibits the same latent structure as the covariates. For example, \citet{Binkiewicz2017} aims to uncover cluster structure underlying the brain connectivity network. They utilize the spatial location and the brain atlas region as covariates. Biologically, the spatial location of the brain and the brain atlas region are good representations of which part of the brain are more likely to be connected. Hence, \citet{Binkiewicz2017} assumes the two covariates share the same latent block structure as the connectivity matrix. 
Denote the covariate matrix as $\boldX \in \mathbb{R}^{N\times W}$. 
Using a similar spirit, we assume the covariates $\boldX$ satisfy that
\begin{equation}
    \label{eqn:X decom}
    \boldX = \underbrace{\boldPi\boldM^\top}_{\text{"signal"}} + \underbrace{\boldE^X}_{\text{"noise"}},
\end{equation}
where $\boldM\in [-\xi,\xi]^{W\times K}$, $\boldE^X$ is a mean zero noise matrix. Similarly as the response matrix, denote $\cX=E[\boldX]$.

\subsection{Covariate-assisted spectral estimation of grade of membership models}

To combine two sources of information from the responses and covariates, Equations (\ref{eqn:R decom}) and (\ref{eqn:X decom}) suggest the signal part of $\boldR$ and $\boldX$ both capture the mixed membership structure of $N$ subjects. Therefore, both signal gram matrices $\cR\cR^\top = \boldPi \boldTheta^\top\boldTheta \boldPi^\top$ and $\cX\cX^\top=\boldPi\boldM^\top \boldM\boldPi$ are informative about the mixed membership scores $\boldPi$.
Although the signal parts share the same mixed membership structure, it is important to notice that in reality, the two data matrices can have very different noise levels. It is essential to account for this difference to achieve accurate estimation of the shared latent space \citep{baharav2025stacked}. Hence, we introduce a balance parameter $\alpha$ to balance these two sources of information. 
We start by considering a weighted sum of the signal gram matrices, which bears the following two low-rank decomposition,
\begin{equation}
    \label{eqn:oracle_g}
    \cG := \cR\cR^\top + \alpha \cX\cX^\top \stackrel{(i)}{=} \boldPi\left( \boldTheta^\top \boldTheta + \alpha \boldM^\top \boldM \right)\boldPi^\top := \boldPi\cD\boldPi^\top
    \stackrel{(ii)}{=} \boldU \boldsymbol\Lambda \boldU^\top.
\end{equation}
where $(i)$ follows from the latent structure of $\cR$ and $\cX$, and $(ii)$ is the eigenvalue decomposition of $\cG$.  Here, $\boldU$ is a $N\times K$ matrix collecting the $K$ eigenvectors of $\cG$ and satisfies $\boldU^\top\boldU = \mathbf{I}_K$, and $\boldsymbol\Lambda =\text{diag}(\lambda_1,\ldots,\lambda_K)$ collects the eigenvalues. We call $\boldU$ the \emph{signal eigenspace}.
The matrices $\cR\cR^\top$ and $\cX\cX^\top$ are both reflecting the same signal eigenspace $\boldU$, but with different noise levels. The balance parameter $\alpha$ controls their relative contribution to estimating $\boldU$. When $\alpha=0$, the method reduces to the covariate-free estimator; when $\alpha$ is too large, the estimated eigenspace can be dominated by noisy covariates.

First, for identifiability of the model, we introduce a mild condition.

\begin{definition}[Pure subject]
    Subject $i$ is a pure subject for an extreme profile $k$ if the only positive entry of $\boldsymbol{\pi}_i$ is located at index $k$; that is, $\pi_{ik}=1$ and $ \pi_{il} = 0$ for $l\neq k$.
\end{definition}
\begin{assumption}
    \label{assump:extreme}
    $\boldPi$ satisfies that each of the $K$ extreme latent profiles has at least one pure subject.
\end{assumption}
Assumption \ref{assump:extreme} implies that after certain row and column permutation, there exists an identity submatrix $\mathbf I_K$ inside $\boldPi$. This condition directly leads to the following proposition.

\begin{proposition}
    \label{prop:geom}
    Under Assumption \ref{assump:extreme}, the eigenspace $\mathbf{U}$ of $\cG=\cR\cR^\top + \alpha \cX\cX^\top$ in \eqref{eqn:oracle_g} satisfies,
    \begin{equation}
    \label{eqn:geom}
        \boldU = \boldPi \boldU_{S,:},
    \end{equation}
    where $S=(S_1,\dots,S_K)$ is one set of the indices of the $K$ pure subjects for each extreme latent profile; i.e., $\boldPi_{S,:} = \mathbf I_K$. Further, $\boldPi$ can also be represented as,
    \begin{equation}
        \label{eqn:Pi est}
        \boldPi = \boldU\boldU_{S,:}^{-1}.
    \end{equation}
\end{proposition}

Similar conditions and geometric structures have been discovered in related works on mixed membership models \citep{Mao2019, ke2022using, chen2024spectral}.
Equation (\ref{eqn:geom}) implies $\boldU_{i,:} = \sum_{k=1}^K \pi_{ik} \boldU_{S_k,:}$ for every $i\in[N]$. This means each row of $\boldU$ lies in a $K$ dimensional simplex. The vertices of this simplex are $\boldU_{S_1,:}, \ldots, \boldU_{S_K,:}$, corresponding to one set of $K$ pure subjects.
Therefore, under Assumption \ref{assump:extreme}, as long as we can estimate $S=(S_1,\ldots,S_K)$, we can recover $\boldPi$ via (\ref{eqn:Pi est}). 

One of the most popular algorithms to find the vertices of a simplex is the successive projection algorithm \citep{Araujo2001successive, Gillis2014}. 
This algorithm successively finds which row of $\boldU$ has the largest $l_2$ norm, denoted as $\bm{u}_k$. Then, it projects other rows of $\boldU$ to the orthogonal space of $\bm{u}_k$. By sequentially performing this procedure $K$ times, one can uncover the $K$ vertices of the simplex. The algorithm is presented in Step.3 in Algorithm \ref{alg:gom}.

Thus far, we have described a procedure to recover the parameters 
for the oracle signal matrices. However, with noise, estimation becomes more complex. By substituting the sample version of $\cR$ and $\cX$ in (\ref{eqn:oracle_g}), we have
\begin{equation}
    \label{eqn:decom}
    \boldR\boldR^\top + \alpha\boldX\boldX^\top = \cR\cR^\top + \boldE^R\cR^\top + \cR\boldE^{R\top} + \boldE^R\boldE^{R\top} + \alpha \left(\cX\cX^\top + \boldE^X\cX^\top + \cX\boldE^{X\top} + \boldE^X\boldE^{X\top}\right).
\end{equation}
Under the previous assumption of mean zero independent noise, the expectations of the two quadratic terms of the noise, $\boldE^R\boldE^{R\top}$ and $\boldE^X\boldE^{X\top}$, are diagonal matrices. 
The diagonal entries are respectively $\{\sum_{j=1}^J\text{var}(e_{ij}^R)\}_{i=1,\ldots,N}$ and $\{\sum_{l=1}^W\text{var}(e_{il}^X)\}_{i=1,\ldots,N}$. Hence, $\boldR\boldR^\top + \alpha\boldX\boldX^\top$ is a biased estimate of the oracle weighted sum $\cG=\cR\cR^\top + \alpha \cX\cX^\top$, and the bias only occurs in the diagonal entries.
Further, the bias becomes more severe under high-dimensional and heteroscedastic noise settings, which are our motivating settings. To be specific, in modern data collection processes, especially in social sciences and genomics, scientists are able to collect large numbers of both item responses and covariates. This leads to the challenge of high-dimensionality with both $J$ and $W$ being large. As one can see from the explicit form of the bias, as $J$ and $W$ get larger, the bias gets larger. An even more challenging issue is that multiple categorical responses and mixed types of covariates inherently create heteroskedastic noise. These two challenges make $\boldR\boldR^\top + \alpha \boldX\boldX^\top$ an unsatisfactory approximation to $\cG$. Hence, the eigenspace of $\boldR\boldR^\top + \alpha \boldX\boldX^\top$ is a poor estimator of the signal eigenspace $\boldU$.

We propose to leverage recent developments in eigen-space estimation to provide a better approximation to $\cG$. Motivated by the observation that only the diagonal entries of the gram matrices are biased, \citet{caiSubspaceEstimationUnbalanced2020} deletes the diagonal entries of the gram matrices when the number of columns greatly exceeds the number of rows and heteroscedasticity is severe. Although at the same time, signal in the diagonal entries is also deleted along with the noise, this procedure still benefits the estimation when the signal is weak and the noise is strong.
To avoid the downside of eliminating the signal in the diagonal entries, 
\citet{Zhang2021} further proposed the heteroskedastic principal component analysis (HeteroPCA).
HeteroPCA iteratively imputes the diagonal entries of the gram matrix using the off-diagonals of its low-rank approximation.
Through this iterative approach, HeteroPCA can achieve a better estimate of the left singular subspace both in theory and in simulation studies \citep{Yan2021}.
The full algorithm is presented in Algorithm \ref{alg:HeteroPCA}. Notably, empirical evidence shows a small number of iterates, like 10, succeeds in having a significantly better estimate of the eigen-subspace.

Motivated by the above recent developments, we propose to first apply HeteroPCA to both $\boldR\boldR^\top$ and $\boldX\boldX^\top$ to obtain better estimators of the signal gram matrices $\boldG_R$ and $\boldG_X$.
Then, we define the refined estimate of $\cG$ as
\begin{equation*}
    \boldG = \boldG_R + \alpha\boldG_X.
\end{equation*}
From now on, we will denote the top $K$ singular value decomposition of $\cG$ to be $\boldU \mathbf{\Lambda} \mathbf{V}^\top $, and top-$K$ singular value decomposition of $\boldG$ to be $\hat{\boldU}\hat{\mathbf{\Lambda}}\hat{\mathbf{V}}^\top$. 

For the estimation of $\boldTheta$, we rely on the population level decomposition
\begin{equation*}
    \cL:= [\cR \mid \sqrt{\alpha}\cX] = \boldPi[\boldTheta \mid \sqrt{\alpha}\boldM]^\top.
\end{equation*}
Following the previous procedures, we have obtained an estimate of $\boldPi$. An estimate of $\cL$ can be obtained by projecting the sample level $\boldL=[\boldR\mid \alpha \boldX]$ to the column space of $\hat{\boldU}$, that is $\hat{\boldU}\hat{\boldU}^\top \boldL$.
Therefore, following a similar idea as multivariate linear regression, we can have a closed form estimate for $\boldTheta$, which is written in Step.7 of Algorithm \ref{alg:gom}.

We present the complete algorithm in Algorithm \ref{alg:gom}. Step.3 describes each step of the successive projection algorithm. Due to noise, each row of the estimated $\tilde{\boldPi}$ may not lie in a $K$ dimensional simplex. Hence, we find the closest vector inside the unit simplex by employing Algorithm 1 in \citet{condat2016fast}.
Step.8 is not necessary for our estimation procedure; but it is the default setting in package \texttt{sirt} in \texttt{R} for the joint MLE method for the grade of membership model. For a fair comparison, we set $\varepsilon=0.001$ to match their settings.

\begin{algorithm}[h!]
\caption{Heteroskedastic principal component analysis \citep{Zhang2021}.}
\label{alg:HeteroPCA}
\begin{algorithmic}[1]
    \Require Given gram matrix $\boldG = \boldR\boldR^\top \text{ or } \boldX\boldX^{\top}\in \mathbb{R}^{N\times N}$, rank $K$, maximum iterations $T$, error threshold $\epsilon$.
    \State Set $\mathbf{N}^{(0)}=\od(\mathbf{G})$, $t=0$.
    \While{$\hat{\epsilon} \geq \epsilon$ and $t \leq T$}
        \State Perform singular value decomposition on $\mathbf{N}^{(t)}$ and let its top $K$ approximation be $\tilde{\mathbf{N}}^{(t)}= \sum_{i=1}^K \sigma_i^{(t)} \boldsymbol{u}^{(t)}_i\boldsymbol{v}^{(t)\top}_i$.
        \State Let $\mathbf{N}^{(t+1)}= \di\bigl(\tilde{\mathbf{N}}^{(t)}\bigr) + \od\left(\mathbf{N}^{(t)}\right)$.
        \State Set $t = t+1$ and $\hat{\epsilon} = \|\mathbf{N}^{(t)}-\mathbf{N}^{(t-1)}\|$.
    \EndWhile
    \State Set $\tilde{\mathbf{G}} = \mathbf{N}^{(t)}$.
\end{algorithmic}
\end{algorithm}

\begin{algorithm}[h!]
\caption{Covariate-assisted spectral estimation of grade of membership models via the shared simplex geometry.}
\label{alg:gom}
\begin{algorithmic}[1]
    \Require Given response matrix $\boldR$, covariate matrix $\boldX$, number of extreme profiles $K$, threshold $\epsilon$.
    \State Input $\boldR\boldR^\top$, $\boldX\boldX^\top$ and $K$ into Algorithm \ref{alg:HeteroPCA} to obtain $\boldG_R$ and $\boldG_X$.
    \State Obtain the top $K$ singular value decomposition of $\boldG=\boldG_R+\alpha\boldG_X$ as $\hat{\boldU}\hat{\mathbf{\Lambda}} \hat{\mathbf{V}}$.
    \State Let $\mathbf{Y}=\hat{\mathbf{U}}$.
    \For{$k \in [K]$}
        \State $\hat{S}_k = \arg\max(\{\|\mathbf{Y}_{i,:}\|_2 :i\in [N] \}) $.
        \State $\boldsymbol{u} = \mathbf{Y}_{\hat{S}_{k,:}}/ \|\mathbf{Y}_{\hat{S}_{k,:}}\|_2 $.
        \State $\mathbf{Y} = \mathbf{Y}\left(\mathbf{I}_K-\boldsymbol{u} \boldsymbol{u}^{\top}\right) $.
    \EndFor
    \State $\tilde{\boldPi} = \hat{\boldU} \left(\hat{\boldU}_{\hat{S},:}\right)^{-1} $.
    \State Project each row of $\tilde{\boldPi}$ to $\Delta^K$ to get $\hat{\boldPi}$.
    \State Compute $\hat{\boldSigma} = \hat{\mathbf{\Lambda}}^{1/2}$, and $\hat{\mathbf{V}} = \mathbf{L}^\top \hat{\boldU}\hat{\boldSigma}^{-1}$.
    \State $\tilde{\boldTheta} = \left[\hat{\mathbf{V}} \hat{\boldSigma} \hat{\boldU}^{\top} \hat{\boldPi}\left(\hat{\boldPi}^{\top}\hat{\boldPi}\right)^{-1}\right]_{1:J,:}$, and $\hat{\mathbf{M}} = \frac{1}{\sqrt{\alpha}} \left[\hat{\mathbf{V}} \hat{\boldSigma} \hat{\boldU}^{\top} \hat{\boldPi}\left(\hat{\boldPi}^{\top}\hat{\boldPi}\right)^{-1}\right]_{(J+1):(J+W),:}$.
    \State Truncate each entry of $\hat{\boldTheta}$ by $\hat{\theta}_{jk} = \min\{\max(\Tilde{\theta}_{jk}, \varepsilon), 1-\varepsilon\}$.
\end{algorithmic}
\end{algorithm}

\subsection{Selection of the balance parameter $\alpha$}
\label{sec:chooseparam}

Our goal is accurate recovery of the shared latent structure encoded by $\boldsymbol\Pi$, but $\boldsymbol\Pi$ is unobserved. We therefore select the balance parameter $\alpha$ using a predictive criterion on held-out response entries, exploiting the principle that better recovery of the low-rank signal leads to improved out-of-sample prediction in matrix estimation problems.
Predicting unseen entries in $\boldR$ is closely related to matrix recovery with missing data \citep{Chen2021}. For a low-rank matrix $\mathcal{B}$ with observed matrix with missing $\mathbf{B}$, if each entry is observed randomly with probability $p$, we can construct an unbiased estimate of $\mathcal{B}$ by inverse probability weighting
\begin{equation*}
    \tilde{B}_{ij} = \left\{
    \begin{aligned}
        &\frac{1}{p} B_{ij}, \text{ if } (i,j) \text{ is observed,}\\
        &0, \text{ otherwise.}
    \end{aligned}
    \right.
\end{equation*}
Then, we can recover the missing entries of $\mathcal{B}$ by performing spectral methods on $\tilde{\mathbf{B}}$, with sharp theoretical guarantee \citep{Yan2021}.

Following the idea of matrix recovery with missing data, we propose a cross-validation based algorithm to select the appropriate balance parameter. First, we randomly partition the entries in $\boldR$ into $k$ disjoint subsets $\mathcal{A}_1^{R},\ldots , \mathcal{A}_k^R$. Similarly, we partition the entries in $\boldX$ into $\mathcal{A}_1^X,\ldots, \mathcal{A}_k^X$. Then, let $\mathcal{A}_i=\mathcal{A}_i^R\cup \mathcal{A}_i^X$ for all $i \in [k]$, where $\mathcal{A}_i \cap \mathcal{A}_j=\emptyset$ if $i\neq j$. 
For each set, we mask the entries in the set and apply Algorithm \ref{alg:gom} to predict the entries in the set in $\boldR$. 
We use the mean absolute error as the metric for the prediction error for each fold. The mean absolute error for $\alpha_i$ in fold $j$ is
\begin{equation}
    \label{eqn:mae_pred}
    \text{MAE}_{i,j} = \sum_{(a,b)\in \mathcal{A}_j}\frac{|\boldR_{a,b}-\hat{\boldR}_{a,b}|}{|\mathcal{A}_j|}.
\end{equation}
Then, we select the balance parameter that gives the smallest cross-validated prediction error of the unseen entries in the matrix.
\begin{equation}
    \label{eqn:choose_alpha}
    \alpha = \arg\min_{\alpha_i}\sum_{j=1}^k\text{MAE}_{i,j}.
\end{equation}
A similar idea of cross-validated prediction error was used in
\citet{McGrath2024LEARNER}. 
Algorithm \ref{alg:cv} describes the full details of our procedure.

We conduct extensive simulation studies to assess our proposed Algorithm \ref{alg:cv}. Figure \ref{fig:tuning illus} presents the simulation results. The two simulations are conducted under $N=500$, $J=50$, $W=10$, $K=3$. Each row of $\boldPi$ is independently generated from $\text{Dirichlet}(1,1,1)$. Each entry of $\boldTheta$ is independently generated from $\text{unif}[0,1]$, and each entry of $\boldM$ follows $N(0,1)$. Additional independent $N(0,1)$ noise is added to $\cX$ for $\boldX$. 
Plots of the same column come from the same simulation, where the first row is the cross-validated prediction error and the second row is the mean absolute error between the oracle $\boldPi$ and the estimated $\hat{\boldPi}$ via Algorithm \ref{alg:gom}. The optimal $\alpha$ should minimize the mean absolute error between the estimated $\hat{\boldPi}$ and the oracle $\boldPi$. Figure \ref{fig:tuning illus} shows that our algorithm successfully selects $\alpha$ that is close to the optimal $\alpha$.
The similar performance is observed across all the simulations, and demonstrates the validity of Algorithm \ref{alg:cv}.

\begin{algorithm}[h!]
\caption{Selection procedure for the balance parameter $\alpha$.}
\label{alg:cv}
\begin{algorithmic}[1]
    \Require Given main response data $\boldR$, covariates $\boldX$, rank $K$, a sequence of $\bm{\alpha}=(\alpha_1, \dots, \alpha_l)$, fold $k$.
    \State Randomly partition each entry of $\boldR$ into $k$ disjoint subsets $\mathcal{A}_1^R,\ldots ,\mathcal{A}_k^R$, and partition $\boldX$ to $k$ disjoint subsets $\mathcal{A}_1^X,\ldots, \mathcal{A}_k^X$.
    \State Let $\mathcal{A}_i=\mathcal{A}_i^R \cup\mathcal{A}_i^X$. Note that $\mathcal{A}_i \cap \mathcal{A}_j=\emptyset$ for $i\neq j$, and $\cup_{i=1}^k \mathcal{A}_i = [N]\times [J+W]$.
    \For{$i=1$ to $p$}
        \For{$j=1$ to $k$}
            \State $\boldR^{(j)}=\boldR$. Let $\boldR^{(j)}_{m,n}=0$ for all $(m,n)\in \mathcal{A}_{j}$. Similarly define $\boldX^{(j)}$.
            \State Perform Algorithm \ref{alg:gom} with $\boldR^{(j)}$, $\boldX^{(j)}$, $\alpha_i$, and obtain the prediction $\hat{\boldR}^{(j)}=\hat{\boldPi}^{(j)}\hat{\boldTheta}^{(j)\top}$.
            \State Calculate $\text{MAE}_{i,j}$ with (\ref{eqn:mae_pred}).
        \EndFor
    \EndFor
    \State Choose $\alpha$ according to (\ref{eqn:choose_alpha}).
\end{algorithmic}
\end{algorithm}

\begin{figure}[h!]
    \centering
    \includegraphics[width=0.65\textwidth]{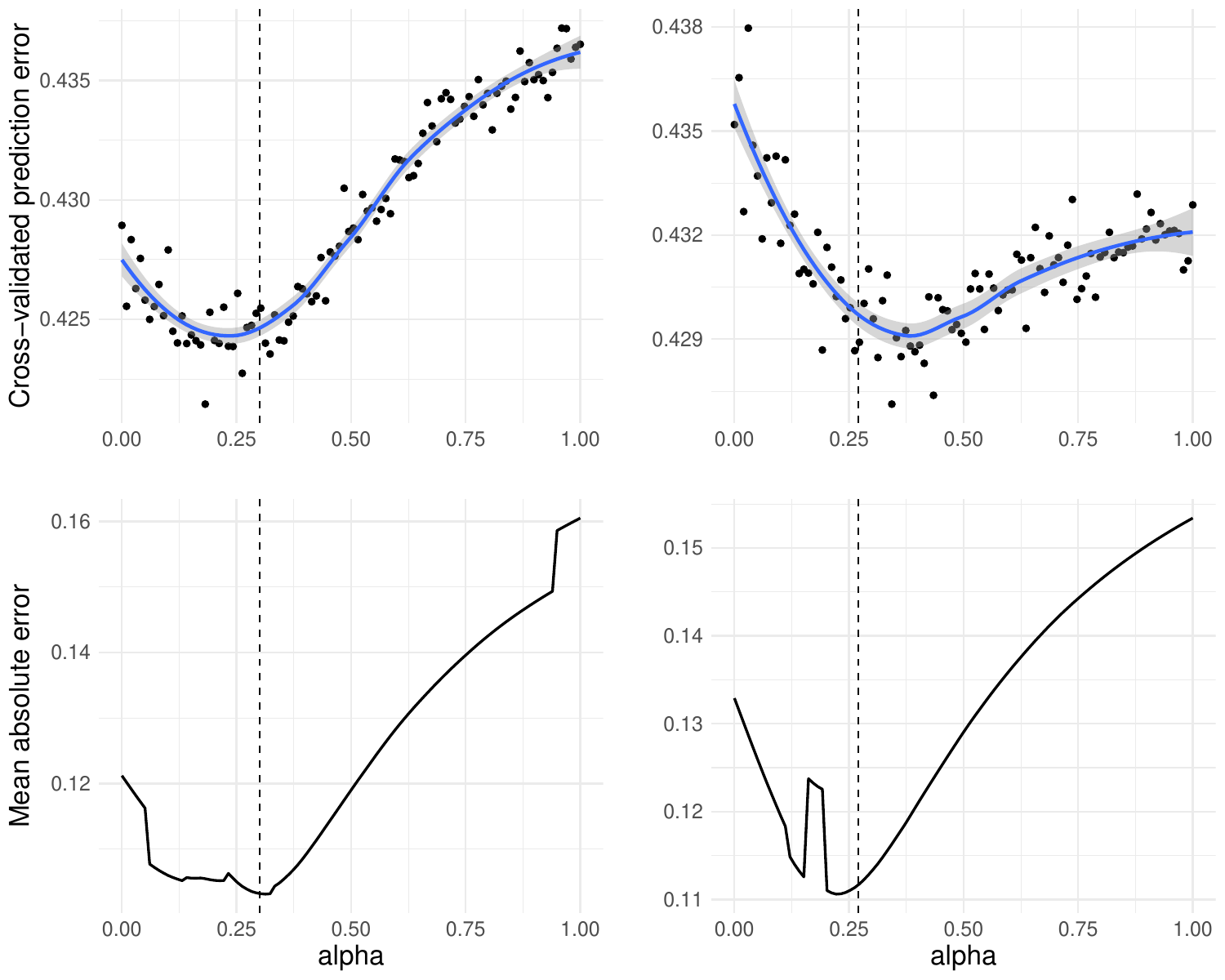}
    \caption{The $x$ axis shows different values of $\alpha$ and the $y$ axis the prediction error (\ref{eqn:mae_pred}) for the first row, and the mean absolute error between $\hat{\boldPi}$ and $\boldPi$ for the second row.}
    \label{fig:tuning illus}
\end{figure}

To determine the range of the balance parameter $(\alpha_1,\ldots,\alpha_l)$ as input to Algorithm \ref{fig:tuning illus}, we follow a similar spirit in \citet{Binkiewicz2017} and define
\begin{align}
    \label{eqn:alphamin alphamax}
    \alpha_1=\alpha_{\min}&=\frac{\lambda_K(\boldG_R)-\lambda_{K+1}(\boldG_R)}{\lambda_1(\boldG_X)},\quad
    \alpha_l=\alpha_{\max}=\frac{\lambda_1(\boldG_R)}{\lambda_K(\boldG_X)-\lambda_{K+1}(\boldG_X)}.
\end{align}
The rationale for defining this range comes from how the balance parameter $\alpha$ affects the singular value decomposition of $\boldG_R+\alpha \boldG_X$. To provide more intuition, let us first consider two general matrices $\boldG_R$ and $\boldG_X$ without any structural assumptions. The increase of $\alpha$ changes the magnitude of singular values of $\boldG_R+\alpha \boldG_X$, and can cause the $K$th and $(K+1)$th singular values to cross. This reordering leads to a sharp change in the singular vectors, which leads to the non-smooth transitions in the estimation result as observed in the second row of Figure \ref{fig:tuning illus}. By applying Weyl's inequality \citep{Chen2021}
, we can characterize the range of $\alpha$ that enables such exchange to take place. Readers can refer to \citet{Binkiewicz2017} for further technical details. 
Since the analysis in \citet{Binkiewicz2017} applies to general matrices without any structural assumptions, this interval is a sound starting point for selecting $\alpha$ under our problem. We recommend using it as an initial guide, and then refining the search for $\alpha$ based on the cross-validated prediction error.

\section{Theoretical guarantees for the proposed method}
\label{sec:theory}
\subsection{Identifiability}
\label{sec:identifiability}
Identifiability is vital to statistical models, especially for latent variable models, because it is the prerequisite for reliable statistical inference and meaningful interpretation of latent structures. 
To study identifiability, we 
investigate under what conditions, the matrix decomposition at the population level gives unique $\boldPi$ and $\boldTheta$.
The formal definition is as follows.
\begin{definition}
\label{def-id}
    A covariate-assisted grade of membership model with parameter set $(\boldPi, \cD, \boldTheta, \boldM)$ is identifiable, if any other set of parameters $(\tilde{\boldPi}, \tilde{\cD}, \tilde{\boldTheta}, \tilde{\boldM})$, $\boldPi\cD\boldPi^\top = \tilde{\boldPi}\tilde{\cD}\tilde{\boldPi}^\top$,  $\boldPi\boldTheta^{\top} = \tilde{\boldPi} \tilde{\boldTheta}^{\top}$, and $\boldPi\boldM^\top = \tilde{\boldPi} \tilde{\boldM}^\top$ holds if and only if $(\boldPi,\cD,\boldTheta, \boldM)$ and $(\tilde{\boldPi}, \tilde{\cD}, \tilde{\boldTheta}, \tilde{\boldM})$ are identical up to a permutation of $K$ extreme profiles.
\end{definition}

We next present the sufficient conditions for the expectation identifiability of covariate-assisted grade of membership models. First, we show $(\boldPi,\cD)$ is identifiable under mild conditions.
The theorem below shows that covariate assistance relaxes the rank requirements needed for identifiability, relative to the covariate-free GoM model in \cite{chen2024spectral}.

\begin{theorem}[Identifiability of $\boldPi$ and $\mathcal{D}$]
    \label{theorem:identifiability}
    Suppose $\boldPi$ satisfies Assumption \ref{assump:extreme}. For notational simplicity, denote $\cD=\boldTheta^{\top} \boldTheta+\alpha \boldM^{\top}\boldM$.
    \begin{itemize}
        \item [(a)] If $rank(\cD)=K$, then the covariate-assisted grade of membership model is identifiable.
        \item [(b)] If $rank(\cD)=K-1$ and no column of the top $K-1$ rows of $\cD$ is an affine combination of the other columns of the top $K-1$ rows of $\cD$. Since $\mathcal{D}$ is symmetric, the same reasoning applies to the transpose. Then, the covariate-assisted grade of membership model is identifiable. (An affine combination of vectors $\boldsymbol{x}_1, \dots, \boldsymbol{x}_n$ is defined as $\sum_{i=1}^na_i\boldsymbol{x}_i $ with $\sum_{i=1}^na_i=1 $.)
        \item [(c)] In any other case, if there exists a subject $i$ such that $\pi_{ik}>0$ for every $k=1,\dots,K$, then the covariate-assisted grade of membership model is not identifiable.
    \end{itemize}
\end{theorem}

We give several examples to further illustrate the implication of Theorem \ref{theorem:identifiability} in Appendix \ref{append:identifiability}. Then, we provide a sufficient condition for the identification of $\boldTheta$ and $\boldM$.

\begin{proposition}[Identifiability of $\boldTheta$ and $\boldM$]
    \label{prop:identifiability}
    If
    $\boldPi$ is already identified, 
    if the columns of $\bigl[
        \boldTheta^\top \mid  \boldM^\top
    \bigr]^\top$ are affine independent, $\boldTheta$ and $\boldM$ are identifiable.
\end{proposition}

In comparison, in the covariate-free grade of membership model, identifiability indicates the identification of $\boldTheta$ and $\boldPi$. \citet{chen2024generalized} shows that under Assumption \ref{assump:extreme} and the columns of $\boldTheta$ is affine independent, the covariate-free model is identifiable. Although simultaneous identification of $\boldTheta$ and $\boldM$ seemingly should require more stringent conditions in the considered covariate-assisted case, Proposition \ref{prop:identifiability} actually requires weaker conditions. Under the condition that columns of $\boldTheta$ are affine independent, the columns of $\bigl[\boldTheta^\top \mid \boldM^\top\bigr]^\top$ are affine independent. Thus, Proposition \ref{prop:identifiability} holds under the identifiability conditions of the covariate-free model. Similarly, if the columns of $\boldM$ are affine independent, conditions in Proposition \ref{prop:identifiability} are satisfied. Even if neither $\boldTheta$ nor $\boldM$ has affine independent columns, conditions in Proposition \ref{prop:identifiability} can still be satisfied with some specific $\boldTheta$ and $\boldM$.

\subsection{Finite-sample estimation error bounds}
\label{sec:guarantee}
This section presents the theoretical guarantees of our method. We establish finite-sample error bound for all the parameters in the covariate-assisted grade of membership models. First, we need to adopt several assumptions in singular subspace estimation \citep{Yan2021}.

\begin{assumption}[Heteroskedastic noise]
    \label{assump:noise}
    The heteroskedastic noise $\boldE^X \in \mathbb{R}^{N\times W}$ satisfies that
    \begin{enumerate}
        \item  ${E}\bigl[E_{i,j}^X\bigr]=0$, and $\text{var}\bigl(E_{i,j}^X\bigr)=\sigma_{i,j}(\boldX)^2\leq \sigma(\boldX)^2$, for all $\forall (i,j)\in [N]\times [W]$.
        \item $\left\|\mathbf{E}^X\right\|_{\infty} \lesssim B_X$, where $B_X:={\sigma(\boldX)\min\{\sqrt{W},\sqrt[4]{NW}\}}/{\sqrt{\log d_X}}$, and $d_X = \max\{N,W\}$.
    \end{enumerate}
\end{assumption}

In covariate-assisted grade of membership models, there are two sources of noise: Bernoulli noise in the response matrix 
, and heteroskedastic noise from the covariates. Bernoulli noise is inherently heteroskedastic and satisfies the conditions in Assumption \ref{assump:noise} \citep{chen2024generalized}. We could similarly define the above quantities for $\boldE^R$ for notational simplicity.
Since the covariates are usually centered and scaled, the assumption on the covariates is mild.

To establish entry-wise guarantee, incoherence assumptions are frequently adopted \citep{candes2010power, candes2012exact, Chen2021}. For the covariate-assisted case, we generalize the incoherence assumption to the following joint incoherence assumption.

\begin{assumption}[Joint incoherence]
    \label{assump:incoherence}
    The response matrix $\cR$ and the covariate matrix $\cX$ are jointly $\mu$-incoherent if,
    \begin{align*}
        \max\left\{\|\boldU_R\|_{2,\infty},\|\boldU_X\|_{2,\infty} \right\} \leq \sqrt{\frac{\mu K}{N}}&, 
        \quad \|\mathbf{V}_R\|_{2,\infty}\leq \sqrt{\frac{\mu K}{J}}, \quad \|\mathbf{V}_X\|_{2,\infty}\leq \sqrt{\frac{\mu K}{W}},\\
        \|\cR\|_{2,\infty} \leq \sqrt{\frac{\mu}{NJ}}\|\cR\|_F&, \quad \|\cX\|_{2,\infty}\leq \sqrt{\frac{\mu}{NW}}\|\cX\|_F,
    \end{align*}
    where $\boldU_R$, $\mathbf{V}_R$ are the left and right singular vectors of $\cR$, respectively, and $\boldU_X$ and $\mathbf{V}_X$ the left and right singular vectors of $\cX$, respectively.
\end{assumption}

The following assumption on the signal strength follows from the assumptions in \citet[Theorem 10]{Yan2021}.

\begin{assumption}[Signal strength and rank]
    \label{assump:information}
    \begin{align*}
        N \gtrsim \kappa_{\max}^2\mu K + \mu^2K\log^2d&, \quad J \gtrsim K\log^4 d_R ,\quad W \gtrsim K\log^4 d_X,\\
        \zetaop(\boldR) \ll \frac{\sigma_K(\cR)^2}{\kappa(\cR)^2}&,\quad \zetaop(\boldX)\ll \frac{\sigma_K(\cX)^2}{\kappa(\cX)^2},
    \end{align*}
    where $\kappa_{\max}^2=\kappa(\cR)^2+\kappa(\cX)^2$, $d_R=\max\{N,J\}$, $d_X=\max\{N,W\}$, $d=\max\{d_R,d_X\}$, and
     \begin{align*}
        \zeta_{\text{op}}(\boldR) &= \sigma(\boldR)^2\sqrt{NJ}\log d_R+\sigma(\boldR)\sigma_1(\cR)\sqrt{N\log d_R},\\
        \zeta_{\text{op}}(\boldX) &= \sigma(\boldX)^2\sqrt{NW}\log d_X+\sigma(\boldX)\sigma_1(\cX)\sqrt{N\log d_X}.
    \end{align*}
\end{assumption}

The terms $\zetaop(\boldR)$ and $\zetaop(\boldX)$ respectively captures the noise in the two regimes. The first quadratic term $\sigma(\boldR)^2\sqrt{NJ}\log d_R$ captures the large noise regime, where $\sigma(\boldR)/\sigma_1(\cR)\geq 1/\sqrt{J\log d_R}$. On the other hand, when the noise is small, the linear term $\sigma(\boldR)\sigma_1(\cR)\sqrt{N\log d_R}$ becomes dominant. Following the decomposition in (\ref{eqn:decom}), the quadratic term comes from the $\boldE^R\boldE^{R\top}$, and the linear term comes from $\boldE^R\cR^\top + \cR\boldE^{R\top}$. The same reasons apply to $\boldX$ as well.

We defer the theoretical result for the singular subspaces to Appendix \ref{sec:theory_subspace}. From Definition \ref{def-id}, we can only estimate $\boldPi$ up to column permutation. 
Therefore, we introduce a column permutation matrix $\mathbf{P}$ to establish the error bound.

\begin{theorem}
    \label{theorem:mmmbound}
    Under Assumptions \ref{assump:noise}, \ref{assump:incoherence}, \ref{assump:information}, $\sigma(\boldX)\sqrt{W\log d_X} \asymp \sigma(\boldR)\sqrt{J\log d_R}$, and  $\kappa^2(\boldPi) \sqrt{K} \sigma_1(\boldPi) \epsilon \lesssim 1$,
    there exists a $K\times K$ permutation matrix $\mathbf{P}$, such that with probability exceeding $1-O(d^{-10})$
    \begin{equation}
        \label{eqn:pi_bound}
        \|\hat{\boldPi}-\boldPi \mathbf{P}\|_{2,\infty} \lesssim \sigma_1(\boldPi)\kappa(\boldPi)^2 \epsilon,
    \end{equation}
    where
        $\epsilon=(1+\alpha+\alpha^2){\kappa_{\max}^2} \sqrt{{\mu K}}\zetaop(\boldG)/({\sigma_K(\cG)}\sqrt{N})$.
\end{theorem}

The assumption $\sigma(\boldX)\sqrt{W\log d_X}\asymp \sigma(\boldR) \sqrt{J\log d_R}$ comes from the theoretical difficulty that $\alpha$ and the signal-to-noise ratio of $\boldR$ and $\boldX$ do not omit an explicit relationship. However, this assumption is mild in practice. Covariate-assistance is most powerful when $\boldR$ and $\boldX$ have comparable contributions. With a good tuning procedure, if the signal of $\boldR$ is much stronger than that of $\boldX$, $\alpha$ would be close to 0, and if the signal of $\boldX$ is much stronger than that of $\boldR$, $\alpha$ would be large. The similar discovery is made by \citet{Ma2024Optimal} in joint singular subspace estimation:
when two matrices shares the same left singular vectors, to jointly achieve optimal estimation, they need the number of columns of every matrix and their noise to be comparable. The additional assumption to establish error bound on the mixed membership matrix is frequently adopted in previous research \citep{Mao2019, chen2024spectral, chen2024generalized}. 

To give better intuition into Theorem \ref{theorem:mmmbound}, we next explain it under a \textit{strong signal regime}, i.e. $\kappa_{\max}, \mu,\sigma(\boldR), \sigma(\boldX) \asymp 1$ and $\sigma_K(\boldPi)\asymp \sqrt{N}$, $\sigma_K(\boldTheta)\asymp \sqrt{J}$, and $\sigma_K(\boldM)\asymp \sqrt{W}$. Under this regime, each individual's mixed membership score uniformly converges at the rate of $\tilde{O}(1/\sqrt{J+W})$. This indicates that as long as one of $J$ or $W$ goes to infinity, we will have an entry-wise consistent estimator of $\boldPi$. Compared with the error bound in the covariate-free case \citep{chen2024generalized}, this is a faster rate with an additional $W$ in the denominator. 

Furthermore, thanks to HeteroPCA,
our estimator is more adaptable to high-dimensional scenarios with $J, W\gg N$ compared to vanilla SVD-based estimators \citep{chen2024generalized}. Analyzing the singular value decomposition requires a stronger assumption than Assumption \ref{assump:information}. Assuming the condition number is 1, common assumptions in singular value decomposition require $\sigma(\boldR)/\sigma_K(\cR) = \tilde{O}(1/\sqrt{J})$ and similarly for $\boldX$. When $J$ or $W$ is large, the condition is difficult to satisfy. Instead, when $J$ or $W$ is large, Assumption \ref{assump:information} only requires $\sigma(\boldR)/\sigma_K(\cR) = \tilde{O}(1/\sqrt[4]{NJ})$ \citep{caiSubspaceEstimationUnbalanced2020}, which is significantly easier to satisfy. 

Next, we elaborate on the role that $\alpha$ plays in the upper bound. First, Equation~(\ref{eqn:pi_bound}) omits the following upper bound, where we apply Lemma~\ref{lemma:keigenvaluebound},
    $\|\hat{\boldPi}-\boldPi \mathbf{P}\|_{2,\infty} \lesssim f(\alpha)\kappa_{\max}^2 \sqrt{{\mu K}/{N}}$,
where
\begin{equation*}
        f(\alpha) = \frac{(1+\alpha+\alpha^2)\left(\zetaop(\boldR)+\alpha\zetaop(\boldX)\right)}{\sigma_K(\cR)^2+\alpha\sigma_K(\cX)^2}.
\end{equation*}
Through this definition, we construct a sharp entry-wise upper bound of $\hat{\boldPi}$, and $f(\alpha)$ contains all the terms concerning $\alpha$. By investigating the property of $f(\alpha)$, we have the following corollary.
\begin{corollary}
    \label{coro:bound w.r.t. alpha}
    \begin{enumerate}
        \item[1.] There exists $\alpha_0\in (0,1)$ such that $f(\alpha_0) < f(1)$.
        \item[2.] If  
        $\left(\sigma_K(\cR)^2-\sigma_K(\cX)^2\right)\zetaop(\boldR)+\sigma_K(\cR)^2\zetaop(\boldX) < 0,$
    where $\zetaop(\boldR)$ and $\zetaop(\boldX)$ are defined in Assumption \ref{assump:information}, then there exists $\alpha_0>0$ such that $f(\alpha_0)>f(0)$.
    \end{enumerate}
\end{corollary}

In particular, Corollary 1 implies that neither fixing $\alpha=1$ (equivalent to unweighted stacking) nor fixing $\alpha=0$ (discarding the covariates) is uniformly optimal. The optimal integration of covariates depends on their relative signal-to-noise ratio, and improper weighting can strictly worsen estimation compared with the covariate-free baseline.

Corollary \ref{coro:bound w.r.t. alpha} has two main implications. First, it highlights the necessity of a data-driven procedure to select the balance parameter instead of fixing $\alpha$ to be some constant. Directly applying SVD on the concatenated matrix $[\boldR\mid \boldX]$ without introducing the balance parameter $\alpha$ will lead to suboptimal performance. Second, we can consider $\sigma_K(\cR)$ and $\sigma_K(\cX)$ as the signal strength, and $\zetaop(\boldR)$ and $\zetaop(\boldX)$ as noise. Then, the second part in Corollary \ref{coro:bound w.r.t. alpha} implies if $\boldX$ is too noisy, that is $\zetaop(\boldX)/(\sigma_K(\cR)^2-\sigma_K(\cX)^2)$ is larger than $\zetaop(\boldR)/\sigma_K(\cR)^2$, we should not incorporate $\boldX$ into the model-based analysis. It is intuitive that if $\boldX$ is noisy, $\boldX$ contains little information about the latent structure. Incorporating $\boldX$ would make the estimation of the latent structure more difficult.

For the entry-wise error bound on the item parameters in the binary responses part and the parameters in the covariates part, we have the following theorem.

\begin{theorem}
    \label{theorem:Theta}
    Instate the assumptions of Theorem \ref{theorem:mmmbound} and $\left\|\mathbf{V}\right\|_{2,\infty}\lesssim \sqrt{\mu K/(J+W)}$, with probability at least $1-O(d^{-10})$, we have,
    \begin{align*}
        \max\left\{\left\| \hat{\boldTheta}\mathbf{P}^\top - \boldTheta \right\|_{\infty}, \sqrt{\alpha}\left\| \hat{\boldM}\mathbf{P}^\top - \boldM \right\|_{\infty}\right\} &\lesssim (1+\alpha+\alpha^2) \kappa(\boldPi)^2\kappa_{\max}^2\left( \mathcal{E}_U^{\text{ip}} + \mathcal{E}_E^{\text{ip}} \right),
    \end{align*}
    where 
    \begin{align*}
        \mathcal{E}_U^{\text{ip}} &=  \frac{\kappa_{\max}^2}{\sigma_K(\cG)^{1/2}} \frac{\mu K^{3/2}}{\sqrt{N(J+W)}} \left( \frac{\zetaop(\boldG)^2}{\sigma_K(\cG)} + \sqrt{\frac{\mu K}{N}}\zetaop(\boldG) \right),\\
        \mathcal{E}_{E}^{\text{ip}} &= \frac{\mu K}{N}\left(\frac{\zetaop(\boldR)}{\sigma_K(\cR)} + \frac{\sqrt{\alpha}\zetaop(\boldX)}{\sigma_K(\cX)}\right).
    \end{align*}
\end{theorem}

The estimation error comes from two sources. In HeteroPCA, the expectation of the matrix is estimated through projection onto the estimated left singular subspace. Compared with the direct singular value decomposition approach, for the estimation of item parameters $\mathbf \Theta$, it results in loss of efficiency. $\mathcal{E}_E^{\text{ip}}$ stems from the inherent noise in the responses and covariates, which also exists in singular value decomposition. But the estimation error of the left singular vectors accumulates into $\mathcal{E}_{U}^{\text{ip}}$. This additionally causes the error bound to be not optimal with respect to $K$. 

Nevertheless, under the \textit{strong signal regime} defined in the discussion of Theorem~\ref{theorem:mmmbound}, the entry-wise estimation errors for $\boldTheta$ and $\boldM$ converge at a rate of $\tilde{O}(1/\sqrt{N})$, dominated by the term $\mathcal{E}_{E}^{\text{ip}}$. This rate is the same as the singular value decomposition based analysis \citep{chen2024generalized}. This indicates if $K$ is a constant, which prevails in most applications, the estimation error of the item parameters by heteroskedastic principal component analysis is close to the one by singular value decomposition. It is important to note that it is achieved with weaker assumptions as discussed before. If one is willing to make stronger assumptions for singular value decomposition to work, they can substitute the current right singular subspace estimate by the one estimated by singular value decomposition. In this way, the error bound of item parameters becomes optimal with respect to $K$. The proof is similar and simpler than the proof of Theorem \ref{theorem:Theta} and thus omitted.

\subsection{Discussion of assumptions}
\label{sec:common scenario}
In this section, we highlight that our assumptions in Section \ref{sec:theory} are mild and easy to satisfy in commonly applied settings. In the literature on Bayesian estimation of the grade of membership model \citep{Erosheva2007Describing, Gu2023DimensionGrouped}, 
it is generally assumed that each row of $\boldPi$ is independently generated from $\text{Dirichlet}(\bm{\beta})$, and each entry of $\boldTheta$ is independently generated from $\text{Beta}(a,b)$. Since covariates can always be standardized, we assume that each entry in $\boldM$ is generated identically and independently from a $N(\check{\mu},\check{\sigma}^2)$ truncated within $[-\xi,\xi]$. We define $\beta_{\max}=\max_k{\beta}_k$, $\beta_{\min}=\min_k\beta_k$, $\beta_0=\sum_{k=1}^K\beta_k$ and $\nu = \beta_0/\beta_{\min}$. The truncation here is to match the previous assumption that $\cX$ is bounded. One can easily extend our following result to normal distribution following our proof.

There have been discussions for the similar data generation mechanism in \citet{Mao2019, chen2024generalized}. The following assumption is similar to Assumption 3.3 in \citet{Mao2019} and Assumption 4 in \citet{chen2024generalized}.

\begin{assumption}
    \label{assump:high_prob_cons}
    There exists constant $C_1$, s.t. $\max_{k\in[K]}\beta_k \leq C_1\min_{k\in [K]}\beta_k $, $\max_{k\in [K]}\beta_k=O(1)$. In addition,
    \begin{equation}
        \label{eqn:high_prob_cons}
        \nu \leq \frac{\sqrt{N/\log N}}{6\sqrt{3}(1+\beta_0)}, \quad c_1 \geq 4K\sqrt{15\log J/J},\quad \tilde{c}_1\geq 4K\sqrt{15\log W/ W},
    \end{equation}
    where $c_1=ab/((a+b)^2(a+b+1))$, and $\tilde{c}_1=\text{var}(m_{jk})$.
\end{assumption}

Assumption \ref{assump:high_prob_cons} is a mild assumption. If $\bm{\beta}$ is evenly distributed, and $a$, $b$, $\mu$, $\sigma$, $\xi$ are constants, as long as $K$ grows slower than $\min\{\sqrt{J/\log J}, \sqrt{W/\log W}\}$, when $J$ and $W$ are large enough, this assumption holds.

\begin{corollary}
    \label{coro:incoherence bound}
    Under the data generation mechanism that $\bm{\pi}_i{\sim} \text{Dirichlet}(\bm{\beta})$, $\theta_{jk}{\sim} \text{Beta}(a,b)$, $m_{jk}{\sim} {N}(\check{\mu},\check{\sigma}^2)$ truncated in $[-\xi, \xi]$ independently, and Assumption \ref{assump:high_prob_cons}, if $\cD$ is full rank, $\cR$ and $\cX$ are jointly incoherent, $K,\xi\asymp 1$, with probability at least $1-O(K\min\{N,J,W\}^{-3})$, we have
        $\mu \lesssim 1,$
    and Assumption 
    where $\xi$ is the upper bound of $\|\cX\|_{\infty}$.
\end{corollary}

Therefore, under the assumptions in Corollary \ref{coro:incoherence bound}, all the assumptions and the \textit{strong signal regime} defined in Section \ref{sec:guarantee} are satisfied. The more detailed discussion is in Appendix \ref{sec:dis_incoherence}.

\section{Simulation studies}
\label{sec:simulation}

\subsection{Comparison with the joint maximum likelihood estimation}
\label{sec:compare}

In this section, we compare the performance of covariate-assisted grade of membership models with covariate-free cases estimated via the joint maximum likelihood and a spectral method, respectively. We consider two settings for the number of latent profiles: $K=3$ and $K=8$. Simulations for the sample size $N=200,500,1000$ are performed, with $J=N/10$, $W=N/20$ for each sample size. Each row of $\mathbf{\Pi}$ is generated from a Dirichlet distribution of all $1$s. Each entry in $\boldTheta$ is generated from a uniform distribution in $[0,1]$. Each entry of $\boldM$ is generated from a standard normal distribution. Each entry of $\boldE^X$ is generated from $N(0,0.5^2)$.

For the balance parameter $\alpha$, we apply Algorithm \ref{alg:cv} and fit a loess curve to promote automatic selection of this parameter. Each setting has 100 independent replicates. The maximum iterations for heteroskedastic principal component analysis is set as $10$. We choose the parameters as the default for the joint MLE algorithm in the package \texttt{sirt}.

The result is shown in Figure \ref{fig:simulation}. We compare their performance based on computation time and mean absolute estimation error of $\hat{\boldPi}$ and $\hat{\boldTheta}$. The first row in the plot shows the results for $K=3$, and the second for $K=8$. From the computation time, we see that even with an additional procedure to select the balance parameter $\alpha$, the computational complexity of our method is still significantly smaller than that of the joint maximum likelihood algorithm. Additionally, 
the efficiency of the joint maximum likelihood algorithm degrades more drastically when the size of the data matrix becomes larger.

The empirical estimation error of $\hat{\boldPi}$ aligns well with our theory. First, the covariate-assisted grade of membership model outperforms the other methods significantly. Further, as $N$, $J$ and $W$ grows, the estimation error of all methods decreases, and the one estimated by the proposes method decreases the fastest. For the estimation error of $\hat{\boldTheta}$, the theory in Section \ref{sec:guarantee} does not directly guarantee the estimate of the covariate-assisted method converges faster than the covariate-free ones. 
However, in practice, thanks to a better estimated $\hat{\boldPi}$, the covariate-assisted estimate provides a better estimate of $\boldTheta$.

\begin{figure}[h!]
    \centering
    \includegraphics[width=\textwidth]{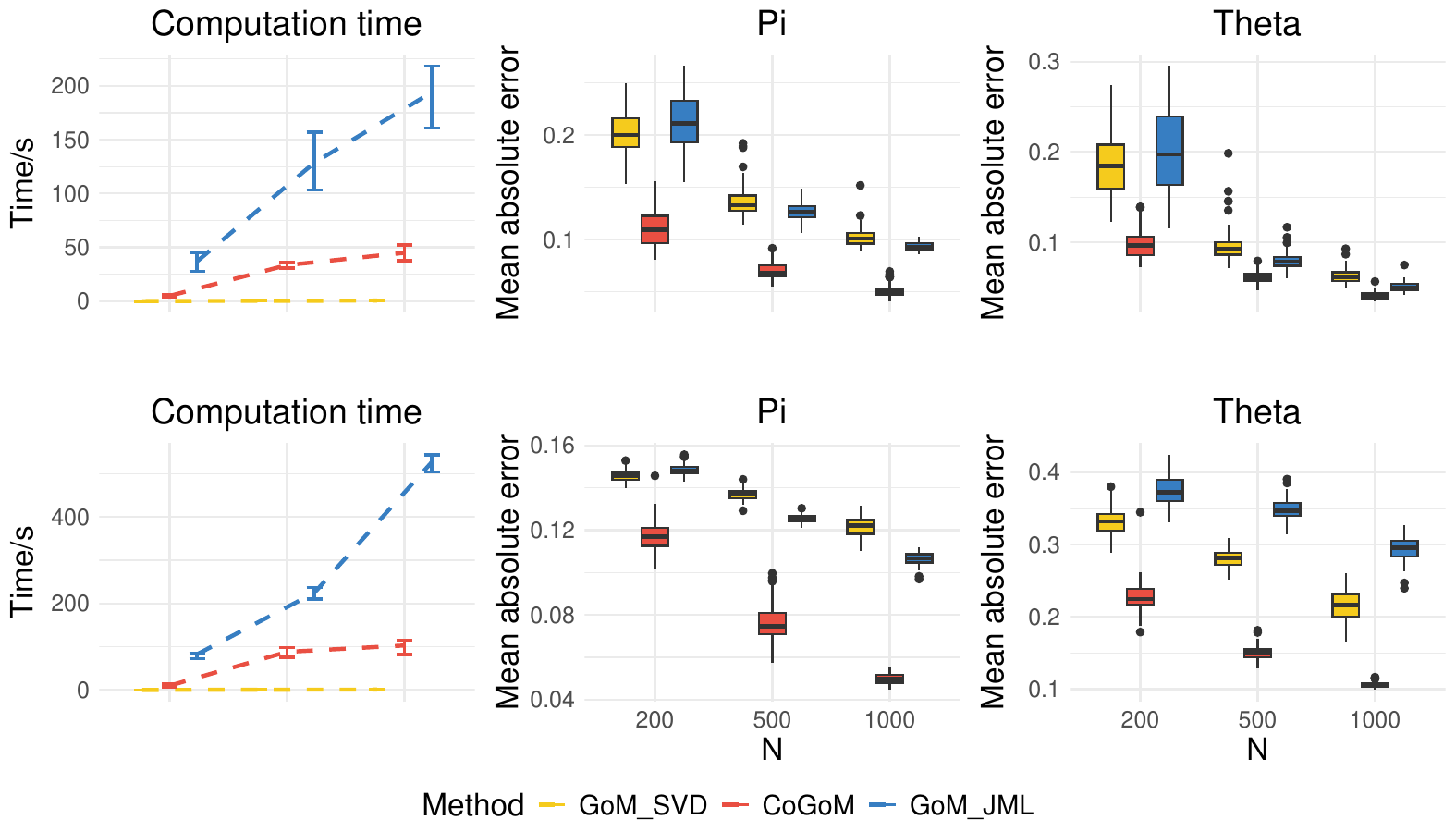}
    \caption{Comparison of different estimation methods. CoGoM represents the estimator of our proposed method. GoM\_JML represents the estimator computed by the joint maximum likelihood method. GoM\_SVD represents the grade of membership model estimated by spectral methods. The $x$ axis for each plot is the sample size $N$. Our proposed method is superior to JML in both computation efficiency and estimation accuracy.}
    \label{fig:simulation}
\end{figure}

\subsection{Why unweighted stacking is suboptimal and the necessity of selecting $\alpha$}
\label{sec:sim_tuning}
The two simulations in this section serve two purposes. First, we aim to illustrate the necessity of carefully selecting the balance parameter in practice. Second, we aim to demonstrate the effectiveness of Algorithm \ref{alg:cv}.
On a related note, \citet{Ma2024Optimal} proposed to perform SVD directly on the stacked matrix $[\boldR\mid \boldX]$ to obtain the estimate of the shared left singular vectors.
However, in the settings we consider, the simulation results show that the direct matrix stacking does not give the best performance. In some scenarios, performing SVD on $[\boldR \mid \boldX]$ can lead to even worse estimation of the signal eigenspace than a direct SVD on $\boldR$.

In the first simulation setting, we generate each row of $\boldPi$ from a Dirichlet distribution with parameters $(1,\dots,1)$, with the first $K$ rows set to be the identity matrix to satisfy the identifiability Assumption \ref{assump:extreme}. Each entry in $\boldTheta$ is independently generated from $\text{unif}[0,1]$. Each entry of $\boldM$ is independently generated from a standard normal distribution. Each entry of $\boldE^X$ is generated from $N(0, \sigma^2)$. The size of the matrices is $(N,J,W,K)=(500,50,25,3)$, and $\sigma$ is set to be $0.5,1,1.5,2$. For the second setting, $\boldPi$, $\boldTheta$ and $\boldM$ are generated the same as before. We set $\sigma=1$ and $(N,J,K)=(500,50,3)$. $J/W$ varies for each simulation to be $1, 2, 5, 10$.

Figure \ref{fig:com_stacksvd} shows that our method achieves higher parameter estimation accuracy by introducing and appropriately selecting $\alpha$. The result also coincides with Corollary \ref{coro:bound w.r.t. alpha} that the optimal value of $\alpha$ is determined by the signal-to-noise ratio between the response matrix $\mathbf R$ and the covariate matrix $\mathbf X$. When the noise of the covariates is strong, a smaller $\alpha$ is preferred. Therefore, when $\sigma$ and $J/W$ grows, the balance parameter helps partly offset the negative effect of the larger noise.

\begin{figure}[h!]
    \centering
    \includegraphics[width=0.8\textwidth]{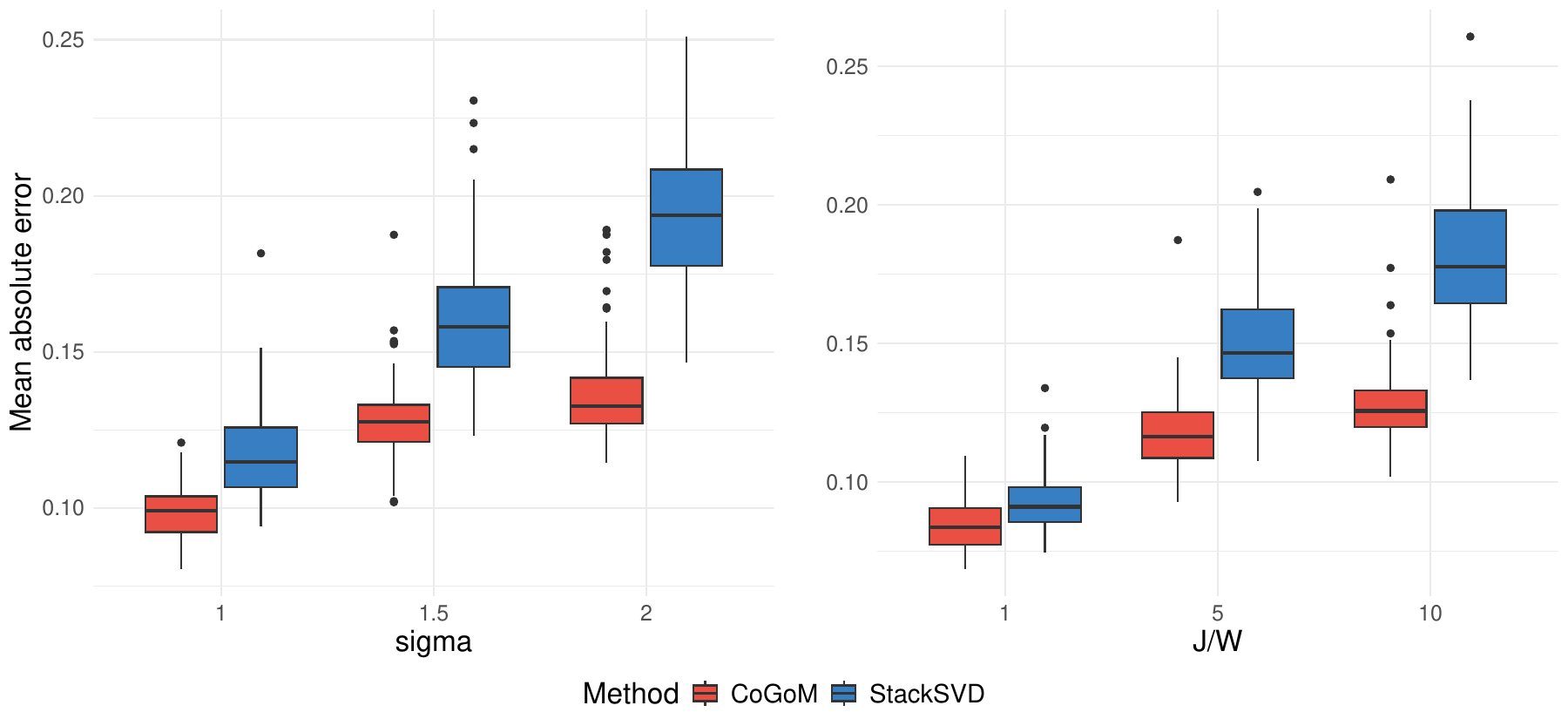}
    \caption{Boxplots of the estimation accuracy of singular value decomposition on the stacked matrix and our proposed method in mean absolute error. The $y$-axis is the mean absolute error between $\hat{\boldPi}$ and $\boldPi$. The $x$-axis for the figure on the left is different variance of noise, and for the figure on the right different ratios of the number of columns of the response matrix and the covariate matrix.}
    \label{fig:com_stacksvd}
\end{figure}

\subsection{Superiority of heteroskedastic principal component analysis}
\label{sec:sim_hetero}
In this section, we present simulation studies in the challenging scenarios where serious heteroscedasticity is present. We compare the performance in estimating the mixed-membership matrix of heteroskedastic principal component analysis and vanilla singular value decomposition.

Our setting is adapted from \citet{Zhang2021}. Let
    $v_1,\dots ,v_N\stackrel{i.i.d.}{\sim} \text{Unif}[0,1]$ and $\sigma_k^2 = {(N v^\beta)}/{(\sum_{i=1}^N v_i^\beta)},$ for  $k=1,\dots ,p$,
where $\sigma_k^2$ is the variance of the noise in each entry in the $k$th row of $\boldX$. In this setting, $\sum_{i=1}^N\sigma_i^2=N$, and the larger $\beta$ is, the more heteroskedastic the noise matrix will be. The noise becomes homoskedastic if $\beta=0$. Since the multiple binary response matrix is inherently heteroskedastic, it is generated in the same way as in Section \ref{sec:compare}. We select $\alpha$ according to Section~\ref{sec:chooseparam} with $l=20$, and $(\alpha_1,\ldots,\alpha_{20})$ equally spaced on $[0,1]$.

The simulation is conducted under $(N,J,W,K)=(100, 1000, 500, 4)$. The result is presented in Figure \ref{fig:HeteroPCA}. One can see that when heteroscedasticity is present, the result for heteroskedastic principal component analysis is significantly better than vanilla singular value decomposition. As the heteroscedasticity becomes more severe, their difference becomes more drastic. The variation in the performance is also smaller for the heteroskedastic principal component analysis.

\begin{figure}[h!]
    \centering
    \includegraphics[width=0.5\textwidth]{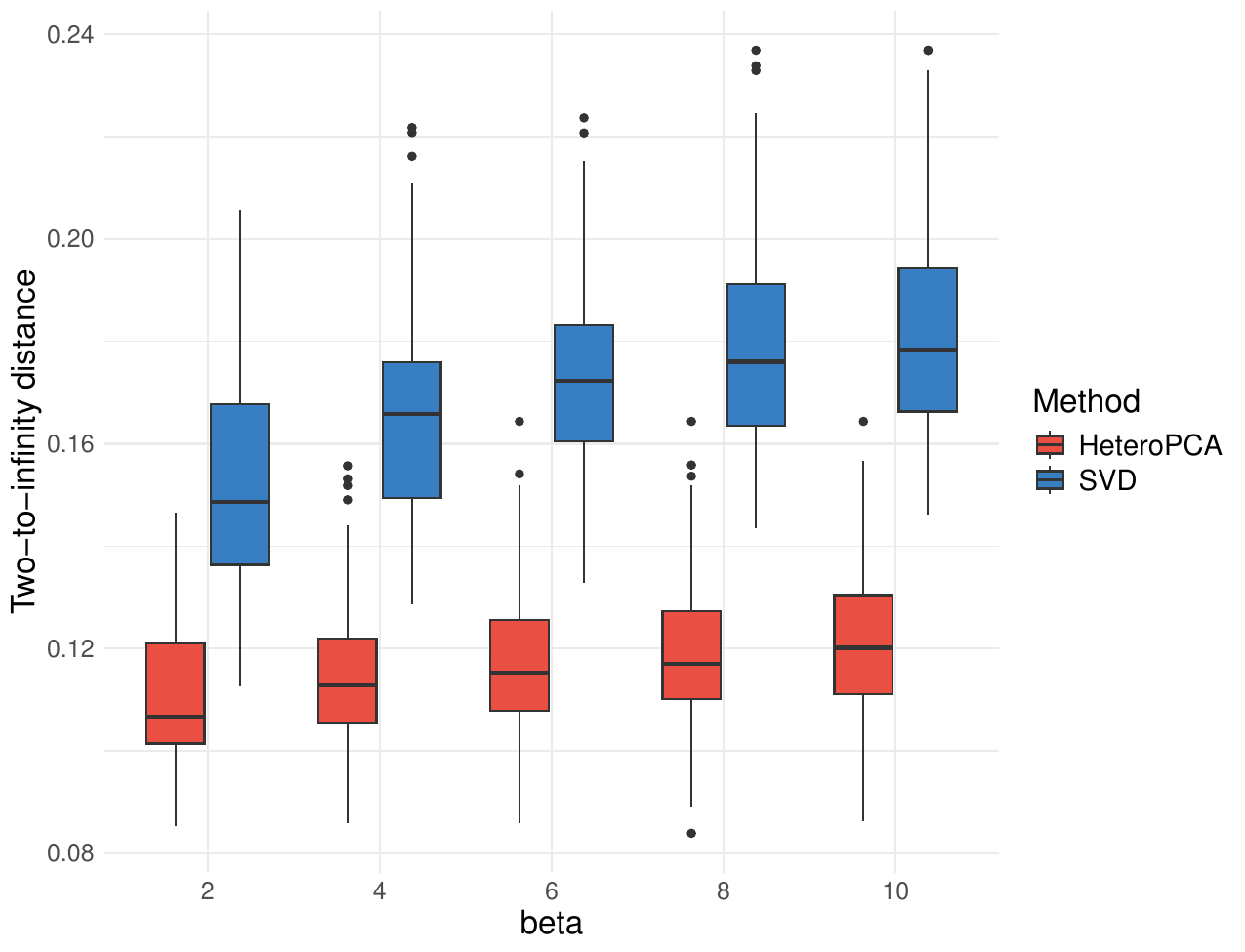}
    \caption{The boxplot of estimation error of $\boldPi$ in two-to-infinity norm.
    The $y$-axis is the two-to-infinity norm. The $x$-axis is $\beta$ that represents the degree of heteroscedasticity. HeteroPCA represents the result of the heteroskedastic-principal-component-analysis-based estimator, and SVD for the singular-value-decomposition-based one.}
    \label{fig:HeteroPCA}
\end{figure}

\section{Real data application}
\label{sec:application}
We illustrate the proposed method using a dataset from the trends in international mathematics and science study (TIMSS) \citep{foy2013timss}. TIMSS is a large-scale international assessment that measures student achievement in mathematics and science. 
We use the grade 8 mathematics assessment data in 2011. One can download the data by \texttt{R} package \texttt{EdSurvey} or from the Database: \texttt{\url{https://timssandpirls.bc.edu/timss2011/international-database.html}}. Each question is labelled with its content domain (e.g. algebra, geometry) and cognitive domain (knowing, applying, and reasoning). Besides mathematics and science test items, the dataset provides extensive contextual covariates from questionnaires administered to students, teachers, and schools. These covariates capture a wide range of aspects, such as students attitudes towards mathematics, students' social economic background, and information about students' schools and teachers.
Based on these covariates, the TIMSS research team provide summary statistics in several domains, including the degree of home resources of a student, whether a student is being bullied at school, etc.
This dataset is particularly suitable for covariate-assisted analysis, as student-level covariates are known to correlate with latent mathematical proficiency.

For subsequent analysis, we select a subset comprising $J=26$ mathematical questions and $W=8$ student-level covariates. These covariates include students' self-assessment of their mathematical ability, the average income level of their residential area, and six composite variables provided by TIMSS: home educational resources, students bullied at school, students like learning mathematics, students value learning mathematics, confidence with mathematics, and students engagement in mathematics lessons. The following analysis only includes students who answers all the 26 selected questions in the United States, resulting in a final sample of $N=620$ participants.

We preprocess the variables in the following way. Students' responses to the mathematics items are binarized, with 1 indicating a correct answer and 0 indicating incorrect. We standardize all eight covariates to be mean 0 and variance 1. For the ordinal covariates, this standardization is preceded by a numerical mapping: the three income levels are coded as -1, 0, and 1, while the four levels of self-perceived mathematics ability are assigned values of -3, -1, 1, and 3.

We apply parallel analysis \citep{Horn1985} and choose $K=3$. We use Algorithm \ref{alg:cv} to select the balance parameter $\alpha$ as $0.45$. The 5-fold cross validation prediction error is displayed in Figure \ref{fig:timss tuning}. A loess line is fitted for clarity. Figure \ref{fig:timss tuning} shows that choosing an appropriate $\alpha$ significantly improves model fit and reduces prediction error compared to the covariate-free approach, which corresponds to the $\alpha=0$ setting.

\begin{figure}[h!]
    \centering
    \includegraphics[width=0.45\textwidth]{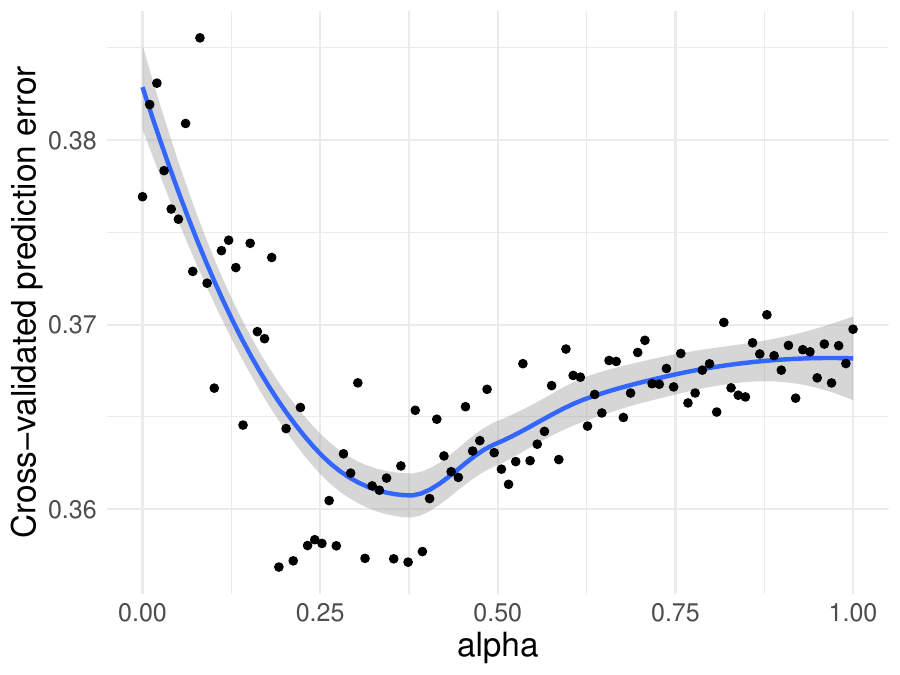}
    \caption{Cross-validated prediction error in mean absolute error by Algorithm \ref{alg:cv}. The $x$-axis represents different values of $\alpha$. The $y$-axis is the mean absolute error. The blue line is the fitted loess line.}
    \label{fig:timss tuning}
\end{figure}

First, we present the estimated item parameters $\hat{\boldTheta}$ from the covariate-assisted model. Figure \ref{fig:TIMSS_theta} displays the heatmap of $\hat{\boldTheta}$. The deeper red shades indicate higher values. Each entry $\hat{\theta}_{jk}$ represents the probability for a student from the $k$th extreme profile answers the $j$th item correctly. The heatmap reveals a clear structure, allowing for a straightforward interpretation of the latent profiles. \textit{Profile 2} represents students who are excellent in mathematics that have high probability in answering all the questions correctly. \textit{Profile 1} characterizes students who are moderately capable in mathematics, and \textit{profile 3} those who are insufficiently capable in mathematics. Additionally, this visualization highlights the specific items that most effectively differentiate students between these proficiency levels.

\begin{figure}[h!]
    \centering
    \includegraphics[width=0.8\textwidth]{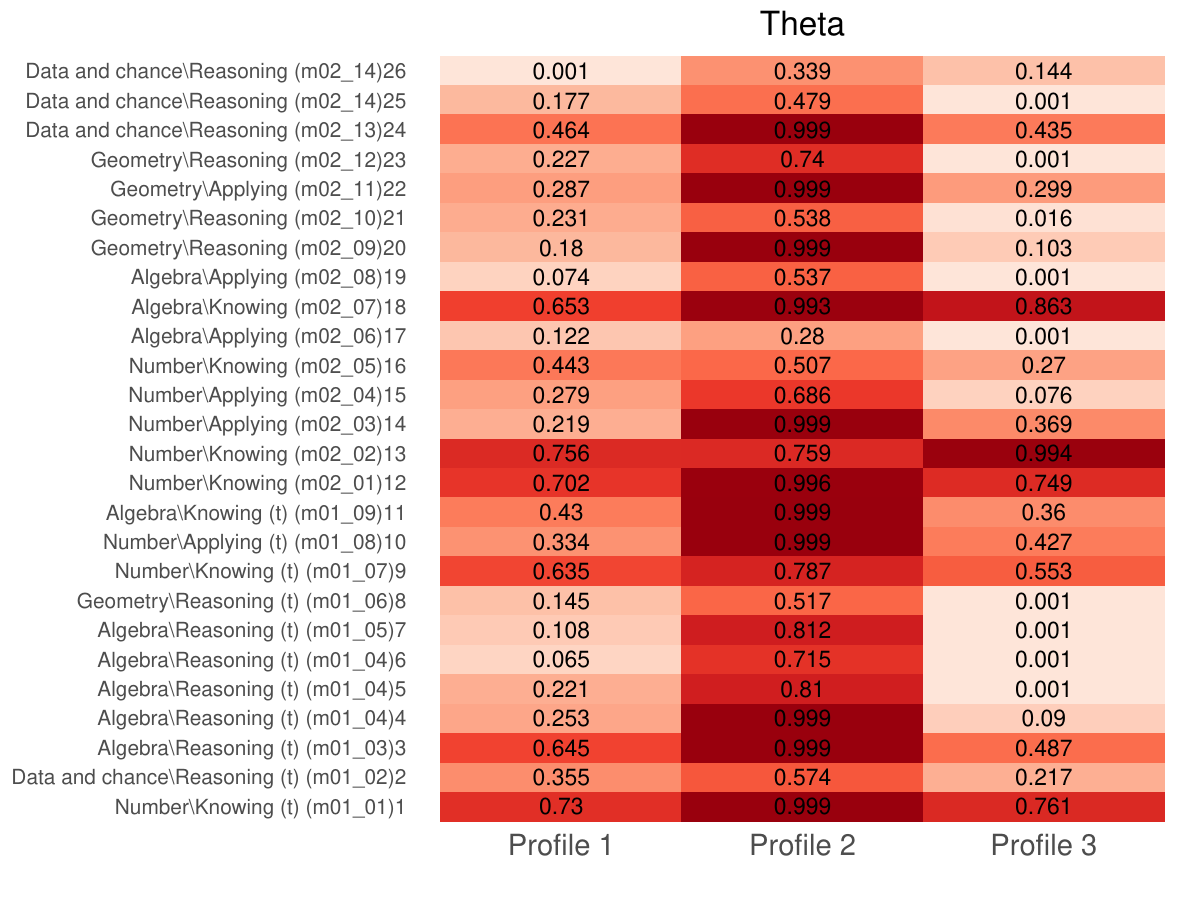}
    \caption{Heatmap of $\Hat{\boldTheta}$ with or without covariates. The values can be interpreted as the probability of answering each question correctly. The $y$ axis displays the knowledge and cognitive domain of each question.}
    \label{fig:TIMSS_theta}
\end{figure}

Next, we examine the estimated mixed membership scores $\hat{\boldPi}$. Figure \ref{fig:timss_pi} presents ternary plots of $\hat{\boldPi}$ from both our proposed covariate-assisted model (left column) and a singular-value-decomposition based covariate-free model (right column). We emphasize that in our procedure, the covariates are not used for post hoc interpretation, but directly inform estimation of the latent structure. In these plots, each point represents a student, whose position reflects its mixed membership score. The closer a point is to the $k$th vertex, the larger its corresponding membership score $\pi_{ik}$ is. For each plot, the points are colored by their value for a given covariate.

The results from our covariate-assisted model show a clear and interpretable structure. For the covariate "home resources", students with more resources are closer to Profile 2, which represent excellence in math, while students with fewer resources are distributed between the "medium" and "insufficient" vertices. There is a different pattern for the "value math" covariate. Students who highly value mathematics are concentrated between Profile 1 and 2, whereas students who do not are positioned closer to being insufficient in math. The clear stratification is also observed for categorical covariates like self-assessed mathematic ability in Figure \ref{fig:timss_ga_c}. 

In contrast, the ternary plots for the covariate-free model do not show such discernible patterns. The covariate values appear randomly distributed across the simplex. In particular, the "good at math" covariate in Figure \ref{fig:timss_ga_nc} is a covariate expected to be highly correlated with actual student performance. However, the covariate-free model fails to capture this relationship, distributing students of all self-assessed abilities almost uniformly.

\begin{figure}[h!]
    \centering
    \subfigure[Home resources without covariates]{
    \centering
    \includegraphics[width=0.40\textwidth]{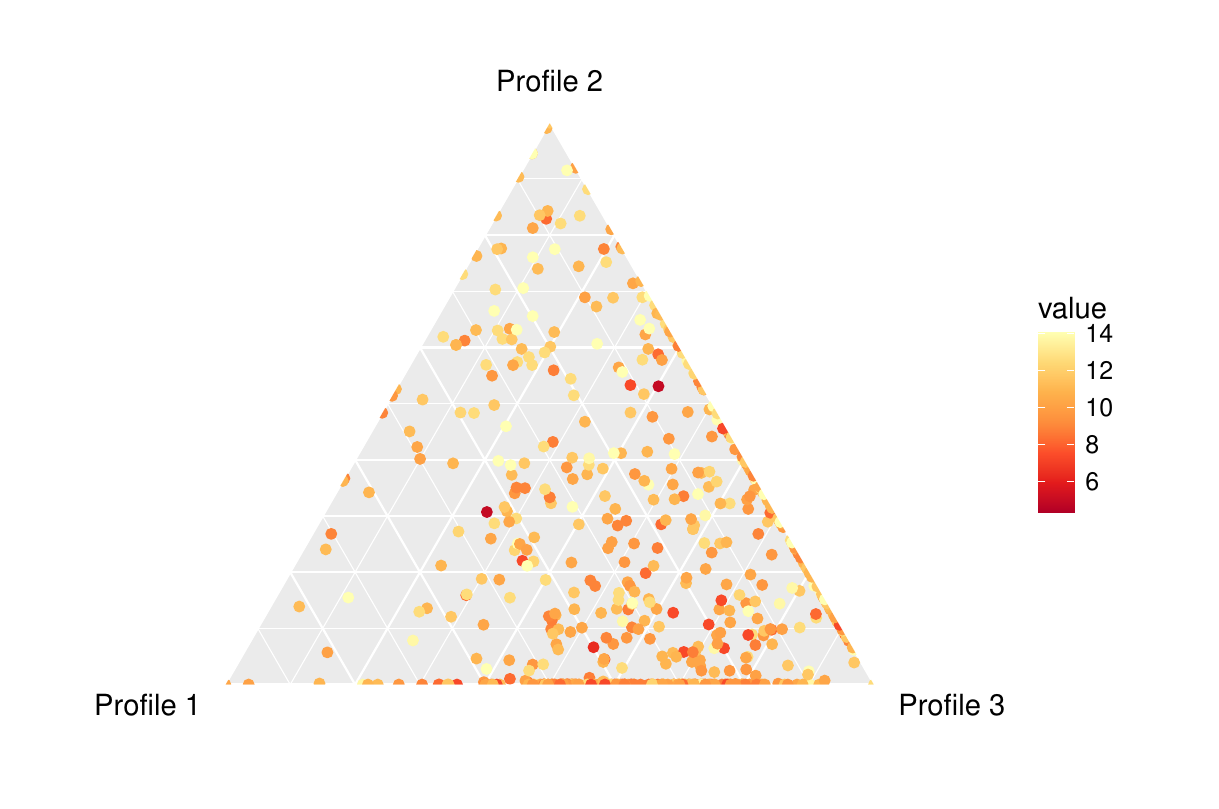}
    \label{fig:timss_hr_nc}
    }
    \hfill
    \subfigure[Home resources with covariates]{
    \centering
    \includegraphics[width=0.40\textwidth]{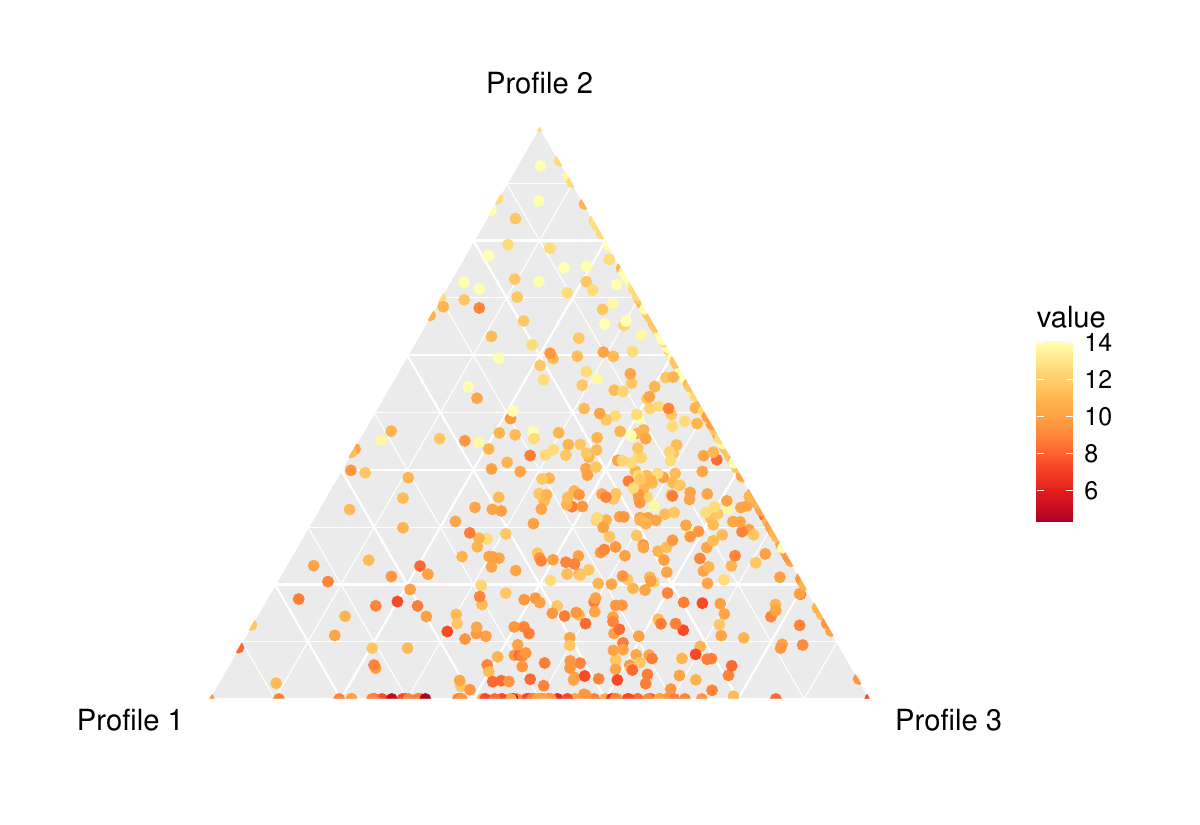}
    \label{fig:timss_hr_c}
    }\\
    \subfigure[Value math without covariates]{
    \centering
    \includegraphics[width=0.40\textwidth]{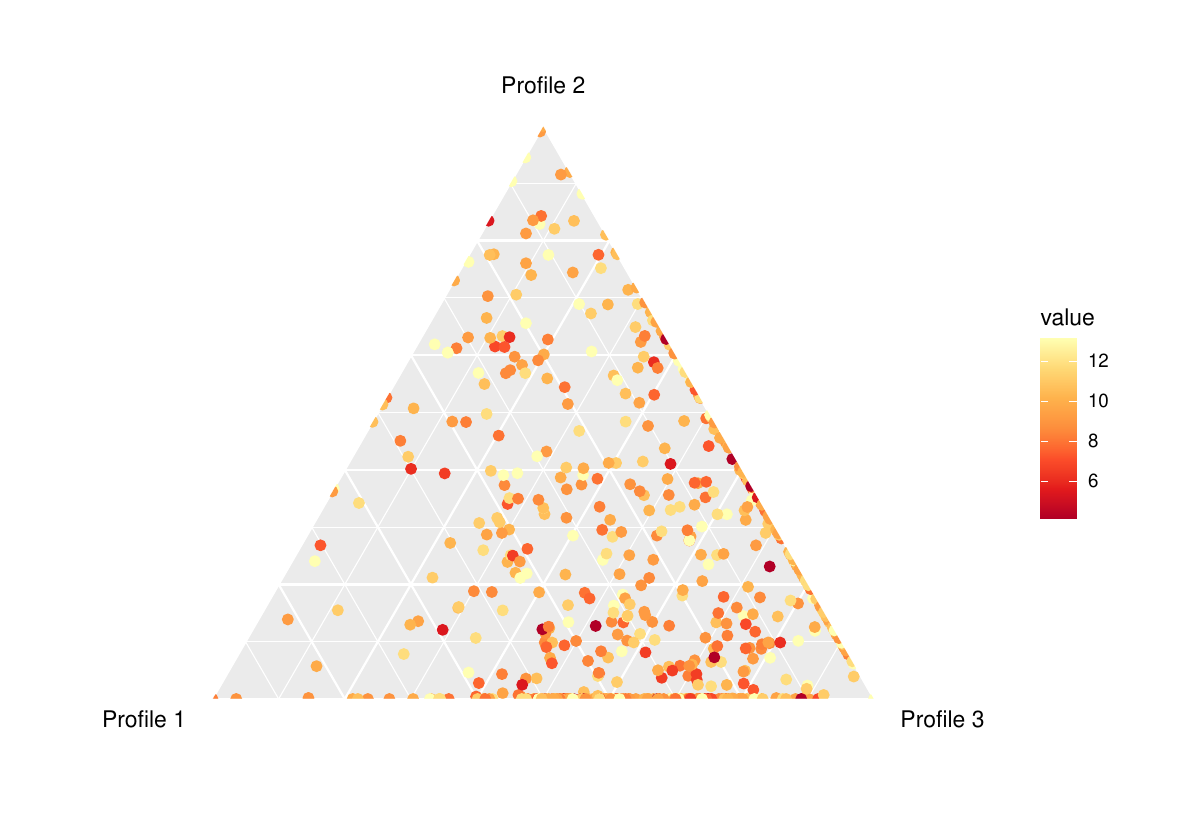}
    \label{fig:timss_value_nc}
    }
    \hfill
    \subfigure[Value math with covariates]{
    \centering
    \includegraphics[width=0.40\textwidth]{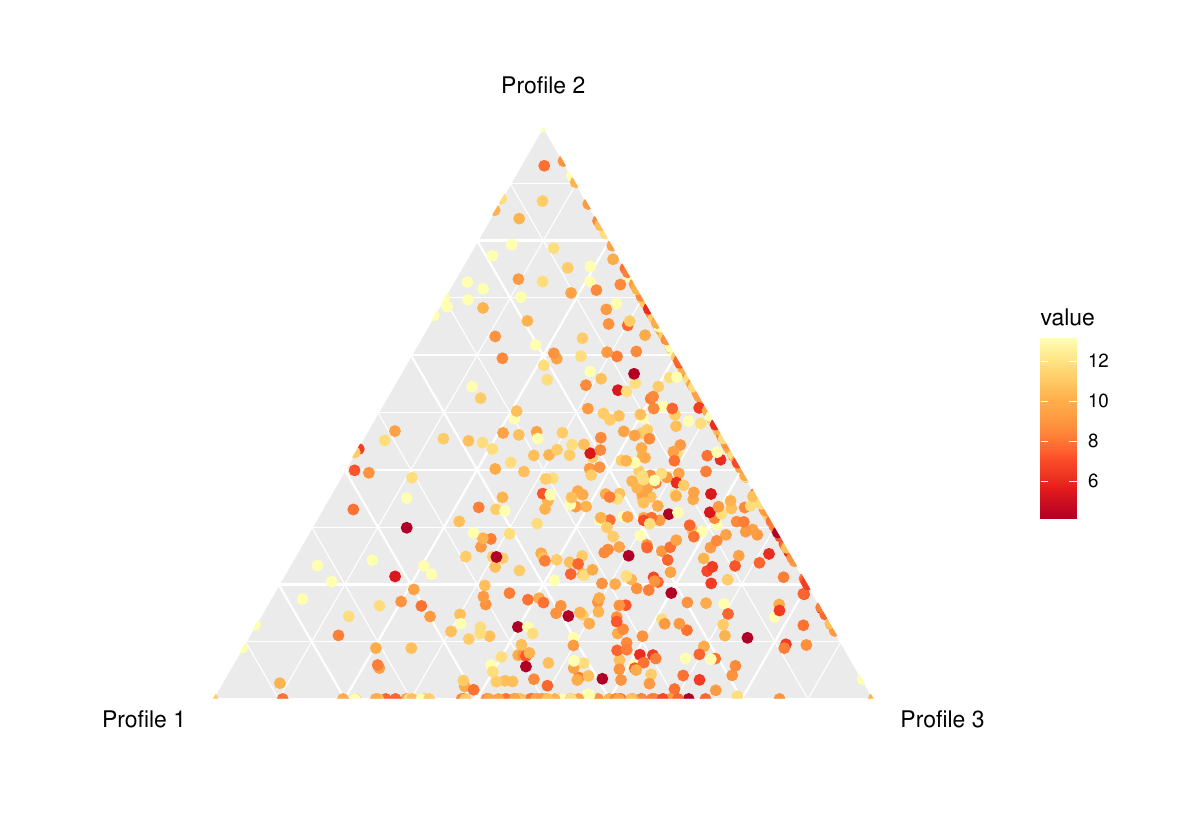}
    \label{fig:timss_value_c}
    }\\
    \subfigure[Good at math without covariates]{
    \centering
    \includegraphics[width=0.40\textwidth]{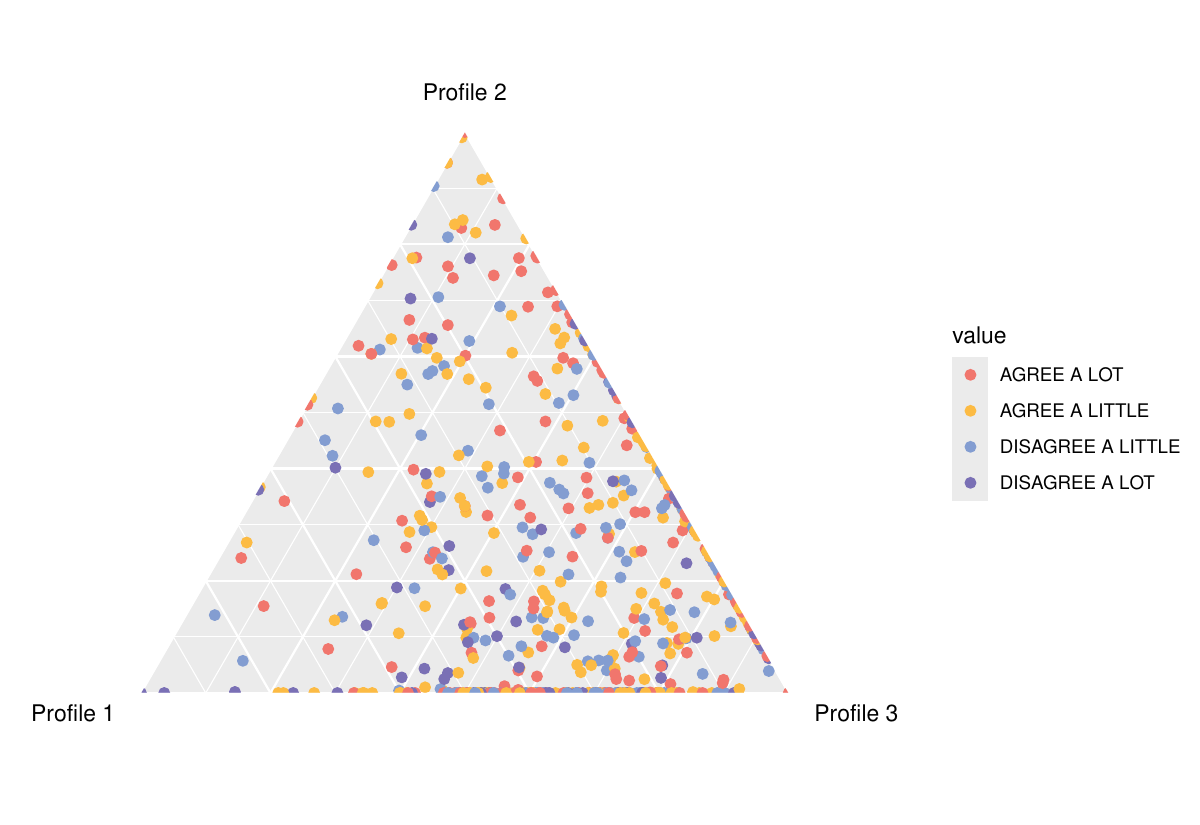}
    \label{fig:timss_ga_nc}
    }
    \hfill
    \subfigure[Good at math with covariates]{
    \centering
    \includegraphics[width=0.40\textwidth]{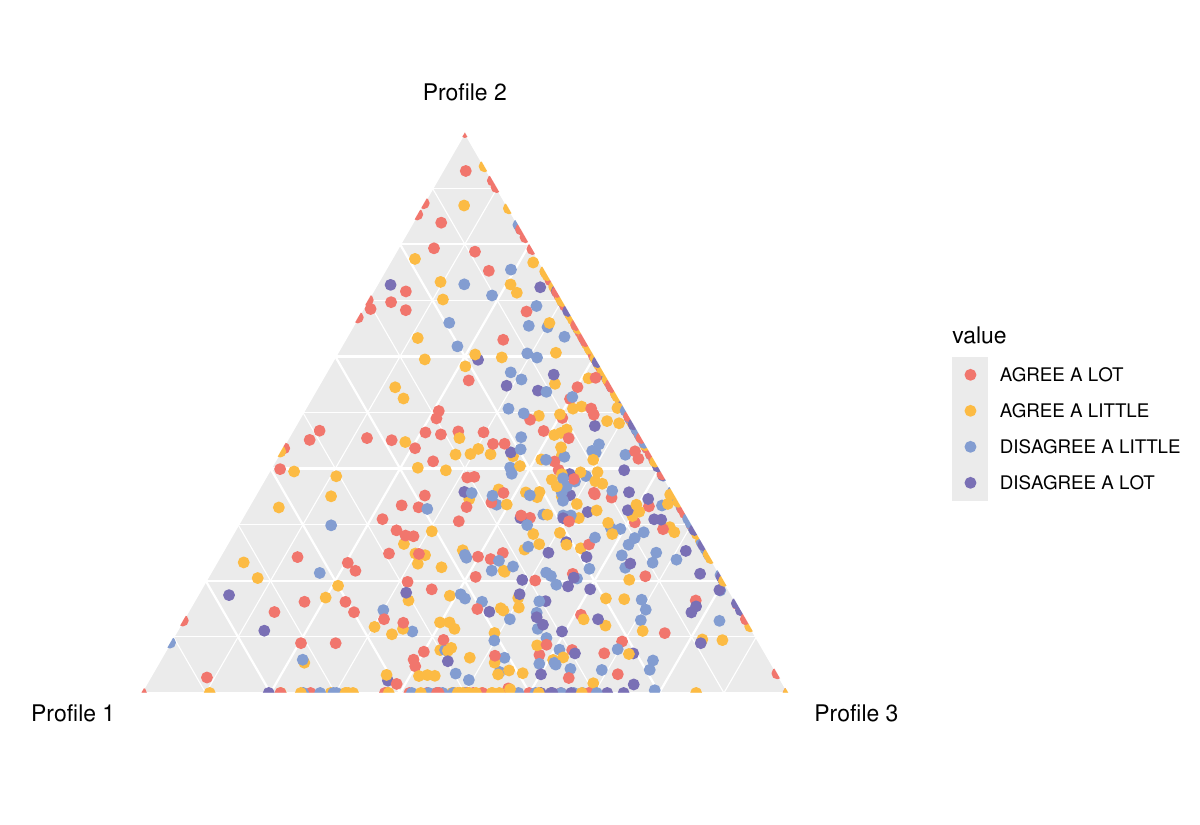}
    \label{fig:timss_ga_c}
    }
    \caption{The ternary plot of the estimated $\hat{\boldPi}$. Each ternary plot is coloured by the corresponding covariate. The relationship between the colour and its value is shown in the legend. Covariate-assistance improves interpretability of the mixed membership estimate.}
    \label{fig:timss_pi}
\end{figure}

In conclusion, this real data application demonstrates two primary advantages of the covariate-assisted grade of membership model. First, by systematically incorporating covariates, the model achieves a lower cross-validated prediction error, indicating superior predictive performance. Second, it produces a more interpretable estimation of the individual mixed membership scores, revealing meaningful relationships between the student-level covariates and their latent mathematical proficiency profiles.

\vspace{-4mm}
\section{Discussion}
\label{sec:conclusion}
In this paper, we have proposed the covariate-assisted grade of membership model. We establish its identifiability under weaker conditions than its covariate-free counterpart. Throughout, the balance parameter $\alpha$ has a structural interpretation as weighting multiple noisy views of the same latent mixed-membership geometry, rather than serving as a purely algorithmic tuning constant. For estimation, we develop an efficient spectral algorithm and provided entry-wise, finite-sample error bounds for all parameters. Our theoretical results show provable improvements in estimating the mixed membership matrix $\boldPi$ when incorporating covariates. We also present a data-driven procedure for selecting the key balance parameter to balance the two sources of information to more effectively estimate the latent structure.
Extensive simulation studies demonstrate that our method has a significant advantage over existing methods in both computational efficiency and estimation accuracy.
While our theory focuses on cases where the number of covariates $W$ exceeds the number of latent profiles $K$, our method is also applicable when $W<K$.

This work focuses on multivariate binary responses in the grade of membership model part. In many applications, the responses are multivariate categorical, such as Likert scale responses. \citet{chen2024generalized} recently proposed an efficient spectral method for such data. They transform each categorical variable with $k$ levels into $k$ correlated binary variables, and flatten the transformed data 
to create an extended response matrix. Our framework can be adapted to this setting. One can first apply this transformation to the categorical responses and then use the algorithm presented in this paper on the resulting extended matrix.

Choosing the number of latent profiles $K$ is a critical step in applying latent variable models. For likelihood-free spectral methods, there are many theoretically sound methods for this task, including parallel analysis \citep{Horn1985, Dobriban2019Deterministic}, singular value thresholding \citep{cai2010singular, donoho2023screenot}, bi-cross validation \citep{owen2009bi, owen2016bi}. In our application, we successfully apply parallel analysis. 

Our work relates to the growing literature on estimating common singular subspaces from multiple matrices \citep{Ma2024Optimal}. \citet{Ma2024Optimal} assume the covariates and the response share the same left singular subspace. They propose directly concatenating the matrices and performing a singular value decomposition to estimate the left singular subspace. While they show this is minimax optimal with respect to the matrix dimensions in the Frobenius norm, their approach does not account for different signal-to-noise ratios across data sources. In Corollary \ref{coro:bound w.r.t. alpha} and Section \ref{sec:sim_tuning}, we theoretically and empirically demonstrate the necessity of adjusting for these differences. Especially, Figure \ref{fig:tuning illus} shows that when the covariate matrix is noisy, without an appropriate balance parameter, direct concatenation of matrices can result in worse performance than the covariate-free case. Under the same setting as \citet{Ma2024Optimal}, \citet{baharav2025stacked} derives an optimal weighting scheme for each source of data under fixed $K$ and known population-level singular values. Our theory allows $K$ to grow to infinity and does not require knowledge of the population singular values. Determining the optimal weighting for mixed membership models is more challenging and remains a topic for future research.

This research also connects to transfer learning of latent representations. Our setting is similar to \citet{Duan2023Target}, but with key differences. First, we focus different models. While \citet{Duan2023Target} studies continuous factor models without additional structures, we focus on mixed membership models, where there are explicit simplex constraints on the mixed membership scores for  interpretability. The more complex mixed membership structures necessitate a tailored estimation algorithm to leverage the simplex geometry and a more delicate analysis to derive theoretical guarantees.
Second, we provide uniform, entry-wise estimation error bounds under minimal assumptions on the balance parameter $\alpha$. \citet{Duan2023Target} provides Frobenius norm guarantees and asymptotic results under more stringent assumptions, including explicit assumptions on $\alpha$. Third, we offer a practical, data-driven selection procedure for $\alpha$.

{We have demonstrated the proposed method on an educational assessment dataset, but our method has far more applications in many other domains.}
For example, our method can be used to analyse the mixed membership structure of single-cell spatial transcriptomics data \citep{BravoGonzalez-Blas2019cisTopic, staahl2016visualization}, and political survey data such as American National Election Studies \citep{debell2010analyze}, with background information of the voters.

\counterwithin{theorem}{section}
\counterwithin{lemma}{section}
\counterwithin{proposition}{section}
\counterwithin{corollary}{section}
\counterwithin{definition}{section}
\renewcommand{\thetheorem}{\Alph{section}.\arabic{theorem}}
\renewcommand{\thelemma}{\Alph{section}.\arabic{lemma}}
\renewcommand{\theproposition}{\Alph{section}.\arabic{proposition}}
\renewcommand{\thecorollary}{\Alph{section}.\arabic{corollary}}
\renewcommand{\thedefinition}{\Alph{section}.\arabic{definition}}

\bibliographystyle{apalike}
\bibliography{ref.bib}

\appendix

\section{Further illustration of Theoretical results}
\subsection{Identifiability}
\label{append:identifiability}
In this section, we first demonstrate Assumption \ref{assump:extreme} is a near necessary condition for identifiability. The proof of the following proposition is deferred to Appendix \ref{sec:iden_nec}.
\begin{proposition}
    \label{prop:iden_nec}
    Suppose $\theta_{jk}\in (0,1)$ for all $j=1,\dots,p$ and $k=1,\dots,K$. If there is one extreme profile that does not have any pure subject, then the covariate-assisted grade of membership model is not identifiable. 
\end{proposition}

Next, we give some examples to further illustrate Theorem \ref{theorem:identifiability}.
\begin{example}
    Consider
    \begin{equation*}
        \cD_1 = \begin{pmatrix}
            0.3 & 0.4 & 0.5\\
            0.4 & 0.3 & 0.7\\
            0.5 & 0.7 & 0.2
        \end{pmatrix},
        \quad 
        \cD_2 = \begin{pmatrix}
            0.1 & 0.2 & 0.3\\
            0.2 & 0.3 & 0.6\\
            0.3 & 0.6 & 0.9
        \end{pmatrix},
        \quad
        \cD_3 = \begin{pmatrix}
            0.1 & 0.2 & 0.3\\
            0.2 & 0.4 & 0.6\\
            0.3 & 0.6 & 1.0
        \end{pmatrix}.
    \end{equation*}
    For the first example, $\text{rank}(\cD_1)=K$. Therefore, it falls into the scenario (a) in Theorem \ref{theorem:identifiability}. If $\boldPi$ satisfies Assumption \ref{assump:extreme}, $(\boldPi,\cD_1)$ is identifiable. 
    
    For the second example, $\text{rank}(\cD_2)=K-1$. Among the top $K-1$ rows of $\cD$, no column is an affine combination of the other columns. Therefore, it falls into the scenario (b) in Theorem \ref{theorem:identifiability}. If Assumption \ref{assump:extreme} is satisfied, $(\boldPi,\cD_2)$ is identifiable.
    
    For the third example, $\text{rank}(\cD_3)=K-1$. In the first two rows, the second column is the average of the first column and the last column. Hence, it falls into the scenario (c) in Theorem \ref{theorem:identifiability}, and $(\boldPi,\cD)$ is not identifiable. 
\end{example}

\subsection{Theoretical guarantees for singular subspaces}
\label{sec:theory_subspace}

Due to the rotation ambiguity of the left singular vectors, we are only able to estimate the singular vectors up to a certain rotation. Therefore, we introduce the rotation matrix $\mathcal{O}_0= \min_{\mathcal{Z}\in \mathcal{O}^{K\times K}} \|\hat{\boldU}\mathcal{Z}-\boldU\|_F$. The proof of the following theorem is in Appendix \ref{sec:proof_leftbound}.

\begin{theorem}
    \label{theorem:leftbound}
    Under Assumption \ref{assump:noise}, \ref{assump:incoherence}, \ref{assump:information} and $\sigma(\boldX)\sqrt{W\log d_X}\asymp \sigma(\boldR)\sqrt{J\log d_R}$, with probability at least $1-O(d^{-10})$, we have
    \begin{align}
        \label{eqn:leftspectralbound}
        \left\|\hat{\boldU}\mathcal{O}_0-\boldU\right\| &\lesssim \frac{1}{\sigma_K(\mathcal{G})} \zetaop(\boldG),\\
        \label{eqn:lefttwoinftybound}
        \left\|\hat{\boldU}\mathcal{O}_0-\boldU\right\|_{2,\infty} &\lesssim (1+\alpha+\alpha^2)\frac{\kappa_{\max}^2}{\sigma_K(\cG)} \sqrt{\frac{\mu K}{N}}\zetaop(\boldG),
    \end{align}
    where $\hat{\boldU}$ and $\boldU$ are respectively the top $K$ eigenvectors of $\boldG$ and $\mathcal{G}$, and $\zetaop(\boldG)=\zeta_{\text{op}}(\boldR) + \alpha \zeta_{\text{op}}(\boldX)$.
\end{theorem}

Then, we present the entry-wise estimation guarantee for the right singular vectors estimated via heteroskedastic principal component analysis. The proof of the following theorem is in Appendix \ref{sec:proof_vl}.
\begin{theorem}
    \label{theorem:VL}
    Instate the assumptions of Theorem \ref{theorem:leftbound}, and additionally $\left\|\mathbf{V}\right\|_{2,\infty}\lesssim \sqrt{\mu K/(J+W)}$, with probability exceeding $1-O(d^{-10})$, we have,
    \begin{align*}
        \left\| \hat{\mathbf{V}} \mathcal{O}_0 - \mathbf{V} \right\|_{2,\infty} &\lesssim \mathcal{E}_{U}^V + \mathcal{E}_E^V,\\
        \left\| \hat{\boldU}\hat{\boldU}^\top \boldL - \cL \right\|_{\infty} &\lesssim \mathcal{E}_U^L + \mathcal{E}_E^L,
    \end{align*}
    where $\hat{\mathbf{V}} = \mathbf{L}^\top \hat{\mathbf{U}} \hat{\boldSigma}^{-1}$, $\mathcal{O}_0$ is defined the same as in Theorem \ref{theorem:leftbound}, and
    \begin{align*}
        \mathcal{E}_{U}^V &= \frac{\kappa_{\max}^2}{\sigma_K(\cG)}\frac{\sqrt{\mu} K}{\sqrt{J+W}}\left(\frac{\zetaop(\boldG)^2}{\sigma_K(\cG)} + \sqrt{\frac{\mu K}{N}}\zetaop(\boldG)\right),\\
        \mathcal{E}_E^V &= \frac{1}{\sigma_K(\cG)^{1/2}}\sqrt{\frac{\mu K}{N}}\left( \frac{\zetaop(\boldR)}{\sigma_K(\cR)} + \frac{\sqrt{\alpha}\zetaop(\boldX)}{\sigma_K(\cX)} \right),\\
        \mathcal{E}_U^L &= (1+\alpha+\alpha^2) \frac{\kappa_{\max}^3}{\sigma_K(\cG)^{1/2}} \frac{\mu K^{3/2}}{\sqrt{N(J+W)}} \left( \frac{\zetaop(\boldG)^2}{\sigma_K(\cG)} + \sqrt{\frac{\mu K}{N}}\zetaop(\boldG) \right),\\
        \mathcal{E}_E^L &=  \kappa_{\max} \frac{\mu K}{N}\left(\frac{\zetaop(\boldR)}{\sigma_K(\cR)} + \frac{\sqrt{\alpha}\zetaop(\boldX)}{\sigma_K(\cX)}\right).
    \end{align*}
\end{theorem}

Different from previous results in singular subspace estimation \citep{Chen2021, chen2024generalized}, we align the estimated right singular vectors and the population ones with the same rotation matrix $\mathcal{O}_0$ for $\hat{\mathbf{U}}$. Since HeteroPCA estimates the right singular vectors via the estimated left singular vectors, the alignment uses the same rotation matrix.

The error bound on $\hat{\mathbf{V}}$ and $\hat{\mathbf{L}}$ can be decomposed into two parts. $\mathcal{E}_U^V$ and $\mathcal{E}_U^L$ are the estimation error resulting from the estimation error of $\boldU$. Because there is inherent error in $\mathbf{L}$, this source of error results in $\mathcal{E}_E^V$ and $\mathcal{E}_E^L$. Under the strong signal scenario, $\mathcal{E}_U^V = \tilde{O}(1 / [(J+W)\sqrt{\min\{N, J+W\}}])$, $\mathcal{E}_U^L = \tilde{O}(1/\sqrt{(J+W)\min\{N, J+W\}})$, $\mathcal{E}_E^V = \tilde{O}(1/\sqrt{N(J+W)})$, and $\mathcal{E}_E^L = \tilde{O}(1/\sqrt{N})$. Even though the estimator of right singular subspace is different from vanilla singular value decomposition \citep{Chen2021, chen2024generalized}, the proposed estimator enjoys the same convergence rate with relaxed assumptions under the strong signal regime. But similarly as in the discussion in Section \ref{sec:theory}, the error bound is not optimal with respect to $K$. If one is willing to make stronger assumptions, they can estimate $\hat{\mathbf{V}}$ via vanilla SVD. In most applications, $K$ is a constant, and there is little difference between the two estimators of the right singular subspace.

\section{Some auxiliary Lemmas}
In this section, we present some basic lemmas concerning the smallest eigenvalue or singular value based on the low-rank structure, and some inequalities frequently applied in later proofs.
\begin{lemma}
    \label{lemma:eigenvaluePiU}
    If $\cD$ is full rank, one has 
    \begin{equation*}
        \boldPi^{\top}\boldPi=(\boldU_S\boldU_S^{\top})^{-1},
    \end{equation*}
    which implies $\lambda_1(\boldU_S\boldU_S^{\top})=1/ \lambda_K(\boldPi^{\top}\boldPi)$, $\lambda_K(\boldU_S\boldU_S^{\top}) = 1/\lambda_1(\boldPi^{\top}\boldPi)$ and $\kappa(\boldU_S\boldU_S^{\top})= \kappa(\boldPi^{\top}\boldPi)$.
\end{lemma}
\begin{proof}
    From Proposition \ref{prop:geom}, $\boldU=\boldPi \boldU_S$, thus,
    \begin{equation*}
        \textbf{I}=\boldU^{\top}\boldU=\boldU_S^{\top}\boldPi^{\top}\boldPi \boldU_S.
    \end{equation*}
    Since $\boldU_S$ is full rank, $\boldPi^{\top}\boldPi=(\boldU_S \boldU_S^{\top})^{-1}$, which implies,
    \begin{align*}
        \lambda_1(\boldU_S\boldU_S^{\top})&=\frac{1}{\lambda_K(\boldPi^{\top} \boldPi)},\\
        \lambda_K(\boldU_S\boldU_S^{\top})&=\frac{1}{\lambda_1(\boldPi^{\top} \boldPi)}.
    \end{align*}
    so $\kappa(\boldU_S\boldU_S^{\top})=\kappa(\boldPi\boldPi^{\top})$.
\end{proof}

\begin{lemma}
    \label{lemma:ksingularvaluebound}
    For any matrix $\boldA_1$, $\boldA_2$, if $\boldA_1\neq \boldsymbol{0}$ and $\boldA_2$ has full row rank, or $\boldA_1$ has full column rank and $\boldA_2\neq \boldsymbol{0}$ one has
    \begin{equation*}
        \sigma_{\min}(\boldA_1\boldA_2)\geq \sigma_{\min}(\boldA_1) \sigma_{\min}(\boldA_2).
    \end{equation*}
\end{lemma}
\begin{proof}
    First, when $\boldA_2\neq\boldsymbol{0}$ and $\boldA_1$ is full column rank, from the quotient form of the smallest singular value, one has,
    \begin{align*}
        \sigma_{\min}(\boldA_1\boldA_2) &= \min_{\mathbf{x}\neq 0} \frac{\|\boldA_1\boldA_2\mathbf{x}\|_2}{\|\mathbf{x}\|_2}\\
        &\geq \min_{\mathbf{x}\neq0} \frac{\|\boldA_1\boldA_2\mathbf{x}\|_2}{\|\boldA_2\mathbf{x}\|_2} \frac{\|\boldA_2\mathbf{x}\|_2}{\|\mathbf{x}\|_2}\\
        &\geq \min_{\mathbf{y}\neq0}\frac{\|\boldA_1\mathbf{y}\|_2}{\|\mathbf{y}\|_2} \min_{\mathbf{x}\neq 0}\frac{\|\boldA_2\mathbf{x}\|_2}{\|\mathbf{x}\|_2}\\
        &\geq \sigma_{\min}(\boldA_1)\sigma_{\min}(\boldA_2).
    \end{align*}
    Note that if $\mathbf{x}$ is in the null space of $\boldA_2$, the inequality is trivial. Now, when $\boldA_1\neq 0$ and $\boldA_2$ is full row rank, similarly, if $\boldA_2\textbf{x}$ is in the null space of $\boldA_1$, the inequality is trivial. Otherwise, it follows by symmetry.
\end{proof}

\begin{lemma}
    \label{lemma:keigenvaluebound}
    If $\cD$ is full rank,
    \begin{equation*}
        \lambda_K(\mathcal{G})\geq \lambda_K(\mathcal{D})\lambda_K(\boldPi^{\top} \boldPi).
    \end{equation*}
\end{lemma}
\begin{proof}
    From \citet[Theorem 1.3.22]{Horn1985}, $(\mathcal{D}\boldPi^{\top})\boldPi$ and $\boldPi(\mathcal{D}\boldPi^{\top})$ have the same $K$ largest eigenvalues in magnitude.
    Note that for all full rank positive definite matrix $\textbf{A}_1, \textbf{A}_2\in \mathbb{R}^{K\times K} $, $\|(\textbf{A}_1 \textbf{A}_2)^{-1}\|\leq \|\textbf{A}_1^{-1}\|\|\textbf{A}_2^{-1}\| $. Besides, $\sigma_K(\textbf{A}_1)=1/\|\textbf{A}_1^{-1}\|$, the same reasoning applies to $\textbf{A}_2$ and $\textbf{A}_1\textbf{A}_2$. Therefore,
    \begin{equation}
        \label{eqn:Keigeninequality}
        \sigma_K(\textbf{A}_1\textbf{A}_2)\geq \sigma_K(\textbf{A}_1) \sigma_K(\textbf{A}_2).
    \end{equation}
    Then, combining \citet[Theorem 1.3.22]{Horn1985} and (\ref{eqn:Keigeninequality}) yields,
    \begin{equation*}
        \lambda_K(\mathcal{G})=\lambda_K(\boldPi\mathcal{D}\boldPi^{\top})= \lambda_K(\mathcal{D}\boldPi^{\top}\boldPi)\geq \lambda_K(\mathcal{D}) \lambda_K(\boldPi^{\top}\boldPi).
    \end{equation*}
\end{proof}

\begin{lemma}
    \label{lemma:small_vals}
    If all $\boldPi$, $\boldTheta$ and $\boldM$ have full column rank $K$, then,
    \begin{equation*}
        \sigma_K(\cG) \geq \sigma_K(\cR)^2 + \alpha\sigma_K(\cX)^2.
    \end{equation*}
\end{lemma}
\begin{proof}
    From the definition of $\cR=\boldPi\boldTheta^\top$, $\cX=\boldPi\boldM^\top$, and $\cG=\boldPi(\boldTheta^\top\boldTheta + \alpha\boldM^\top\boldM)\boldPi^\top$, we can see that
    \begin{equation*}
        \text{col}(\cR)=\text{col}(\cX)=\text{col}(\cG).
    \end{equation*}
    Hence, let $\mathcal{M} = \text{col}(\boldU_R)=\text{col}(\boldU_X)=\text{col}(\boldU)$, where $\boldU_R$, $\boldU_X$ and $\boldU$ are respectively the left singular subspace of $\cR$, $\cX$ and $\cG$. Then, from the variational form of eigenvalues,
    \begin{align*}
        \sigma_K(\cG) = \lambda_K(\cG) &= \min_{\bu\in \mathcal{M},\|\bu\|=1} \bu^\top \cG \bu\\
        &\geq \min_{\bu\in \mathcal{M},\|\bu\|=1}\bu^\top \cR\cR^\top\bu + \alpha \min_{\bu\in \mathcal{M},\|\bu\|=1}\bu^\top\cX\cX^\top\bu\\
        &= \lambda_K(\cR\cR^\top)+\alpha\lambda_K(\cX\cX^\top)\\
        &= \sigma_K(\cR)^2 + \alpha\sigma_K(\cX)^2.
    \end{align*}
\end{proof}

\begin{lemma}
    \label{lemma:incoherence}
    Under Assumption \ref{assump:incoherence} and all $\boldPi$, $\boldTheta$ and $\boldM$ have full column rank $K$, then
    \begin{equation*}
        \left\|\boldU\right\|_{2,\infty}  \leq \sqrt{\frac{\mu K}{N}}.
    \end{equation*}
\end{lemma}
\begin{proof}
    Similar as the proof of Lemma \ref{lemma:small_vals}, we can see that
    \begin{equation*}
        \text{col}(\boldU_R)=\text{col}(\boldU_X)=\text{col}(\boldU).
    \end{equation*}
    Also, note that they are all orthonormal bases of the same column space. Hence, there exists a matrix $\boldA \in \mathbb{R}^{K\times K}$, s.t.
    \begin{equation*}
        \boldU = \boldU_R\boldA,
    \end{equation*}
    where from each column of $\boldU$ is of unit length and every two different columns of $\boldU$ are orthogonal, one has
    \begin{align*}
        \|\bu_k\|_{2}&=\sum_{i=1}^Ka_{ik}^2=1,\quad \forall k\in [K],\\
        \bu_k^\top\bu_l &= \sum_{i=1}^Ka_{ik}a_{il}=0,\quad \forall k\neq l.
    \end{align*}

    Hence, we have
    \begin{equation*}
        \boldA^\top \boldA = \mathbf{I}_K,
    \end{equation*}
    and $\|\boldA\|=1$.

    Then, 
    \begin{equation*}
        \|\boldU\|_{2,\infty}\leq \|\boldU_R\|_{2,\infty}\|\boldA\|\leq \sqrt{\frac{\mu K}{N}}.
    \end{equation*}
\end{proof}

\section{Proofs of identifiability}
\subsection{Proof of Proposition \ref{prop:geom}}
First, denote the corresponding expected data matrix and covariate matrix for the set $S$ as $\cR_{S,:}$ and $\cX_{S,:} $ respectively. Then, denote
\begin{equation*}
    \mathcal{L} = \begin{bmatrix}
        \cR & \sqrt{\alpha}\cX
    \end{bmatrix},
\end{equation*}
and
\begin{equation*}
    \mathcal{L}_{S,:} = \begin{bmatrix}
        \cR_{S,:} & \sqrt{\alpha}\cX_{S,:}
    \end{bmatrix}
    = \begin{bmatrix}
        \boldTheta^\top & \sqrt{\alpha}\boldM^\top
    \end{bmatrix}
    = \boldU_{S,:}\boldSigma\mathbf{V}^\top.
\end{equation*}
Therefore,
\begin{equation*}
    \mathcal{L} = \boldPi\begin{bmatrix}
        \boldTheta^\top & \sqrt{\alpha}\boldM^\top
    \end{bmatrix}=\boldPi\boldU_{S,:}\boldSigma\mathbf{V}^\top = \boldU\boldSigma\mathbf{V}^\top,
\end{equation*}
where the last inequality comes from direct SVD of $\mathcal{L}$. Since $\mathbf{V}$ is an orthonormal matrix with full column rank, and $\boldSigma$ is full rank, one has
\begin{equation*}
    \boldU=\boldPi \boldU_{S,:}.
\end{equation*}
Since both $\boldU$ and $\boldPi$ are rank $K$, $\boldU_{S,:}$ must be full rank. Therefore, $\boldPi$ can be expressed as,
\begin{equation*}
    \boldPi =\boldU\boldU_{S,:}^{-1}.
\end{equation*}

\subsection{Proof of Proposition \ref{prop:iden_nec}}
\label{sec:iden_nec}
\begin{proof}
    Without loss of generality, assume that the first extreme latent profile does not have a pure object. Then $\exists \delta >0$, s.t. $\pi_{i1}\leq 1-\delta$, $\forall i=1,\dots,N$. For every $0<\varepsilon<\delta$, define a $K\times K$ matrix
    \begin{equation*}
        \mathbf{N}_{\varepsilon} = \begin{bmatrix}
            1+(K-1)\varepsilon^2 & -\varepsilon^2\boldsymbol{1}_{K-1}^{\top}\\
            \boldsymbol{0}_{K-1} & \varepsilon \boldsymbol{1}_{K-1} \boldsymbol{1}_{K-1}^{\top} + (1-(K-1)\varepsilon)\mathbf{I}_{K-1}
        \end{bmatrix}.
    \end{equation*}
    Then, let $\Tilde{\boldPi}_{\varepsilon}=\boldPi \mathbf{N}_{\varepsilon}$,  $\Tilde{\boldTheta}_{\varepsilon } = \boldTheta\mathbf{N}_{\varepsilon}^{-1} $, $\Tilde{\boldM}_{\varepsilon} = \boldM \mathbf{N}_{\varepsilon}^{-1} $, and $\Tilde{\cD}_{\varepsilon} = \Tilde{\boldTheta}_{\varepsilon}^\top \Tilde{\boldTheta}_{\varepsilon} + \alpha \tilde{\boldM}_{\varepsilon}^\top \tilde{\boldM}_{\varepsilon}$. We will show that $(\tilde{\boldPi}, \tilde{\cD}, \tilde{\boldTheta})$ also form a valid parameter set. 
    
    Since there is no explicit parametric constraint on $\Tilde{\cD}_{\varepsilon}$, we are to prove all entries of $\Tilde{\boldPi}$ and $\Tilde{\boldTheta}$ lies in $[0,1]$, and the sum of each row of $\Tilde{\boldPi}$ equals 1. Since $\mathbf{N}_0=\mathbf{I}_K$ and by the continuity of matrix determinant, there exists a small enough $\varepsilon$ such that $\mathbf{N}_{\varepsilon}$ is still full-rank. Further, note that $\mathbf{N}_{\varepsilon}\boldsymbol{1}_K =\boldsymbol{1}_K$. Therefore, 
    \begin{equation*}
        \Tilde{\boldPi}_{\varepsilon}\boldsymbol{1}_K = \boldPi \mathbf{N}_{\varepsilon}\boldsymbol{1}_K = \boldPi \boldsymbol{1}_K = \boldsymbol{1}_N.
    \end{equation*}
    
    Further, for each $i$, $\Tilde{\pi}_{i1}=\pi_{i1}(1+(K-1)\varepsilon^2)\geq 0 $. For any fixed $k = 2,\dots, K$, $(\mathbf{N}_{\varepsilon})_{k,k}=1-(K-2)\varepsilon $ and $(\mathbf{N}_{\varepsilon})_{m,k}=\varepsilon $ for $m\neq k$. Thus, when $\varepsilon \leq 1/(K-1)$, we have $(\boldM_{\varepsilon})_{m,k}\geq \varepsilon $ for any $m = 1,\dots, K$. The following inequalities hold for every $i=1,\dots,N$ and $k=2,\dots,K$:
    \begin{align*}
        \Tilde{\pi}_{ik} &= -\varepsilon^2\pi_{i1} + \sum_{m=2}^K\pi_{im}(\textbf{N}_{\varepsilon})_{m,k} \geq -\varepsilon^2 \pi_{i1} + \varepsilon \sum_{m=2}^K \pi_{im}\\
        &\geq -\varepsilon^2(1-\delta) + \varepsilon (1-\pi_{i1}) \geq -\varepsilon^2(1-\delta) + \varepsilon\delta \geq \varepsilon\delta^2 >0. 
    \end{align*}
    
    As for $\tilde{\boldTheta}$, for the continuity of matrix inverse and Frobenius norm, it can be shown that $\|\mathbf{N}_{\varepsilon}^{-1}-\mathbf{I}_{K} \|_F \to 0$, $\varepsilon \to 0$. Therefore,
    \begin{equation*}
        \left\|\Tilde{\boldTheta}-\boldTheta\right\|_F = \left\|\boldTheta\left(\mathbf{I}_K - \mathbf{N}_{\varepsilon}^{-1}\right)\right\|_F \leq \|\boldTheta\|_2 \left\|\mathbf{I}_K - \mathbf{N}_{\varepsilon}^{-1}\right\|_F \xrightarrow{\varepsilon\rightarrow 0} 0. 
    \end{equation*}
    Since all the elements of $\boldTheta$ are strictly in $(0,1)$, the elements of $\Tilde{\boldTheta}$ must be in $[0,1]$ when $\varepsilon$ is sufficiently small. Since $\textbf{N}_{\varepsilon}$ is not a permutation when $\varepsilon >0$, the covariate-assisted grade of membership model is not identifiable.
\end{proof}

\subsection{Proof of Theorem \ref{theorem:identifiability}}
\begin{proof}
    Denote $\text{rank}(\cD)=r\leq K$. Without loss of generality, by reordering the rows of the data matrix and covariate matrix simultaneously, we can let $\boldPi_{1:K,:}=\mathbf{I}_K $ by Assumption \ref{assump:extreme}. Denote the spectral decomposition of the population matrix $\mathcal{G}= \boldU\boldsymbol{\Lambda} \boldU^{\top}$, where $\boldU\in\mathbb{R}^{N\times r} $, $\boldsymbol{\Lambda} \in\mathbb{R}^{r\times r} $. According to Proposition \ref{prop:geom},
    \begin{equation}
        \label{eqn:nopermut}
        \boldU=\boldPi \boldU_{1:K,:}.
    \end{equation}
    Following $\text{rank}(AB)\leq\min\{\text{rank}(A),\text{rank}(B) \} $, and $\text{rank}(\boldPi)=K $, $\text{rank}(\boldU)=r $, it suffices to show $\text{rank}(\boldU_{1:K,:})=\text{rank}(\boldU)=r $.\\
    \indent Suppose another set of parameters $(\Tilde{\boldPi}, \Tilde{\cD}, 
    \Tilde{\boldTheta})$ gives rise to the same population matrix $\mathcal{G}$ and denote its corresponding pure subject index vector as $\Tilde{S}$ so that $\Tilde{\boldPi}_{\Tilde{S},:}=\mathbf{I}_K $. Similarly, we have
    \begin{equation}
        \label{eqn:permut}
        \boldU=\Tilde{\boldPi}\boldU_{\Tilde{S},:}.
    \end{equation}
    Taking the $\Tilde{S}$ rows of (\ref{eqn:nopermut}) and the first $K$ rows of (\ref{eqn:permut}) yields
    \begin{align*}
        \boldPi_{\Tilde{S},:}\boldU_{1:K,:}&=\boldU_{\Tilde{S},:},\\
        \Tilde{\boldPi}_{1:K,:}\boldU_{\Tilde{S},:}&=\boldU_{1:K,:}.
    \end{align*}
    Therefore, $\boldU_{1:K}$ is in the convex hull of $\boldU_{\Tilde{S},:}$, and $\boldU_{\Tilde{S},:}$ is in the convex hull of $\Tilde{\boldU}$. Thus, there exists a permutation matrix $P$ such that $\boldU_{\Tilde{S},:}=\mathbf{P}\boldU_{1:K,:} $, and it gives,
    \begin{equation}
        \label{eqn:equpmut}
        (\boldPi-\Tilde{\boldPi}\mathbf{P})\boldU_{1:K,:}=0.
    \end{equation}
    \textbf{Proof of part (a)}. If $r=K$, $\boldU_{1:K,:} $ is full rank from (\ref{eqn:nopermut}). Then, (\ref{eqn:equpmut}) implies $\boldPi=\Tilde{\boldPi}\textbf{P}$ and thus $\Tilde{\boldTheta}=\boldTheta \mathbf{P}^{\top}$, $\Tilde{\boldM}=\boldM \mathbf{P}^{\top}$.\\
    \indent As for $r<K$, $\mathcal{D}$ can be written in the following form,
    \begin{equation*}
        \cD = 
        \begin{bmatrix}
            \mathbf{C} & \mathbf{C}\mathbf{W}_1\\
            \mathbf{W}_1^{\top}\mathbf{C} & \mathbf{W}_1^{\top}\mathbf{C}\mathbf{W}_1
        \end{bmatrix}.
    \end{equation*}
    where $\mathbf{C}\in\mathbb{R}^{r\times r} $ is full rank, $\mathbf{W}_1\in\mathbb{R}^{r\times (K-r)} $. From the proof of Proposition \ref{prop:geom} that $\mathcal{D}=\boldU_{1:K,:}\boldsymbol{\Lambda} \boldU_{1:K,:} $, focusing on the first block and the second block, it yields,
    \begin{align*}
        \mathbf{C} = (\boldU_{1:r,:})\boldsymbol{\Lambda}(\boldU_{1:r,:})^{\top},\\
        \mathbf{C}\mathbf{W}_1 = (\boldU_{1:r,:})\Lambda(\boldU_{(r+1):K,:})^{\top}.
    \end{align*}
    Since $\mathbf{C}$ is full rank, $\boldU_{1:r,:} $ is full rank, the above equations translate to
    \begin{equation}
        \label{eqn:uprop}
        \boldU_{1:K,:}=\begin{bmatrix}
            \mathbf{I}_r\\
            \mathbf{W}_1^{\top}
        \end{bmatrix}
        \boldU_{1:r,:}.
    \end{equation}
    Combining (\ref{eqn:uprop}) and (\ref{eqn:equpmut}), we have
    \begin{equation}
        \label{eqn:rlk_id}
        (\boldPi-\Tilde{\boldPi}\mathbf{P})\begin{bmatrix}
            \mathbf{I}_r\\
            \mathbf{W}_1^{\top}
        \end{bmatrix} = 0.
    \end{equation}
    \textbf{Proof of part (b)} When $\text{rank}(\cD)=K-1$, by simple calculation, one has,
    \begin{align*}
        \mathcal{D} &= \boldU_{1:K,:}\boldsymbol{\Lambda} \boldU_{1:K,:}\\
        &= \begin{bmatrix}
            \mathbf{I}_{K-1}\\
            \mathbf{W}_1^{\top}
        \end{bmatrix}
        \boldU_{1:K-1,:}\boldsymbol{\Lambda} \boldU_{1:K-1,:}\begin{bmatrix}
            \mathbf{I}_{K-1} & \mathbf{W}_1
        \end{bmatrix}\\
        &= \begin{bmatrix}
            \mathbf{I}_{K-1}\\
            \mathbf{W}_1^{\top}
        \end{bmatrix} \mathcal{D}_{1:(K-1),1:(K-1)} \begin{bmatrix}
            \mathbf{I}_{K-1} & \mathbf{W}_1
        \end{bmatrix}.
    \end{align*}
    Hence, $\mathcal{D}_{1:(K-1),K}=\mathcal{D}_{1:(K-1),1:(K-1)}\mathbf{W}_1 $ .
    
    Eqn (\ref{eqn:rlk_id}) is equivalent to the condition in the proof of \citet[Theorem 2]{chen2024spectral}. Hence, to identify the parameters, it needs to satisfy $\mathbf{W}_1^{\top}\boldsymbol{1}_r\neq 1$. In this case, this condition implies that among the first $K-1$ rows, the $K$th column of $\mathcal{D}$ is not an affine combination of the other columns. Since $\mathcal{D}$ is symmetric, it applies to the first $K-1$ columns as well.\\
    \textbf{Proof of part (c).} This is proven in \citet{chen2024spectral} Theorem 2, scenario (c). It states that when conditions in (a) and (b) fail, one can construct $\Tilde{\boldPi}\neq \boldPi$ that $(\Tilde{\boldPi},\boldTheta)$ still qualifies as a valid parameter set. The same reasoning applies here as well.
\end{proof}

\section{Proofs of estimation guarantee}
\subsection{Proof of Bernoulli noise}
In this section, we will show Bernoulli noise does satisfy the noise condition in Assumption 2. First, note that Bernoulli noise of probability $p_{ij}$ has the following form,
\begin{equation*}
    Y_{ij}=\left\{
    \begin{aligned}
        1-p_{ij}&, \text{ with probability } p_{ij}\\
        -p_{ij}&, \text{ with probability } 1-p_{ij}
    \end{aligned}
    \right.
\end{equation*}
\indent It is easy to verify that ${E}[Y_{ij}]=0$. And variance is bounded by,
\begin{align*}
    \mathrm{var}(Y_{ij})=p_{ij}(1-p_{ij})\leq \frac{1}{4}.
\end{align*}
\indent Therefore, $\sigma^2=1$, $T=1$, and $T^2/\sigma^2\leq 4$.

\subsection{A few key Lemmas}

First, we define some quantities for simplicity. Define $\mathbf{E}^R=\boldR-\mathcal{R}$. Following \citet{caiSubspaceEstimationUnbalanced2020}, we have
\begin{align*}
    \label{eqn:incoherencebound}
    \max_{i,j}|E_{i,j}^R|&\leq \sigma(\boldR)\frac{\min\{\sqrt{J}, \sqrt[4]{NJ} \}}{\sqrt{\log d_R}} := B_R,\\
    \max_{i,j}\sqrt{\E[E_{i,j}^2]}&\leq \sigma(\boldR),\\
    \max_{i}\sqrt{\sum_j\E[E_{i,j}^2]}&\leq \sigma(\boldR) \sqrt{J},\\
    \max_{j}\sqrt{\sum_i\E[E_{i,j}^2]} &\leq \sigma(\boldR)\sqrt{N},
\end{align*}
with probability at least $1-O(d_R^{-12})$. The main difference between the above quantities and the ones in \citet{caiSubspaceEstimationUnbalanced2020} come from HeteroPCA while \citet{caiSubspaceEstimationUnbalanced2020} applies diagonal deletion. Similarly, we define the quantities for $\boldX$ using the similar notation.

Our following proof strategy follows the framework of \citet{Yan2021}. A key distinction comes from the two sources of information. We apply the triangle inequality to derive a sharper bound on the balance parameter $\alpha$. We detail only these specific deviations, omitting steps that are analogous to the original analysis to avoid redundancy.

\begin{lemma}
    \label{lemma:wholebound}
    Under Assumption \ref{assump:noise} and \ref{assump:information}, suppose the number of iterations exceeds $t_0\geq\\ \log \max \left\{{\sigma_1(\cR)^2}/{\zeta_{\text{op}}(\boldR)}, {\sigma_1(\cX)^2}/{\zeta_{\text{op}}(\boldX)} \right\}$, with probability at least $1-O(d^{-10})$, 
    \begin{align*}
        \|\boldG-\mathcal{G}\| &\lesssim \zeta_{\text{op}}(\boldR)  + \alpha \zeta_{\text{op}}(\boldX) := \zetaop(\boldG),\\
        \left\|\mathcal{P}_{\text{diag}}\left( \boldG - \cG \right)\right\| &\lesssim \kappa_{\max}^2\sqrt{\frac{\mu K}{N}}(\zetaop(\boldR)+\alpha\zetaop(\boldX)),
    \end{align*}
    where $\zeta_{op}(\boldR)$ and $\zeta_{op}(\boldX)$ are defined to be
    \begin{align*}
        \zeta_{\text{op}}(\boldR) &= \sigma^2(\boldR)\sqrt{NJ}\log d_R+\sigma(\boldR)\sigma_1(\cR)\sqrt{N\log d_R},\\
        \zeta_{\text{op}}(\boldX) &= \sigma^2(\boldX)\sqrt{NR}\log d_X+\sigma(\boldX)\sigma_1(\cX)\sqrt{N\log d_X},
    \end{align*}
    and $d_R$, $d_X$ are defined as in Assumption \ref{assump:information}, and $d=\max\{N,J,W\}$.
\end{lemma}

\begin{proof}
    See Proof in Appendix \ref{sec:lem_wholebound}.
\end{proof}

\begin{lemma}
    \label{lemma:combine}
    Under Assumption \ref{assump:noise} and \ref{assump:information}, 
    \begin{equation*}
        \zeta_{op}(\boldR) +\alpha\zeta_{op}(\boldX) \ll \lambda_K(\mathcal{G}).
    \end{equation*}
\end{lemma}

\begin{proof}
    See Proof in Appendix \ref{sec:lem_com}.
\end{proof}

\begin{lemma}
    \label{lemma:spectralbound}
    Under Assumption \ref{assump:noise} and \ref{assump:information}, with probability at least $1-O(d^{-10})$, we have
    \begin{align*}
        \left\|\hat{\boldU}\boldH - \boldU\right\|&\lesssim \frac{\zeta_{\text{op}}(\boldR)+\alpha\zeta_{\text{op}}(\boldX)}{\sigma_K(\cG)},\\
        \|\boldH-\mathcal{O}_0\|&\lesssim \frac{(\zeta_{\text{op}}(\boldR)+\alpha\zeta_{\text{op}}(\boldX))^2}{\sigma_K(\cG)^2},\\
        \left\|\boldH^\top\boldH-\mathbf{I}_K\right\|&\lesssim \frac{(\zeta_{\text{op}}(\boldR)+\alpha\zeta_{\text{op}}(\boldX))^2}{\sigma_K(\cG)^2},
    \end{align*}
    and 
    \begin{equation*}
        \sigma_1(\boldH)\asymp \sigma_K(\boldH) \asymp 1.
    \end{equation*}
\end{lemma}

\begin{proof}
    See proof in Appendix \ref{sec:lem_spec}.
\end{proof}

\begin{lemma}
    \label{lemma:kappa}
    Under the assumption that both $\cR$ and $\cX$ are full rank, one has,
    \begin{equation*}
        \kappa(\mathcal{L})\leq \kappa_{\max},
    \end{equation*}
    where $\kappa_{\max}$ is defined as in Assumption \ref{assump:information}.
\end{lemma}
\begin{proof}
    See in Appendix \ref{sec:pf_lemma_kappa}.
\end{proof}

\begin{lemma}
    \label{lemma:29}
    Suppose $N\gtrsim \kappa_{\max}^4\mu K $, then with probability $1-O(d^{-10})$,
    \begin{align*}
        \left\|\cGm - \boldG\right\| &\lesssim \left[\sigma(\boldR)^2+\alpha\sigma(\boldX)^2 \right]\sqrt{N\max\{J, W \}}\log d + \\ &\left[\sigma(\boldR) \sigma_1(\cR) + \alpha \sigma(\boldX)\sigma_1(\cX) \right] \kappa_{\max}^2\sqrt{\mu K\log d} ,\\
        \left\|\cGm - \cG\right\| &\lesssim \zeta_{\text{op}}(\boldR) + \alpha \zeta_{\text{op}}(\boldX).
    \end{align*}
    hold for all $1\leq m \leq N$.
\end{lemma}
\begin{proof}
    See proof in Appendix \ref{sec:lem_29}
\end{proof}

\begin{lemma}
    \label{lemma:leftbound1}
    Suppose $N\gtrsim \kappa_{\max}^4\mu K$, with probability exceeding $1-O(d^{-10})$,
    \begin{equation*}
        \left\|(\boldG-\cG)_{m,:}\boldU \right\|_2 \lesssim \left( \zeta_{\text{op}}(\boldR)+\alpha\zeta_{\text{op}}(\boldX) \right)\sqrt{\frac{\mu K}{N}},
    \end{equation*}
\end{lemma}

\begin{proof}
    See proof in Appendix \ref{sec:lem_left1}.
\end{proof}

\begin{lemma}
    \label{lemma:32}
    Under Assumption \ref{assump:noise}, \ref{assump:information}, and $\sigma(\boldR)^2J\log d_R  \asymp \sigma(\boldX)^2W\log d_X $, with probability at least $1-O(d^{-10})$,
    \begin{equation}
        \label{eqn:F18}
        \begin{aligned}
        &\left\|\mathbf{E}_{m,:}^R[\Pm(\boldR)]^\top\left( \Um\Hm-\boldU\right)\right\|_2 \\&\lesssim \zeta_{\text{op}}(\boldR)\left\|\Um\Hm-\boldU\right\|_{2,\infty} + \frac{\zeta_{\text{op}}(\boldR)^2}{\sigma_K(\cG)}\left\| \Um\Hm \right\|_{2,\infty}\\
        &+ \sqrt{\frac{\mu K}{N}}\frac{\zeta_{\text{op}}(\boldR)^2}{\sigma_K(\cG)} +\alpha \left[\frac{\zeta_{\text{op}}(\boldR)\zeta_{\text{op}}(\boldX)}{\sigma_K(\cG)}\left\| \Um\Hm \right\|_{2,\infty} + \sqrt{\frac{\mu K}{N}} \frac{\zetaop(\boldG)^{2}}{\sigma_K(\cG)} \right],
        \end{aligned}
    \end{equation}
    \begin{equation}
        \label{eqn:F19}
        \begin{aligned}
            \left\|\mathbf{E}_{m,:}^R[\Pm(\boldR)]^\top \Um\Hm \right\|_2 \lesssim \zeta_{\text{op}}(\boldR)\left(\left\|\Um\Hm \right\|_{2,\infty} + \sqrt{\frac{\mu K}{N}} \right).
        \end{aligned}
    \end{equation}
\end{lemma}
\begin{proof}
    See proof in Appendix \ref{sec:lem_32}.
\end{proof}

\begin{lemma}
\label{lemma:GUbound}
    Under Assumption \ref{assump:noise}, \ref{assump:information} and $\sigma(\boldR)^2J\log d_R  \asymp \sigma(\boldX)^2W\log d_X$, with probability exceeding $1-O(d^{-10})$, one has
    \begin{align*}
        \left\|\hat{\boldU} \hat{\boldSigma}^2\boldH -\boldG\boldU\right\|_{2,\infty} &= \left\|\hat{\boldG}(\hat{\boldU}\boldH-\boldU) \right\|_{2,\infty}\\
        &\lesssim \kappa_{\max}^2\frac{\zeta_{\text{op}}(\boldR) \zeta_{\text{op}}(\boldG)}{\sigma_K(\cG)}\sqrt{\frac{\mu K}{N}} + \zetaop(\boldR)\left\|\hat{\boldU}\boldH-\boldU\right\|_{2,\infty}\\
        &+\alpha\left[ \kappa_{\max}^2\frac{\zeta_{\text{op}}(\boldX) \zeta_{\text{op}}(\boldG)}{\sigma_K(\cG)}\sqrt{\frac{\mu K}{N}} + \zetaop(\boldX)\left\|\hat{\boldU}\boldH-\boldU\right\|_{2,\infty} +\sqrt{\frac{\mu K}{N}}\frac{\zetaop(\boldG)^{2}}{\sigma_K(\cG)}\right]\\
        &+\alpha^2\left[ \frac{\zeta_{\text{op}}(\boldR)\zetaop(\boldX)}{\sigma_K(\cG)}\sqrt{\frac{\mu K}{N}} + \frac{\zetaop(\boldR)\zetaop(\boldX)}{\sigma_K(\cG)} \left\|\hat{\boldU}\boldH-\boldU\right\|_{2,\infty} +  \sqrt{\frac{\mu K}{N}} \frac{\zetaop(\boldX)^{2}}{\sigma_K(\cG)} \right]. 
    \end{align*}
    
\end{lemma}
\begin{proof}
    See proof in Appendix \ref{sec:lem_GU}.
\end{proof}

\begin{lemma}
    \label{lemma:l1_l2_U}
    Under certain regulatory assumptions, with probability exceeding $1-O\left(d^{-10} \right)$,
    \begin{align*}
        \left\| \Um\UmT- \Uml\UmlT \right\| &\lesssim \frac{1}{\sigma_K(\cG)}\left(B_R\log d_R + \sigma(\boldR)\sqrt{N\log d_R}\right)^2 \left\|\Um\Hm\right\|_{2,\infty}\\
        &+ \frac{\sigma(\boldR)^{2}}{\sigma_K(\cG)} + \frac{1}{\sigma_K(\cG)}\left( B_R\log d_R + \sigma(\boldR)\sqrt{N\log d_R} \right)\left\|\cR^\top\right\|_{2,\infty}\\
        &+ \alpha\left[ \frac{1}{\sigma_K(\cG)}\left(B_X\log d_X + \sigma(\boldX)\sqrt{N\log d_X}\right)^2 \left\|\Um\Hm\right\|_{2,\infty}\right.\\
        &+\left. \frac{\sigma(\boldX)^{2}}{\sigma_K(\cG)} + \frac{1}{\sigma_K(\cG)}\left( B_X\log d_X + \sigma(\boldX)\sqrt{N\log d_X} \right)\left\|\cX^\top\right\|_{2,\infty} \right],
    \end{align*}
    simultaneously for all $m\in [N]$ and $l\in [J+W]$.
\end{lemma}
\begin{proof}
    See proof in Appendix \ref{sec:lem_l1l2}.
\end{proof}

\begin{lemma}
    \label{lemma:33}
    Under the same assumptions, with probability at least $1-O(d^{-10})$,
    \begin{equation*}
        \left\|\Um\UmT-\hat{\boldU}\hat{\boldU}^\top\right\| \lesssim \kappa_{\max}^2\frac{\zetaop(\cG)}{\sigma_K(\cG)} \left( \left\| \Um\Hm\right\|_{2,\infty} + \sqrt{\frac{\mu K}{N}}\right),
    \end{equation*}
    simultaneously for $m\in [N]$.
\end{lemma}
\begin{proof}
    See proof in Appendix \ref{sec:lem_33}.
\end{proof}

\subsection{Proof of Theorem \ref{theorem:leftbound}}
\label{sec:proof_leftbound}
\textbf{(1) Proof of Spectral norm bound:}
First, by \citet[Lemma 2.5 and 2.6]{Chen2021},
    \begin{equation}
        \left\|\Hat{\boldU}\mathcal{O}_0-\boldU\right\|\leq \sqrt{2}\left\|\Hat{\boldU}\Hat{\boldU}^{\top}-\boldU\boldU^{\top}\right\| = \sqrt{2}\left\|\sin \mathbf{\Psi} (\Hat{\boldU},\boldU)\right\|.
    \end{equation}
    $\mathbf{\Psi}(\Hat{\boldU},\boldU)$ is a diagonal matrix whose $i$th diagonal entry is the $i$th principal angle between two subspaces represented by $\Hat{\boldU}$ and $\boldU$. From Lemma \ref{lemma:wholebound} and Lemma \ref{lemma:combine}, 
    \begin{equation*}
        \|\boldG-\mathcal{G}\|\ll \sigma_K(\cG),
    \end{equation*}
    with probability at least $1-O(d^{-10})$. Together with Weyl's inequality gives,
    \begin{equation*}
        \sigma_{K+1}(\boldG)\leq \sigma_{K+1}(\mathcal{G}) + \|\boldG-\mathcal{G}\| \leq \frac{\sigma_K(\cG)}{2}.
    \end{equation*}
    Then, from Davis-Kahan inequality \citep[Chapter 5, Theorem 3.6]{stewart1990matrix}, it is shown that,
    \begin{equation*}
        \left\|\sin \mathbf{\Psi} (\Hat{\boldU},\boldU)\right\| \leq \frac{\|\boldG-\mathcal{G}\|}{\sigma_K(\mathcal{G})-\sigma_{K+1}(\boldG)}\leq \frac{\|\boldG-\mathcal{G}\|}{\sigma_K(\cG)-\sigma_K(\cG)/2} = \frac{2\|\boldG-\mathcal{G}\|}{\sigma_K(\cG)}.
    \end{equation*}
    Therefore, with probability at least $1-O(d^{-10})$,
    \begin{equation}
        \label{eqn:spectralbound}
        \left\|\Hat{\boldU}\mathcal{O}_0-\boldU\right\|\leq \frac{2\sqrt{2}\|\boldG-\mathcal{G}\|}{\sigma_K(\cG)} \leq 2\sqrt{2} \frac{\zeta_{\text{op}}(\boldR)+\alpha\zeta_{\text{op}}(\boldX)}{\sigma_K(\cG)}.
    \end{equation}

\textbf{(2) Proof of $\|\cdot\|_{2,\infty} $ bound:}
    Since $\cG = \boldU\boldSigma^2\boldU^\top$, $\boldU=\cG\boldU (\boldSigma)^{-2}$. We have
    \begin{align*}
        \|\hat{\boldU}\boldH-\boldU\|_{2,\infty} \leq \underbrace{\left\|\hat{\boldU}\boldH - \boldG\boldU(\boldSigma)^{-2}\right\|_{2,\infty}}_{\beta_1} + \underbrace{\left\| (\boldG-\cG)\boldU(\boldSigma)^{-2} \right\|_{2,\infty}}_{\beta_2}.
    \end{align*}

    \textbf{Bounding} $\beta_1$. First, from \citet[Lemma 1]{Abbe2019}, for every $m\in [N]$,
    \begin{equation*}
        \left\|\left(\hat{\boldU}\boldH - \boldG\boldU(\boldSigma)^{-2}\right)_{m,\cdot}\right\|_{2} \lesssim \underbrace{\frac{1}{\sigma_K(\cG)^2} \|\boldG-\cG\|\|\boldG_{m,\cdot}\boldU\|_2}_{\beta_{1,1}} + \underbrace{\frac{1}{\sigma_K(\cG)}\left\| \boldG_{m,\cdot}(\hat{\boldU}\boldH-\boldU) \right\|_2}_{\beta_{1,2}} .
    \end{equation*}
    For $\beta_{1,1}$, 
    \begin{align*}
        \|\boldG_{m,\cdot}\boldU\|_2 &\leq \left\|(\boldG-\cG)_{m,\cdot}\boldU \right\|_2 + \left\|\cG_{m,\cdot}\boldU\right\|_2\\
        &\leq \left\|(\boldG-\cG)_{m,\cdot}\boldU \right\|_2 + \left\|\boldU\boldSigma^2 \boldU^\top \boldU\right\|_{2,\infty}\\
        &= \left\|(\boldG-\cG)_{m,\cdot}\boldU \right\|_2 + \left\| \boldU\boldSigma^2 \right\|_{2,\infty}\\
        &\lesssim \zeta_{\text{op}}(\boldG)\sqrt{\frac{\mu K}{N}} + \sigma_1(\cG)\|\boldU\|_{2,\infty}\\
        &\lesssim \sigma_1(\cG)\sqrt{\frac{\mu K}{N}},
    \end{align*}
    where the second to last inequality comes from Lemma \ref{lemma:leftbound1}, and the last inequality comes from $\zeta_{\text{op}}(\boldG)\lesssim \sigma_1(\cG)$.

    Combining with Lemma \ref{lemma:individualbound}, one has,
    \begin{align*}
        \beta_{1,1} \lesssim \frac{1}{\sigma_K(\cG)^2}\zeta_{\text{op}}(\boldG)\sigma_1(\cG)\sqrt{\frac{\mu K}{N}} \lesssim \kappa_{\max}^2\frac{\zeta_{\text{op}(\boldG)}}{\sigma_K(\cG)}\sqrt{\frac{\mu K}{N}}.
    \end{align*}

    For $\beta_{1,2}$, following Lemma \ref{lemma:GUbound}, one has
    \begin{align*}
        \beta_{1,2} &\leq \frac{1}{\sigma_K(\cG)}\left\|\boldG\left(\hat{\boldU}\boldH-\boldU\right)\right\|_{2,\infty} \\
        &\lesssim \kappa_{\max}^2\frac{\zeta_{\text{op}}(\boldR) \zeta_{\text{op}}(\boldG)}{\sigma_K(\cG)^2}\sqrt{\frac{\mu K}{N}} + \frac{\zetaop(\boldR)}{\sigma_K(\cG)}\left\|\hat{\boldU}\boldH-\boldU\right\|_{2,\infty}\\
        &+\alpha\left[ \kappa_{\max}^2\frac{\zeta_{\text{op}}(\boldX) \zeta_{\text{op}}(\boldG)}{\sigma_K(\cG)^2}\sqrt{\frac{\mu K}{N}} + \frac{\zetaop(\boldX)}{\sigma_K(\cG)}\left\|\hat{\boldU}\boldH-\boldU\right\|_{2,\infty} + \sqrt{\frac{\mu K}{N}} \frac{\zetaop(\boldG)^2}{\sigma_k(\cG)^2} \right]\\
        &+\alpha^2\left[ \frac{\zeta_{\text{op}}(\boldR)\zetaop(\boldX)}{\sigma_K(\cG)^2}\sqrt{\frac{\mu K}{N}} + \frac{\zetaop(\boldR)\zetaop(\boldX)}{\sigma_K(\cG)^2} \left\|\hat{\boldU}\boldH-\boldU\right\|_{2,\infty} + \sqrt{\frac{\mu K}{N}} \frac{\zetaop(\boldG)^2}{\sigma_k(\cG)^2} \right].
    \end{align*}

    Hence, for $\beta_1$, we have,
    \begin{align*}
        \beta_1 &\leq \beta_{1,1} + \beta_{1,2}\\
        &\lesssim \kappa_{\max}^2\frac{\zeta_{\text{op}(\boldG)}}{\sigma_K(\cG)}\sqrt{\frac{\mu K}{N}} + \frac{\zetaop(\boldR)}{\sigma_K(\cG)}\left\|\hat{\boldU}\boldH-\boldU\right\|_{2,\infty}\\
        &+\alpha\left[ \kappa_{\max}^2\frac{\zeta_{\text{op}}(\boldG)}{\sigma_K(\cG)}\sqrt{\frac{\mu K}{N}} + \frac{\zetaop(\boldX)}{\sigma_K(\cG)}\left\|\hat{\boldU}\boldH-\boldU\right\|_{2,\infty} + \sqrt{\frac{\mu K}{N}} \frac{\zetaop(\boldG)^2}{\sigma_k(\cG)^2} \right]\\
        &+\alpha^2\left[ \frac{\zeta_{\text{op}}(\boldR)\zetaop(\boldX)}{\sigma_K(\cG)^2}\sqrt{\frac{\mu K}{N}} + \frac{\zetaop(\boldR)\zetaop(\boldX)}{\sigma_K(\cG)^2} \left\|\hat{\boldU}\boldH-\boldU\right\|_{2,\infty} + \sqrt{\frac{\mu K}{N}} \frac{\zetaop(\boldG)^2}{\sigma_k(\cG)^2} \right],
    \end{align*}
    with provision that $\zetaop(\boldX)\ll \sigma_K(\cX)^2 \lesssim \sigma_K(\cG)$.

    Then, for $\beta_2$, with the help of Lemma \ref{lemma:leftbound1},
    \begin{align*}
        \beta_2 \leq \frac{1}{\sigma_K(\cG)} \leq \zetaop(\cG) \sqrt{\frac{\mu K}{N}}.
    \end{align*}

    Therefore, we obtain
    \begin{align*}
        \left\| \hat{\boldU}\boldH - \boldU \right\|_{2,\infty} &\lesssim \kappa_{\max}^2\frac{\zeta_{\text{op}}(\boldG)}{\sigma_K(\cG)}\sqrt{\frac{\mu K}{N}} + \frac{\zetaop(\boldR)}{\sigma_K(\cG)}\left\|\hat{\boldU}\boldH-\boldU\right\|_{2,\infty}\\
        &+\alpha\left[ \kappa_{\max}^2\frac{\zeta_{\text{op}}(\boldG)}{\sigma_K(\cG)}\sqrt{\frac{\mu K}{N}} + \frac{\zetaop(\boldX)}{\sigma_K(\cG)}\left\|\hat{\boldU}\boldH-\boldU\right\|_{2,\infty} + \sqrt{\frac{\mu K}{N}} \frac{\zetaop(\boldG)^2}{\sigma_k(\cG)^2} \right]\\
        &+\alpha^2\left[ \frac{\zeta_{\text{op}}(\boldR)\zetaop(\boldX)}{\sigma_K(\cG)^2}\sqrt{\frac{\mu K}{N}} + \frac{\zetaop(\boldR)\zetaop(\boldX)}{\sigma_K(\cG)^2} \left\|\hat{\boldU}\boldH-\boldU\right\|_{2,\infty} + \sqrt{\frac{\mu K}{N}} \frac{\zetaop(\boldG)^2}{\sigma_k(\cG)^2} \right]\\
        &\lesssim (1+\alpha+\alpha^2)\kappa_{\max}^2\frac{\zeta_{\text{op}}(\boldG)}{\sigma_K(\cG)} \sqrt{\frac{\mu K}{N}},
    \end{align*}
    where the last inequality comes from $\max\{\zetaop(\boldR),\zetaop(\boldX)\}\lesssim \zetaop(\boldG)\ll \sigma_K(\cG)$.

    Combined with Lemma \ref{lemma:spectralbound}, we arrive at
    \begin{align*}
        \left\|\hat{\boldU}\mathcal{O}-\boldU\right\|_{2,\infty} &\leq \left\| \hat{\boldU}\boldH - \boldU \right\|_{2,\infty} + \left\|\hat{\boldU}\left(\boldH-\mathcal{O}\right)\right\|\\
        &\lesssim (1+\alpha+\alpha^2)\kappa_{\max}^2\frac{\zeta_{\text{op}}(\boldG)}{\sigma_K(\cG)} \sqrt{\frac{\mu K}{N}} + \left(\|\boldU\|_{2,\infty} + \left\|\hat{\boldU}\mathcal{O}-\boldU\right\|_{2,\infty}\right)\left\| \boldH-\mathcal{O} \right\|\\
        &\lesssim (1+\alpha+\alpha^2)\kappa_{\max}^2\frac{\zeta_{\text{op}}(\boldG)}{\sigma_K(\cG)} \sqrt{\frac{\mu K}{N}} + \frac{\zetaop(\boldG)^2}{\sigma_K(\cG)^2}\sqrt{\frac{\mu K}{N}} + \frac{\zetaop(\boldG)^2}{\sigma_K(\cG)^2}\left\| \hat{\boldU}\mathcal{O}-\boldU \right\|_{2,\infty}\\
        &\lesssim (1+\alpha+\alpha^2)\kappa_{\max}^2\frac{\zeta_{\text{op}}(\boldG)}{\sigma_K(\cG)} \sqrt{\frac{\mu K}{N}}.
    \end{align*}    
        
\subsection{Proof of Theorem \ref{theorem:mmmbound}}
First, we present Theorem 3 of \citet{Gillis2014} for completeness.
\begin{lemma}[\citet{Gillis2014}, Theorem 3]
    \label{lemma:quote}
    Let $\boldM^{\prime}=\boldM+\mathbf{N}=\mathbf{W}\mathbf{B}+\mathbf{N}\in \mathbb{R}^{m\times n} $, where $\boldM=\mathbf{W}\mathbf{H}=\mathbf{W}[\mathbf{I}_r|\mathbf{B}^{\prime}]$, $\mathbf{W}\in\mathbb{R}^{m\times r} $, $\mathbf{B}\in \mathbb{R}^{r\times n}_+ $ and $\sum_{k=1}^rB_{kj}^{\prime}\leq 1 $, $\forall j$ and $r \geq 2$. Let $K(\mathbf{B})=\max_i \|\mathbf{B}(:,i)\|_2 $, and $\|\mathbf{N}(:,i)\|_2\leq \epsilon$ for all $i$ with
    \begin{equation*}
        \epsilon < \sigma_r(\mathbf{W})\min\left(\frac{1}{2\sqrt{r-1}},\frac{1}{4} \right)\left(1+80\frac{K(\mathbf{W})^2}{\sigma_r(\mathbf{W})^2} \right)^{-1},
    \end{equation*}
    and $J$ be the index set of cardinality $r$ extracted by sequential projection algorithm, where $\sigma_r(\mathbf{W})$ is the $r$th singular value of $\mathbf{W}$. Then there exists a permutation $\mathbf{P}$ of $\{1,2,\cdots, r\}$ such that
    \begin{equation*}
        \max_{1\leq j\leq r}\|\mathbf{M}^{\prime}(:,J(j))-\mathbf{W}(:,P(j))\|\leq \Bar{\epsilon} = \epsilon\left(1+80\frac{K(\mathbf{W})^2}{\sigma_r(\mathbf{W})^2} \right).
    \end{equation*}
\end{lemma}
\begin{lemma}
    \label{lemma:errorbound1}
    Let $S$ be the index set returned by Algorithm \ref{alg:gom}. Under Assumption \ref{assump:noise} and \ref{assump:information}, there exists a permutation matrix $\mathbf{P}\in \mathbb{R}^{K\times K} $ such that,
    \begin{equation*}
        \left\|\Hat{\boldU}_{S,:}-\mathbf{P}^{\top}\boldU_{S,:} \text{sgn}(\boldH)^\top\right\|_{2,\infty} \lesssim \kappa(\boldPi)^2 \epsilon,
    \end{equation*}
    with probability at least $1-O(d^{-10})$, where $\epsilon$ is defined in Theorem \ref{theorem:mmmbound}.
\end{lemma}
\begin{proof}
    From Proposition \ref{prop:geom}, $\boldU=\boldPi \boldU_{S,:}$. Let $\mathbf{M}^{\prime}=\Hat{\boldU}^{\top}$, $\mathbf{M}=sgn(\boldH)\boldU^{\top}$, $\mathbf{W}=sgn(\boldH)\boldU_{S,:}^{\top}$, $\mathbf{B}=\boldPi^{\top}$. From Theorem \ref{theorem:leftbound} and the identity that $\mathcal{O}_0=sgn(\boldH)$, $\|\Hat{\boldU}-\boldU sgn(\boldH)^{\top}\|_{2,\infty}\lesssim \epsilon $ uniformly for probability at least $1-O(d^{-10})$. By Lemma \ref{lemma:quote},
    \begin{equation*}
        \max_{1\leq j\leq K}\left|\left|\left(\Hat{\boldU}_{S,:}^{\top}-sgn(\boldH)\boldU_{S,:}^{\top}\mathbf{P} \right)e_j \right|\right| \lesssim \epsilon \left(1+80 \frac{K(sgn(\boldH)\boldU_{S,:}^\top)^2}{\sigma_K(sgn(\boldH)\boldU_{S,:}^\top)^2} \right).
    \end{equation*}
    \indent Note that $K(sgn(\boldH)\boldU_{S,:}^\top)=\|\boldU_{S,:} sgn(\boldH)^\top\|_{2,\infty}\leq \|\boldU_{S,:}\| \|sgn(\boldH)\| =\sigma_1(\boldU_{S,:}) $. From Lemma \ref{lemma:keigenvaluebound}, $\sigma_K(sgn(\boldH)\boldU_{S,:}^\top) \geq \sigma_K(sgn(\boldH))\sigma_K(\boldU_{S,:})=\sigma_K(\boldU_{S,:}) $. Therefore, one has,
    \begin{equation*}
        \frac{K(sgn(\boldH)\boldU_{S,:}^\top)^2}{\sigma_K^2(sgn(\boldH)\boldU_{S,:}^\top)} \leq \kappa(\boldU_{S,:})^2 = \kappa(\boldPi)^2,
    \end{equation*}
    where the second equality comes from Lemma \ref{lemma:eigenvaluePiU}. Then, by taking transpose, it is shown that
    \begin{equation*}
        \left\|\Hat{\boldU}_{S,:}-\mathbf{P}^{\top}\boldU_{S,:} sgn(\boldH)^\top\right\|_{2,\infty} \lesssim \kappa(\boldPi)^2\epsilon.
    \end{equation*}
\end{proof}

\begin{lemma}
    \label{lemma:errorbound2}
    If Assumption \ref{assump:noise} and \ref{assump:information} are satisfied, 
    \begin{equation*}
        \left|\left|\boldU sgn(\boldH)^\top \left(\Hat{\boldU}_{S,:}^{-1} - (\mathbf{P}^{\top}\boldU_{S,:} sgn(\boldH)^\top)^{-1} \right)\right|\right|_{2,\infty} \lesssim \sigma_1(\boldPi)\kappa^2(\boldPi)\epsilon.
    \end{equation*}
    with probability at least $1-O(d^{-10})$.
\end{lemma}
\begin{proof}
    For notational simplicity, define $\mathbf{F}=sgn(\boldH)^\top$, and $\Tilde{\boldU}_{S,:}=\mathbf{P}^{\top} \boldU_{S,:} \mathbf{F}$, then,
    \begin{align}
        \label{eqn:Ubound2all}
        \begin{aligned}
            &\left|\left|\boldU sgn(\boldH)^\top\left(\Hat{\boldU}_{S,:}^{-1}-(\mathbf{P}^{\top}\boldU_{S,:} sgn(\boldH)^\top)^{-1} \right)\right|\right|_{2,\infty}\\ &= \|\boldU\mathbf{F}(\Hat{\boldU}_{S,:}^{-1}-\Tilde{\boldU}_{S,:}^{-1})\|_{2,\infty}\\
            &= \|\boldU\mathbf{F}\Tilde{\boldU}_{S,:}^{-1}\left(\Tilde{\boldU}_{S,:}-\Hat{\boldU}_{S,:} \right)\Hat{\boldU}_{S,:}^{-1}\|_{2,\infty}\\
            &\leq \|\boldU\mathbf{F}\mathbf{F}^{-1}\boldU_{S,:}^{-1}\mathbf{P} \left(\Tilde{\boldU}_{S,:}-\Hat{\boldU}_{S,:} \right)\|_{2,\infty} \|\Hat{\boldU}_{S,:}^{-1}\|\\
            &=\|\boldPi \mathbf{P}(\Tilde{\boldU}_{S,:}-\Hat{\boldU}_{S,:})\|_{2,\infty} \|\Hat{\boldU}_{S,:}^{-1}\|\\
            &\leq \|\boldPi\|_{\infty} \|\Hat{\boldU}_{S,:}-\mathbf{P}^{\top}\boldU_{S,:}sgn(\boldH)^\top\|_{2,\infty} \|\Hat{\boldU}_{S,:}^{-1}\|\\
            &\leq \|\Hat{\boldU}_{S,:}-\mathbf{P}^{\top} \boldU_{S,:} sgn(\boldH)^\top\|_{2,\infty} \|\Hat{\boldU}_{S,:}^{-1}\|\lesssim \kappa^2(\boldPi)\epsilon\|\Hat{\boldU}_{S,:}^{-1}\|
        \end{aligned}
    \end{align}
    The first and second inequality follows Proposition 5.6 in \cite{Cape2018}. The third inequality stems from the fact that all entries of $\boldPi$ lies in $[0,1]$ and the fourth inequality follows from Lemma \ref{lemma:errorbound1}. Then, denote $\Hat{\sigma}_i$ as the $i$th singular value of $\Hat{\boldU}_S$. \\
    
    For $\|\hat{\boldU}_S^{-1}\| $, by Weyl's inequality and Lemma \ref{lemma:errorbound1},
    \begin{align*}
        |\Hat{\sigma}_i-\sigma_i(\boldU_{S,:})|&\leq \|\Hat{\boldU}_{S,:}-\mathbf{P}^{\top}\boldU_{S,:}sgn(\boldH)^\top\|
        \lesssim \kappa^2(\boldPi)\epsilon,
    \end{align*}
    with probability at least $1-O(d^{-10})$, where $\hat{\sigma}_i$ is the $i$th largest singular vector of $\hat{\boldU}_{S,:}$. Hence,
    \begin{align}
        \label{eqn:Usigmabound}
        \begin{aligned}
            \Hat{\sigma}_K&\geq \frac{1}{\sqrt{\lambda_1(\boldPi^{\top}\boldPi)}} \left(1-O\left(\kappa(\boldPi^{\top}\boldPi) \sqrt{\lambda_1(\boldPi^{\top}\boldPi)}\epsilon \right) \right)\\
            \Hat{\sigma}_1 &\leq \frac{1}{\sqrt{\lambda_K(\boldPi^{\top}\boldPi)}} \left(1+O\left(\kappa(\boldPi^{\top}\boldPi)\sqrt{\lambda_K(\boldPi^{\top}\boldPi)}\epsilon \right) \right)
        \end{aligned}        
    \end{align}
    Then, combining Eq (\ref{eqn:Usigmabound}) and $\|\Hat{\boldU}_{S,:}^{-1}\| = 1/\Hat{\sigma}_K$
    \begin{equation*}
        \left\|\Hat{\boldU}_{S,:}^{-1}\right\| = \sqrt{\lambda_1(\boldPi^{\top}\boldPi)}\left(1-O\left(\kappa(\boldPi^{\top}\boldPi)\sqrt{\lambda_1(\boldPi^{\top}\boldPi)}\epsilon \right) \right) = O\left(\sqrt{\lambda_1(\boldPi^{\top}\boldPi)} \right).
    \end{equation*}
    The last equality stems from the Assumption in Theorem \ref{theorem:mmmbound}. Plugging the above inequalities into Equation (\ref{eqn:Ubound2all}) and we obtain
    \begin{align*}
        \left|\left|\boldU(\boldU^{\top}\Hat{\boldU}) \left(\Hat{\boldU}_{S,:}^{-1}-(\mathbf{P}^{\top}\boldU_{S,:}(\boldU^{\top}\Hat{\boldU}))^{-1} \right) \right| \right|_{2,\infty} &\lesssim \kappa(\boldPi)\epsilon \left\|\Hat{\boldU}_{S,:}^{-1}\right\| \\
        &\lesssim \sqrt{\lambda_1(\boldPi^{\top}\boldPi)}\kappa^2(\boldPi)\epsilon,
    \end{align*}
    with probability greater than $1-O(d^{-10})$.
\end{proof}

\indent Then, we are able to prove Theorem \ref{theorem:mmmbound}. Recall that $\Hat{\boldPi}=\Hat{\boldU}\Hat{\boldU}_{S,:}^{-1} $. Then,
\begin{align*}
    \left\|\Hat{\boldPi}-\boldPi \mathbf{P}\right\|_{2,\infty}&=\|\Hat{\boldU} \Hat{\boldU}_{S,:}^{-1} -\boldU\boldU_{S,:}^{-1}\mathbf{P} \|_{2,\infty}\\
    &\leq \left\|(\Hat{\boldU}-\boldU sgn(\boldH)^\top) \Hat{\boldU}_{S,:}^{-1}\right\|_{2,\infty}+ \left\|\boldU sgn(\boldH)^\top \left(\Hat{\boldU}_{S,:}^{-1}-(\mathbf{P}^{\top}\boldU_{S,:}sgn(\boldH)^\top)^{-1} \right)\right\|\\
    &\lesssim \left\|\Hat{\boldU}-\boldU sgn(\boldH)^\top\right\|_{2,\infty} \left\|\Hat{\boldU}_{S,:}^{-1}\right\| + \sigma_1(\boldPi) \kappa(\boldPi)^2\epsilon \\
    &\lesssim \epsilon \sigma_1(\boldPi) + \sigma_1(\boldPi) \kappa(\boldPi)^2\epsilon\\
    &\lesssim \sigma_1(\boldPi)\kappa(\boldPi)^2 \epsilon.
\end{align*}

\subsection{Proof of Corollary \ref{coro:bound w.r.t. alpha}}
First, with Lemma \ref{lemma:small_vals}, we have,
\begin{equation*}
    \left\|\hat{\boldU}\mathcal{O}_0-\boldU\right\|_{2,\infty} \lesssim (1+\alpha+\alpha^2)\frac{\kappa_{\max}^2}{\sigma_K(\cR)^2+\alpha\sigma_K(\cX)^2} \sqrt{\frac{\mu K}{N}}\zetaop(\boldG).
\end{equation*}
Combined with Theorem \ref{theorem:mmmbound}, we can see that $f(\alpha)$ is the only term concerned with $\alpha$ in the upper bound for the mixed membership estimation. 

By taking the derivative with respect to $\alpha$, one has,
\begin{align*}
f^{\prime}(\alpha) &= \frac{(2\alpha+1)\zetaop(\boldR)+ \alpha\zetaop(\boldX)+(\alpha^2+\alpha+1)\zetaop(\boldX)}{(\sigma_K(\cR)^2+\alpha\sigma_K(\cX)^2)}\\
&-\frac{(1+\alpha+\alpha^2)\sigma_K(\cX)^2}{(\sigma_K(\cR)^2+\alpha\sigma_K(\cX)^2)^2}.
\end{align*}

Then, when $f^{\prime}(0)<0$, there exists $\alpha>0$, s.t. $f(\alpha)<f(0)$. Substituing $\alpha=0$ into $f^{\prime}(\alpha)$, we can obtain that $f^{\prime}(0)<0$ is equivalent to
\begin{equation*}
    \left[\zetaop(\boldR)+\zetaop(\boldX)\right]\sigma_K(\cR)^2-\zetaop(\boldR)\sigma_K(\cX)^2 < 0.
\end{equation*}

Similarly, one can see that
\begin{equation*}
    f^{\prime}(1) = \frac{3\left[\zetaop(\boldR)\sigma_K(\cR)^2+2\zetaop(\boldX)\sigma_K(\cR)^2+\zetaop(\boldX)\sigma_K(\cX)^2\right]}{(\sigma_K(\cR)^2+\sigma_K(\cX)^2)^2} > 0.
\end{equation*}
Therefore, there exists $\alpha_0<1$ such that the $f(\alpha_0)<f(1)$.

\section{Proof of estimation guarantee for Theta}
\subsection{Some Key Lemmas}
In this section, we present some key technical lemmas to prove the estimation guarantee for $\hat{\boldTheta}$.
\begin{lemma}
    \label{lemma:sigma2}
    Instate the assumptions of Theorem \ref{theorem:leftbound}, and $\left\|\mathbf{V}\right\|_{2,\infty} \leq \sqrt{{\mu K}/{(J+W)}}$
    with probability at least $1-O(d^{-10})$, we have,
    \begin{equation*}
        \left\| \mathcal{O}_0^\top\hat{\boldSigma}^2\mathcal{O}_0 - \boldSigma^2 \right\| \lesssim \kappa_{\max}^2\frac{\zetaop(\boldG)^2}{\sigma_K(\cG)} + \kappa_{\max}^2\sqrt{\frac{\mu K}{N}} \zetaop(\boldG).
    \end{equation*}
\end{lemma}
\begin{proof}
    The following proof is adapted from \citet[Lemma 27]{Yan2021}. First, we can decompose the left hand side by,
    \begin{align*}
        \left\| \mathcal{O}_0^\top\hat{\boldSigma}^2\mathcal{O}_0 - \boldSigma^2 \right\| \leq \underbrace{\left\| \mathcal{O}_0^\top \hat{\boldSigma}^2\mathcal{O}_0 - \boldH^\top \hat{\boldSigma}^2\boldH \right\|}_{\beta_1} +\underbrace{\left\| \boldH^\top \hat{\boldSigma}^2\boldH - \boldU^\top \boldG\boldU \right\|}_{\beta_2} + \underbrace{\left\|\boldU^\top \boldG\boldU - \boldSigma^2 \right\|}_{\beta_3}.
    \end{align*}

    In the following, we shall bound each of the three terms.

    For $\beta_1$,
    \begin{align*}
        \beta_1 &\leq \left\|\left( \boldH - \mathcal{O}_0 \right)^\top \hat{\boldSigma}^2 \boldH\right\| + \left\| \mathcal{O}_0^\top \hat{\boldSigma}^2 \left( \boldH - \mathcal{O}_0 \right) \right\|\\
        &\leq \left\|\boldH - \mathcal{O}_0\right\| \left\|\hat{\boldSigma}\right\|^2 \left(\left\| \boldH \right\| + \left\| \mathcal{O}_0 \right\|\right)\\
        &\lesssim \sigma_1(\cG)^2\frac{\zetaop(\boldG)^2}{\sigma_K(\cG)^2}\\
        &\lesssim \frac{\kappa_{\max}^2\zetaop(\boldG)^2}{\sigma_K(\cG)},
    \end{align*}
    where we invoke Lemma \ref{lemma:spectralbound} in the penultimate inequality, and 
    \begin{equation*}
        \left\| \hat{\boldSigma}^2 \right\| \leq \left\| \boldSigma^2 \right\| + \left\| \boldG - \cG \right\| \lesssim \sigma_1(\cG)^2,
    \end{equation*}
    where we apply Lemma \ref{lemma:wholebound}.

    For $\beta_2$, following the same treatment as in \citet[Lemma 27]{Yan2021}, we have
    \begin{align*}
        \beta_2 &= \left\| \boldU^\top \hat{\boldU}_{\perp} \hat{\mathbf{\Lambda}}_{\perp} \hat{\boldU}_{\perp}^{\top} \boldU \right\|\\
        &\leq \left\| \boldU^\top\hat{\boldU}_{\perp} \right\|^2 \left\|\hat{\mathbf{\Lambda}}_{\perp}\right\|\\
        &= \left\| \hat{\boldU}\hat{\boldU}^\top - \boldU\boldU^\top \right\|^2 \left\|\hat{\mathbf{\Lambda}}_{\perp}\right\|\\
        &\lesssim \frac{\zetaop(\boldG)^3}{\sigma_K(\cG)^2},
    \end{align*}
    where by Weyl's inequality and Lemma \ref{lemma:wholebound},
    \begin{equation*}
        \left\|\hat{\mathbf{\Lambda}}_{\perp}\right\| \leq \lambda_{K+1} \left( \cG \right) + \left\| \boldG - \cG \right\| \lesssim \zetaop(\boldG).
    \end{equation*}

    For $\beta_3$, we can similarly decompose $\boldG$ as,
    \begin{align*}
        \boldG &= \od\left[ \left( \mathcal{L} + \boldE^L \right)\left(\mathcal{L} + \boldE^L\right)^\top \right] + \di\left(\cG\right) + \di\left( \boldG - \cG \right)\\
        &= \cG + \di\left(\boldG - \cG\right) + \boldE^R\cR^\top + \cR\boldE^{R\top} + \od\left(\boldE^R\boldE^{R\top}\right) - \di\left(\boldE^R\cR^\top + \cR\boldE^{R\top}\right)\\
        &+ \alpha \left[ \boldE^X\cX^\top + \cX\boldE^{X\top} + \od\left(\boldE^X\boldE^{X\top}\right) - \di\left(\boldE^X\cX^\top\right) \right].
    \end{align*}

    Then, we have,
    \begin{align*}
        \left\|\boldU^\top\boldG\boldU - \boldSigma^2\right\| &\leq \underbrace{\left\|\boldU^\top \left(\boldE^L\mathcal{L}^\top + \mathcal{L}\boldE^{R\top}\right) \boldU\right\|}_{J_1} + \underbrace{ \left\|\boldU^\top \od \left( \boldE^R \boldE^{R\top} \right) \boldU \right\|}_{J_2}\\ &+ \underbrace{\left\|\boldU^\top \di\left( \boldG-\cG \right)\boldU\right\|}_{J_3} + \underbrace{\left\|\boldU^\top \di\left( \boldE^R\cR^\top + \cR\boldE^{R\top} \right)\boldU\right\|}_{J_4}\\
        &+ \alpha\left[ \left\|\boldU^\top \od \left( \boldE^X \boldE^{X\top} \right) \boldU\right\| + \left\|\boldU^\top \di\left( \boldE^X\cX^\top + \cX\boldE^{X\top} \right)\boldU\right\|\right].
    \end{align*}
    In the sequel, we will bound the $J_1,\dots,J_4$ respectively, and the rest can be bounded similarly.

    For $J_1$, one has,
    \begin{equation*}
        \left\| \boldU^\top \boldE \mathcal{L}^\top \boldU \right\| = \left\| \boldU^\top \boldE \mathbf{V}\boldSigma \right\| \leq \sigma_1(\cG)\left\| \mathbf{U}^\top \mathbf{E}\mathbf{V} \right\|.
    \end{equation*}
    Then, by Eqn (C.11) in \citet{Yan2021}, with probability at least $1-O(d^{-10})$,
    \begin{align*}
    \left\| \boldU\boldE\mathbf{V} \right\| &\lesssim \max\{\sigma(\boldR), \sqrt{\alpha}\sigma(\boldX)\}\sqrt{K\log d} + \frac{\max\{B_R, \sqrt{\alpha} B_X\}\mu K\log d}{\sqrt{N\max\{J,W\}}} \\
    &\lesssim \max\{\sigma(\boldR), \sqrt{\alpha}\sigma(\boldX)\} \left( \sqrt{K\log d} + \frac{\mu K\sqrt{\log d}}{\sqrt[4]{N\max\{J,W\}}} \right),
    \end{align*}
    where we invoke the definition of $B_R$ and $B_X$ in the second inequality.

    Therefore,
    \begin{align*}
        J_1 &\leq 2\left\|\boldU^\top \boldE \mathcal{L}^\top \boldU\right\|\\
        &\leq 2\sigma_1(\cG) \left\| \boldU\boldE\mathbf{V} \right\|\\
        &\lesssim \sigma_1(\cG)\max\{\sigma(\boldR), \sqrt{\alpha}\sigma(\boldX)\} \left( \sqrt{K\log d} + \frac{\mu K\sqrt{\log d}}{\sqrt[4]{N\max\{J,W\}}} \right)\\
        &\lesssim \sqrt{\frac{\mu K}{N}}\zetaop(\boldG).
    \end{align*}

    For $J_2$, note that the same procedures in \citet[Lemma 27]{Yan2021} applies here. The only difference is that $\boldU$ is the left singular subspace of $\cG$ instead of $\cR$. But this does not change the proof. Therefore, with probability exceeding $1-O(d_R^{-10})$,
    \begin{equation*}
        J_2 \lesssim \sqrt{\frac{K}{N}}\zetaop(\boldR).
    \end{equation*}
    
    For $J_3$, we have
    \begin{equation*}
        J_3\leq \left\| \di\left(\boldG - \cG \right) \right\| \lesssim \kappa_{\max}^2 \sqrt{\frac{\mu K}{N}} \zetaop(\boldG),
    \end{equation*}
    where the second inequality comes from Lemma \ref{lemma:wholebound}.

    For $J_4$, we have,
    \begin{align*}
        J_4 \leq 2\max_{i\in [N]}\left| \sum_{j=1}^J E_{ij}^R\cR_{ij} \right|.
    \end{align*}
    By the same procedures as in \citet[Lemma 27]{Yan2021}, we have,
    \begin{equation*}
        J_4 \lesssim \frac{\sqrt{\mu K}}{N} \zetaop(\boldR).
    \end{equation*}
    
    Combining all these terms, we have,
    \begin{equation*}
        \beta_3 \lesssim \sqrt{\frac{\mu K}{N}}\zetaop(\boldG) + \sqrt{\frac{K}{N}} \zetaop(\boldG) + \kappa_{\max}^2\sqrt{\frac{\mu K}{N}}\zetaop(\boldG) + \frac{\sqrt{\mu K}}{N}\zetaop(\boldG) \lesssim \kappa_{\max}^2\sqrt{\frac{\mu K}{N}} \zetaop(\boldG),
    \end{equation*}
    as long as $N\gtrsim \mu K$.

    Therefore, combining $\beta_1$, $\beta_2$, $\beta_3$, we arrive at,
    \begin{align*}
        \left\| \mathcal{O}_0^\top\hat{\boldSigma}^2\mathcal{O}_0 - \boldSigma^2 \right\| &\lesssim \beta_1 + \beta_2 + \beta_3\\
        &\lesssim \kappa_{\max}^2\frac{\zetaop(\boldG)^2}{\sigma_K(\cG)} + \frac{\zetaop(\boldG)^3}{\sigma_K(\cG)^2} + \kappa_{\max}^2\sqrt{\frac{\mu K}{N}} \zetaop(\boldG)\\
        &\lesssim \kappa_{\max}^2\frac{\zetaop(\boldG)^2}{\sigma_K(\cG)} + \kappa_{\max}^2\sqrt{\frac{\mu K}{N}} \zetaop(\boldG),
    \end{align*}
    given $\zetaop(\boldG)\ll \sigma_K(\cG)$.
    
\end{proof}

\begin{lemma}
    \label{lemma:sigma}
    Instate the assumptions of Lemma \ref{lemma:sigma2}, with probability at least $1-O(d^{-10})$, we have,
    \begin{equation*}
        \left\| \mathcal{O}_0^\top \hat{\boldSigma}\mathcal{O}_0 - \boldSigma \right\| \lesssim \frac{\kappa_{\max}^2\sqrt{K}}{\sigma_K(\cG)^{1/2}} \left( \frac{\zetaop(\boldG)^2}{\sigma_K(\cG)} + \sqrt{\frac{\mu K}{N}} \zetaop(\boldG) \right).
    \end{equation*}
\end{lemma}
\begin{proof}
    In the following, we denote $\sigma_j=\sigma_j(\mathcal{L})$, $\hat{\sigma}_j= \sigma_j(\mathbf{L})$ for simplicity. Note that we have the following equality,
    \begin{equation}
        \label{eqn:eq_sigma}
        \mathcal{O}_0^\top\hat{\boldSigma}^2\mathcal{O}_0 - \boldSigma^2 = \left(\mathcal{O}_0^\top\hat{\boldSigma}\mathcal{O}_0 - \boldSigma\right)\mathcal{O}_0^\top\hat{\boldSigma}\mathcal{O}_0 + \boldSigma \left(\mathcal{O}_0^\top\hat{\boldSigma}\mathcal{O}_0 - \boldSigma\right).
    \end{equation}
    For the following proof, for notational simplicity, we denote $\boldA_2=\mathcal{O}_0^\top\hat{\boldSigma}^2\mathcal{O}_0 - \boldSigma^2$, $\boldA_1 = \mathcal{O}_0^\top\hat{\boldSigma}\mathcal{O}_0 - \boldSigma$. Additionally, define $\tilde{\boldA}_2=\boldA_2\mathcal{O}_0^\top$, $\tilde{\boldA}_1 = \boldA_1\mathcal{O}_0^\top$. By right multiplying $\mathcal{O}_0$ on both sides of Equation (\ref{eqn:eq_sigma}), we have,
    \begin{equation*}
        \tilde{\boldA}_2 = \tilde{\boldA}_1\hat{\boldSigma} + \boldSigma\tilde{\boldA}_1.
    \end{equation*}
    Then, this equality implies, for every $(i,j)\in [K]\times [K]$,
    \begin{equation*}
        \tilde{a}^{(2)}_{i,j} = \hat{\sigma}_j\tilde{a}^{(1)}_{i,j}+ \sigma_i\tilde{a}^{(1)}_{i,j}.
    \end{equation*}
    Hence,
    \begin{equation*}
        \tilde{a}^{(1)}_{i,j} = \frac{\tilde{a}^{(2)}_{i,j}}{\hat{\sigma}_j+\sigma_i}.
    \end{equation*}

    Then,
    \begin{align*}
        \left|\left( \tilde{\boldA}_1 \mathbf x \right)_i\right| &=  \left|\sum_{j=1}^K\tilde{a}^{(1)}_{i,j}\mathbf{x}_{j}\right|\\
        &= \left|\sum_{j=1}^K\frac{\tilde{a}_{i,j}^{(2)}\mathbf{x}_j}{\sigma_i + \hat{\sigma}_j} \right|\\
        &\leq \sqrt{\sum_{j=1}^K\frac{\left(\tilde{a}_{i,j}^{(2)}\right)^2}{(\sigma_i + \hat{\sigma}_j)^2}}\\
        &\leq \frac{1}{\sigma_K+\hat{\sigma}_K}\sqrt{\sum_{j=1}^K \left(\tilde{a}_{i,j}^{(2)}\right)^2},
    \end{align*}
    for all $\|\mathbf{x}\|_2=1$, where the penultimate inequality comes from Cauchy-Schwarz inequality.

    Therefore,
    \begin{align*}
        \left\|\tilde{\boldA}_1\right\|^2 &= \sup_{\|\mathbf x\|=1} \left\|\tilde{\boldA}_1 \mathbf{x}\right\|_2^2\\
        &\leq \frac{1}{(\sigma_K+\hat{\sigma}_K)^2}\sum_{i=1}^K\sum_{j=1}^K  \left(\tilde{a}_{i,j}^{(2)}\right)^2\\
        &= \frac{\left\| \tilde{\boldA}_2 \right\|_F^2}{(\sigma_K+\hat{\sigma}_K)^2}\\
        &\leq \frac{K\left\| \tilde{\boldA}_2 \right\|^2}{(\sigma_K+\hat{\sigma}_K)^2}.
    \end{align*}
    Recall that $\|\tilde{\boldA}_1\|=\|\boldA_1\mathcal{O}_0^\top\|=\|\boldA_1\|$, and the similar reason applies to $\|\tilde{\boldA}_2\|=\|\boldA_2\|$. Hence, with probability exceeding $1-O(d^{-10})$
    \begin{align*}
        \left\| \mathcal{O}_0^\top \hat{\boldSigma}\mathcal{O}_0 - \boldSigma \right\| &\leq \frac{\sqrt{K}\left\|\mathcal{O}_0^\top \hat{\boldSigma}^2\mathcal{O}_0 - \boldSigma^2\right\|}{\hat{\sigma}_K+\sigma_K}\\
        &\lesssim \frac{\sqrt{K}}{\sigma_K(\cG)^{1/2}} \left( \kappa_{\max}^2\frac{\zetaop(\boldG)^2}{\sigma_K(\cG)} + \kappa_{\max}^2\sqrt{\frac{\mu K}{N}} \zetaop(\boldG) \right).
    \end{align*}
\end{proof}

\subsection{Proof of Theorem \ref{theorem:VL}}
\label{sec:proof_vl}

\textbf{Proof of estimation guarantee for} $\mathbf{V}$:
The estimate of $\mathbf{V}$ by heteroskedastic principal component analysis is,
\begin{equation*}
    \hat{\mathbf{V}} = \boldL^\top\hat{\boldU} \hat{\boldSigma}^{-1}.
\end{equation*}

Then, the row-wise estimation error of $\mathbf{V}$ can be written as,
\begin{align*}
    \left\| \boldL^\top \hat{\boldU} \hat{\boldSigma}^{-1} \mathcal{O}_0 - \mathbf{V} \right\|_{2,\infty} &= \left\| \mathbf{V}\boldSigma \boldU^\top \hat{\boldU} \hat{\boldSigma}^{-1} \mathcal{O}_0 + \boldE^\top \hat{\boldU} \hat{\boldSigma}^{-1}\mathcal{O}_0 - \mathbf{V} \right\|_{2,\infty}\\
    &\leq \underbrace{\left\| \mathbf{V}\left( \boldSigma \boldH^\top\hat{\boldSigma}^{-1}\mathcal{O}_0 - \mathbf{I}\right) \right\|_{2,\infty}}_{\beta_1} + \underbrace{\left\| \boldE^\top \hat{\boldU} \right\|_{2,\infty}}_{\beta_2}\left\|\hat{\boldSigma}^{-1}\right\|,
\end{align*}
where the first equality results from $\boldL = \boldU\boldSigma\mathbf{V}^\top + \boldE^L$. Then, we bound $\beta_1$ and $\beta_2$ respectively.

For $\beta_1$, we have
\begin{align*}
    \beta_1&\leq \left\| \mathbf{V} \right\|_{2,\infty} \left\|\boldSigma \boldH^\top\hat{\boldSigma}^{-1}\mathcal{O}_0 - \mathbf{I}_K\right\|\\
    &\leq \sqrt{\frac{\mu K}{J+W}}\left\| \boldSigma\boldH^\top\hat{\boldSigma}^{-1} \mathcal{O}_0 - \mathbf{I}_K \right\|\\
    &\leq \sqrt{\frac{\mu K}{J+W}} \left\| \boldSigma \boldH^\top - \mathcal{O}_0^\top \hat{\boldSigma}\right\| \left\| \hat{\boldSigma}^{-1} \mathcal{O}_0 \right\|\\
    &\leq \sqrt{\frac{\mu K}{J+W}} \left\| \hat{\boldSigma}^{-1} \right\|
    \left( \left\|\boldSigma\boldH^\top - \boldSigma\mathcal{O}_0^\top\right\| + \left\| \boldSigma\mathcal{O}_0^\top - \mathcal{O}_0^\top\hat{\boldSigma}\mathcal{O}_0\mathcal{O}_0^\top \right\| \right)\\
    &\lesssim \frac{1}{\sigma_K(\cG)^{1/2}} \sqrt{\frac{\mu K}{J+W}}  \left(\sigma_1(\cG)^{1/2}\left\|\boldH - \mathcal{O}_0\right\| + \left\|\boldSigma -\mathcal{O}_0^\top\hat{\boldSigma}\mathcal{O}_0\right\| \right)\\
    &\lesssim \frac{1}{\sigma_K(\cG)^{1/2}} \sqrt{\frac{\mu K}{J+W}} \left( \kappa_{\max}\frac{\zetaop(\boldG)^2}{\sigma_K(\cG)^{3/2}} + \frac{\sqrt{K}}{\sigma_K(\cG)^{1/2}} \left( \kappa_{\max}^2\frac{\zetaop(\boldG)^2}{\sigma_K(\cG)} + \kappa_{\max}^2\sqrt{\frac{\mu K}{N}} \zetaop(\boldG) \right) \right)\\
    &\lesssim \frac{\kappa_{\max}^2}{\sigma_K(\cG)}\frac{\sqrt{\mu} K}{\sqrt{J+W}}\left(\frac{\zetaop(\boldG)^2}{\sigma_K(\cG)} + \sqrt{\frac{\mu K}{N}}\zetaop(\boldG)\right),
\end{align*}
where the third-to-last inequality comes from \citet{Yan2021} Equation (F.31), and the last inequality invokes $\sigma_K(\hat{\boldSigma})^2 \geq {\sigma_K(\cG)}/{2}$.

For $\beta_2$, 
\begin{align*}
    \beta_2 &\leq \left\| \mathbf{E}^\top\hat{\boldU}\mathcal{O}_0 \right\|_{2,\infty} \left\|\mathcal{O}_0^\top\right\|\\
    &\leq \underbrace{\left\|\mathbf{E}^\top \mathbf{U}\right\|_{2,\infty}}_{\beta_{2,1}} + \underbrace{\left\|\mathbf{E}^{\top}\left(\hat{\boldU}\mathcal{O}_0-\boldU \right)\right\|_{2,\infty}}_{\beta_{2,2}}.
\end{align*}

Then, for $\beta_{2,1}$, by \citet[Lemma 5]{Yan2021}, we have
\begin{align*}
    \left\|\mathbf{E}^\top \mathbf{U} \right\|_{2,\infty} &\lesssim \sigma(\boldR) \left\| \mathbf{U} \right\| _F\sqrt{\log d_R} + B_R\left\|\mathbf{U}\right\|_{2,\infty}\log d_R\\
    &+ \sqrt{\alpha}\left( \sigma(\boldX) \left\| \mathbf{U} \right\| _F\sqrt{\log d_X} + B_X\left\|\mathbf{U}\right\|_{2,\infty}\log d_X \right)\\
    &\leq \sigma(\boldR) \sqrt{K\log d_R} + B_R\sqrt{\frac{\mu K}{N}}\log d_R\\
    &+\sqrt{\alpha}\left(\sigma(\boldX) \sqrt{K\log d_X} + B_X\sqrt{\frac{\mu K}{N}}\log d_X \right)\\
    &\lesssim \sigma(\boldR)\sqrt{K\log d_R} + \sqrt{\alpha} \sigma(\boldX)\sqrt{K\log d_X},
\end{align*}
where the last inequality comes from $B_R\lesssim\sigma(\boldR) \sqrt{N/\mu\log d_R}$ and similarly for $B_X$.

For $\beta_{2,2}$,
\begin{align*}
    \beta_{2,2} &\leq \|\mathbf{E}\|_{\infty}\left\|\hat{\boldU}\mathcal{O}_0-\boldU\right\|_{2,\infty}\\
    &\leq \left(B_R+\sqrt{\alpha}B_X\right)\left\|\hat{\boldU}\mathcal{O}_0-\boldU\right\|_{2,\infty}
\end{align*}

Combining $\beta_{2,1}$ and $\beta_{2,2}$, by Theorem \ref{theorem:leftbound}, we have
\begin{align*}
    \beta_2 &\lesssim \sigma(\boldR)\sqrt{K\log d_R} + B_R\left\|\hat{\boldU}\mathcal{O}_0-\boldU\right\|_{2,\infty} + \sqrt{\alpha}\left[ \sigma(\boldX)\sqrt{K\log d_X} + B_X\left\|\hat{\boldU}\mathcal{O}_0-\boldU\right\|_{2,\infty} \right] \\
    &\lesssim \sigma(\boldR)\sqrt{K\log d_R} + \sigma(\boldR)\sqrt{\frac{N}{\mu \log d_R}}\frac{\kappa_{\max}^2}{\sigma_K(\cG)}\sqrt{\frac{\mu K}{N}} \zetaop(\boldG) (1+\alpha+\alpha^2)\\
    &+ \sqrt{\alpha}\left[ \sigma(\boldX)\sqrt{K\log d_X} + \sigma(\boldX)\sqrt{\frac{N}{\mu \log d_X}}\frac{\kappa_{\max}^2}{\sigma_K(\cG)}\sqrt{\frac{\mu K}{N}} \zetaop(\boldG) (1+\alpha+\alpha^2) \right]\\
    &\lesssim \sigma(\boldR)\sqrt{K\log d_R}\left(1+\frac{1+\alpha + \alpha^2}{\log d_R}\right) + \sigma(\boldX)\sqrt{K\log d_R}\left(1+\frac{1+\alpha + \alpha^2}{\log d_X}\right)\\
    &\lesssim \sigma(\boldR) \sqrt{K\log d_R} + \sqrt{\alpha}\sigma(\boldX) \sqrt{K\log d_X},
\end{align*}
where we invoke the condition that $B_R\lesssim \sigma(\boldR)\sqrt{N/(\mu\log d_R)}$ and $\zetaop(\boldG)\lesssim \sigma_K(\cG)/\kappa_{\max}^2$, and the same for the $\boldX$ terms in the penultimate inequality. The last inequality holds when $d_R$ and $d_X$ is large enough.

Combining $\beta_1$ and $\beta_2$, we have,
\begin{align*}
    \left\| \boldL^\top \hat{\boldU} \hat{\boldSigma}^{-1} \mathcal{O}_0 - \mathbf{V} \right\|_{2,\infty} &\lesssim \frac{\kappa_{\max}^2}{\sigma_K(\cG)}\frac{\sqrt{\mu} K}{\sqrt{J+W}}\left(\frac{\zetaop(\boldG)^2}{\sigma_K(\cG)} + \sqrt{\frac{\mu K}{N}}\zetaop(\boldG)\right)\\
    &+ \frac{\sigma(\boldR)\sqrt{K\log d_R} + \sqrt{\alpha}\sigma(\boldX)\sqrt{K\log d_X}}{\sigma_K(\cG)^{1/2}}\\
    &\lesssim \frac{\kappa_{\max}^2}{\sigma_K(\cG)}\frac{\sqrt{\mu} K}{\sqrt{J+W}}\left(\frac{\zetaop(\boldG)^2}{\sigma_K(\cG)} + \sqrt{\frac{\mu K}{N}}\zetaop(\boldG)\right)\\
    &+ \frac{1}{\sigma_K(\cG)^{1/2}}\sqrt{\frac{\mu K}{N}}\left( \frac{\zetaop(\boldR)}{\sigma_K(\cR)} + \frac{\sqrt{\alpha}\zetaop(\boldX)}{\sigma_K(\cX)} \right),
\end{align*}
where the last inequality comes from $\sigma(\boldR)\sigma_1(\cR)\sqrt{N\log d_R}\leq \zetaop(\boldR)$, and similarly for $\boldX$.

\textbf{Proof of estimation guarantee for} $\cL$.
    \begin{align*}
        \left\| \hat{\boldU}\hat{\boldU}^\top \boldL - \cL \right\|_{\infty} &= \left\| \hat{\boldU}\hat{\boldSigma}\hat{\boldSigma}^{-1}\hat{\boldU}^\top\boldL - \boldU\boldSigma\mathbf{V}^\top \right\|_{\infty}\\
        &= \underbrace{\left\|\hat{\boldU}  \hat{\boldSigma} \left( \boldL\hat{\boldU}\hat{\boldSigma}^{-1} - \mathbf{V}\mathcal{O}_0^\top \right)^\top \right\|_{\infty}}_{\beta_1} + \underbrace{\left\| \left( \hat{\boldU} - \boldU \mathcal{O}_0^\top \right) \hat{\boldSigma} \mathcal{O}_0\mathbf{V}^\top \right\|_{\infty}}_{\beta_2} + \\
        &+ \underbrace{\left\|\boldU \left( \mathcal{O}_0^\top \hat{\boldSigma}\mathcal{O}_0 - \boldSigma \right)\mathbf{V}^\top \right\|_{\infty}}_{\beta_3}.
    \end{align*}
    In the sequel, we shall bound $\beta_1$, $\beta_2$, and $\beta_3$ individually.

    For $\beta_1$,
    \begin{align*}
        \beta_1 &\leq \left\| \hat{\boldU} \right\|_{2,\infty} \left\|\hat{\boldSigma}\right\| \left\|  \boldL\hat{\boldU}\hat{\boldSigma}^{-1} - \mathbf{V}\mathcal{O}_0^\top \right\|_{2,\infty}\\
        &\leq \left\| \hat{\boldU} \right\|_{2,\infty} \left\|\hat{\boldSigma}\right\| \left\|  \boldL\hat{\boldU}\hat{\boldSigma}^{-1} \mathcal{O}_0 - \mathbf{V} \right\|_{2,\infty}\\
        &\lesssim \sigma_1(\cG)^{1/2}\sqrt{\frac{\mu K}{N}}\left[\frac{\kappa_{\max}^2}{\sigma_K(\cG)}\frac{\sqrt{\mu} K}{\sqrt{J+W}}\left(\frac{\zetaop(\boldG)^2}{\sigma_K(\cG)} + \sqrt{\frac{\mu K}{N}}\zetaop(\boldG)\right)\right.\\
        &\left.+ \frac{1}{\sigma_K(\cG)^{1/2}}\sqrt{\frac{\mu K}{N}}\left( \frac{\zetaop(\boldR)}{\sigma_K(\cR)} + \frac{\sqrt{\alpha}\zetaop(\boldX)}{\sigma_K(\cX)} \right)\right]\\
        &\lesssim \frac{\kappa_{\max}^3}{\sigma_K(\cG)^{1/2}} \frac{\mu K^{3/2}}{\sqrt{N(J+W)}} \left( \frac{\zetaop(\boldG)^2}{\sigma_K(\cG)} + \sqrt{\frac{\mu K}{N}}\zetaop(\boldG) \right)\\
        &+ \kappa_{\max} \frac{\mu K}{N}\left(\frac{\zetaop(\boldR)}{\sigma_K(\cR)} + \frac{\sqrt{\alpha}\zetaop(\boldX)}{\sigma_K(\cX)}\right),
    \end{align*}
    where in the penultimate inequality, we employ $\|\hat{\boldU}\|_{2,\infty} \leq \|\boldU\|_{2,\infty} + \|\hat{\boldU}\mathcal{O}_0 - \boldU\|_{2,\infty} \lesssim \sqrt{\mu K/N} $, and similarly ${\sigma}_{1}(\boldG)\lesssim \sigma_1(\cG)$.

    For $\beta_2$,
    \begin{align*}
        \beta_2 &\leq \left\|\hat{\boldU} - \boldU \mathcal{O}_0^\top\right\|_{2,\infty}\left\| \hat{\boldSigma} \right\| \left\| \mathbf{V} \right\|_{2,\infty}\\
        &\leq \left\| \hat{\boldU}\mathcal{O}_0 - \boldU \right\|_{2,\infty} \left\| \hat{\boldSigma} \right\| \left\| \mathbf{V} \right\|_{2,\infty}\\
        &\lesssim \sigma_1(\cG)^{1/2}\sqrt{\frac{\mu K}{J+W}} (1+ \alpha + \alpha^2) \frac{\kappa_{\max}^2}{\sigma_K(\cG)}\sqrt{\frac{\mu K}{N}} \zetaop(\boldG)\\
        &\lesssim \left(1 + \alpha + \alpha^2 \right)\kappa_{\max}^3\frac{\mu K}{\sqrt{N(J+W)}} \frac{\zetaop(\boldG)}{\sigma_K(\cG)^{1/2}}.
    \end{align*}

    For $\beta_3$,
    \begin{align*}
        \beta_3 &\leq \left\| \boldU \right\|_{2,\infty}\left\| \mathcal{O}_0^\top \hat{\boldSigma}\mathcal{O}_0 - \boldSigma \right\| \left\| \mathbf{V} \right\|_{2,\infty}\\
        &\lesssim \frac{\mu K}{\sqrt{N(J+W)}}\frac{\sqrt{K}}{\sigma_K(\cG)^{1/2}} \left( \kappa_{\max}^2\frac{\zetaop(\boldG)^2}{\sigma_K(\cG)} + \kappa_{\max}^2\sqrt{\frac{\mu K}{N}} \zetaop(\boldG) \right).
    \end{align*}

    Combining $\beta_1$, $\beta_2$, $\beta_3$, we can obtain,
    \begin{align*}
        \left\| \hat{\boldU}\hat{\boldU}^{\top} \boldL - \cL \right\|_{\infty} &\leq \beta_1 + \beta_2 + \beta_3\\
        &\lesssim (1+\alpha+\alpha^2)\frac{\kappa_{\max}^3}{\sigma_K(\cG)^{1/2}} \frac{\mu K^{3/2}}{\sqrt{N(J+W)}} \left( \frac{\zetaop(\boldG)^2}{\sigma_K(\cG)} + \sqrt{\frac{\mu K}{N}}\zetaop(\boldG) \right)\\
        &+ \kappa_{\max} \frac{\mu K}{N}\left(\frac{\zetaop(\boldR)}{\sigma_K(\cR)} + \frac{\sqrt{\alpha}\zetaop(\boldX)}{\sigma_K(\cX)}\right).
    \end{align*}
\subsection{Proof of Theorem \ref{theorem:Theta}}
\label{sec:proof_theta}
    \begin{align*}
        &\max\left\{\left\| \hat{\boldTheta}\mathbf{P}^\top - \boldTheta \right\|_{\infty}, \sqrt{\alpha}\left\| \hat{\boldM}\mathbf{P}^\top - \boldM \right\|_{\infty} \right\}\\ &= \left\|\mathbf{P}\hat{\mathbf{U}}_{S,:}\hat{\boldSigma} \hat{\mathbf{V}}^\top-\mathbf{U}_{S,:} \boldSigma \mathbf{V}^\top\right\|_{\infty}\\
        &\leq \left\| \left( \mathbf{P}\hat{\mathbf{U}}_{S,:}-\mathbf{U}_{S,:} \mathcal{O}_0^\top \right)\hat{\boldSigma}\hat{\mathbf{V}}^\top \right\|_{\infty} + \left\| \left(\mathbf{U}_{S,:}\mathcal{O}_0^\top - \hat{\mathbf{U}}_{S,:}\right) \hat{\boldSigma}\hat{\mathbf{V}}^\top \right\|_{\infty}\\ 
        &+ \left\|\hat{\boldU}_{S,:}\hat{\boldSigma} \hat{\mathbf{V}}^\top - \mathbf{U}_{S,:}\boldSigma \mathbf{V}^\top \right\|_{\infty}\\
        &\leq \left\|\mathbf{P}\hat{\boldU}_{S,:}\mathcal{O}_0-{\boldU}_{S,:} \right\|_{2,\infty}\left\| \hat{\boldSigma} \right\| \left\| \hat{\mathbf{V}} \right\|_{2,\infty}+\left\| \hat{\boldU}\mathcal{O}_0 - \boldU \right\|_{2,\infty}\left\|\hat{\boldSigma}\right\| \left\| \hat{\mathbf{V}} \right\|_{2,\infty}\\ 
        &+ \left\|\hat{\boldU} \hat{\boldSigma} \hat{\mathbf{V}}^\top - \boldU\boldSigma \mathbf{V}^\top \right\|_{\infty}\\
        &\stackrel{(i)}{\lesssim} \left[ \kappa(\boldPi)^2 \epsilon + \epsilon \right] \frac{1}{\sigma_K(\cG)^{1/2}} \left[ \sqrt{\frac{\mu K}{J+W}} +  \frac{\kappa_{\max}^2}{\sigma_K(\cG)}\frac{\sqrt{\mu} K}{\sqrt{J+W}}\left(\frac{\zetaop(\boldG)^2}{\sigma_K(\cG)} + \sqrt{\frac{\mu K}{N}}\zetaop(\boldG)\right)\right.\\
        &\left.+ \frac{1}{\sigma_K(\cG)^{1/2}}\sqrt{\frac{\mu K}{N}}\left( \frac{\zetaop(\boldR)}{\sigma_K(\cR)} + \frac{\sqrt{\alpha}\zetaop(\boldX)}{\sigma_K(\cX)} \right)\right]\\
        &+ \frac{\kappa_{\max}^3}{\sigma_K(\cG)^{1/2}} \frac{\mu K^{3/2}}{\sqrt{N(J+W)}} \left( \frac{\zetaop(\boldG)^2}{\sigma_K(\cG)} + \sqrt{\frac{\mu K}{N}}\zetaop(\boldG) \right)\\
        &+ \kappa_{\max} \frac{\mu K}{N}\left(\frac{\zetaop(\boldR)}{\sigma_K(\cR)} + \frac{\sqrt{\alpha}\zetaop(\boldX)}{\sigma_K(\cX)}\right)\\
        &\lesssim (1+\alpha+\alpha^2) \frac{\kappa(\boldPi)^2 \kappa_{\max}^4}{\sigma_K(\cG)^{1/2}} \frac{\mu K^{3/2}}{\sqrt{N(J+W)}} \left( \frac{\zetaop(\boldG)^2}{\sigma_K(\cG)} + \sqrt{\frac{\mu K}{N}}\zetaop(\boldG) \right)\\
        &+ (1+\alpha+\alpha^2)\kappa(\boldPi)^2\kappa_{\max}^2 \frac{\mu K}{N}\left(\frac{\zetaop(\boldR)}{\sigma_K(\cR)} + \frac{\sqrt{\alpha}\zetaop(\boldX)}{\sigma_K(\cX)}\right),
    \end{align*}
    with probability at least $1-O(d^{-10})$, where $(i)$ comes from Lemma \ref{lemma:errorbound1} and Theorem \ref{theorem:VL}, and $\|\hat{\mathbf{V}}\|_{2,\infty}\leq \|\mathbf{V}\|_{2,\infty} + \|\hat{\mathbf{V}}\mathcal{O}_0-\mathbf{V}\|_{2,\infty} $, and the last inequality stems from $\zetaop(\boldG)\lesssim \sigma_K(\cG)$.

\section{Proof of auxiliary Lemmas}

\subsection{Proof of Lemma \ref{lemma:wholebound}}
\label{sec:lem_wholebound}
For completeness, we present Lemma 23 from \citet{Yan2021}.
\begin{lemma}[\citet{Yan2021} Lemma 23]
    \label{lemma:individualbound}
    Under Assumption \ref{assump:noise} and \ref{assump:information}, suppose the number of iterations $t_0\geq \log\left({\sigma_1(\cR)^2}/{\zeta_{op}(\boldR)} \right)$
    \begin{equation*}
        \|\boldG_R-\cG_R\|\lesssim \zeta_{\text{op}}(\boldR),
    \end{equation*}
    \begin{equation*}
        \|\mathcal{P}_{\text{diag}}(\boldG_R-\mathcal{G}_R)\|\lesssim \kappa_{\max}^2\sqrt{\frac{\mu K}{N}}\zeta_{\text{op}}(\boldR),
    \end{equation*}
    with probability at least $1-O(d_R^{-10})$, where $\boldG_R:= \boldG_R^{t_0}$ and $\cG_R=\mathcal{R}\mathcal{R}^{\top}$. Similarly, suppose the number of iterations $t_0\geq \log\left({\sigma_1(\cX)^2}/{\zeta_{op}(\boldX)} \right)$, with probability at least $1-O(d_X^{-10})$, and
    \begin{equation*}
        \|\boldG_X-\cG_X\|\lesssim \zeta_{\text{op}}(\boldX),
    \end{equation*}
    \begin{equation*}
        \|\mathcal{P}_{\text{diag}}(\boldG_X-\mathcal{X}_R)\|\lesssim \kappa_{\max}^2\sqrt{\frac{\mu K}{N}}\zeta_{\text{op}}(\boldX),
    \end{equation*}
    where $\boldG_X:=\boldG_X^{t_0}$ and $\cG_X=\mathcal{X}\mathcal{X}^{\top}$. $\zeta_{\text{op}}(\boldR)$ and $\zeta_{\text{op}}(\boldX)$ are defined as Lemma \ref{lemma:wholebound}. $\sigma(\boldR)$ and $\sigma(\boldX)$ are the upper bound of the variance of noise.
\end{lemma}
\indent Then, we are able to establish Lemma \ref{lemma:wholebound}.

    From Lemma \ref{lemma:individualbound} and triangle inequality, the lemma holds if the above inequalities hold simultaneously. Denote $\mathbb{P}(\boldR)$ as the probability of the bound on $\boldR$ holds and $\mathbb{P}(\boldX)$ the probability of the bound on $\boldX$ holds. One has,
    \begin{align*}
        \mathbb{P}(\boldR)\geq 1-O(d_R^{-10}),\\
        \mathbb{P}(\boldX)\geq 1-O(d_X^{-10}),
    \end{align*}
    Then,
    \begin{align*}
        \mathbb{P}(\boldR\cap \boldX)&= 1-\mathbb{P}(\boldR^C\cup \boldX^C) \\
        &\geq 1-\mathbb{P}(\boldR^C)-\mathbb{P}(\boldX^C)\\
        &\geq 1-O(d_R^{-10}+d_X^{-10}).
    \end{align*}
    Recall the definition of $d_R$ and $d_X$, and $d=\max\{N, J, W \}$. Therefore, with probability at least $1-O(d^{-10})$, 
    \begin{align*}
        \|\boldG-\mathcal{G}\| &\leq \|\boldG_R-\cG_R\|+\alpha\|\boldG_X-\cG_X\|\\
        &\lesssim \zeta_{\text{op}}(\boldR)+ \alpha \zeta_{\text{op}}(\boldX) .
    \end{align*}

    Similarly, we have
    \begin{equation*}
        \left\|\mathcal{P}_{\text{diag}}\left( \boldG - \cG \right)\right\| \lesssim \kappa_{\max}^2\sqrt{\frac{\mu K}{N}}(\zetaop(\boldR)+\alpha\zetaop(\boldX)),
    \end{equation*}
    with probability exceeding $1-O(d^{-10})$.

\subsection{Proof of Lemma \ref{lemma:combine}}
\label{sec:lem_com}
\begin{proof}
    By Lemma \ref{lemma:small_vals}, we have
    \begin{equation*}
        \sigma_K(\cR)^2 + \alpha \sigma_k(\cX)^2 \ll \sigma_K(\cG).
    \end{equation*}
    Additionally, since $\kappa(\cR)\geq 1$, $\kappa(\cX) \geq 1$, we have
    \begin{equation}
        \label{eqn:small_vals}
        \frac{\sigma_K(\cR)^2}{\kappa(\cR)^2} + \alpha\frac{\sigma_{K}(\cX)^2}{\kappa(\cX)^2} \leq \sigma_K(\cG).
    \end{equation}
    From \citet[Fact 1]{caiSubspaceEstimationUnbalanced2020}, we have
    \begin{equation}
        \label{eqn:err_vals}
        \zeta_{op}(\boldR) +\alpha\zeta_{op}(\boldX) \ll \frac{\sigma_K(\cR)^2}{\kappa(\cR)^2} + \alpha\frac{\sigma_{K}(\cX)^2}{\kappa(\cX)^2}.
    \end{equation}
    Combine (\ref{eqn:small_vals}) and (\ref{eqn:err_vals}) yields,
    \begin{equation*}
        \zeta_{op}(\boldR) +\alpha\zeta_{op}(\boldX) \ll \sigma_K(\cG).
    \end{equation*}
\end{proof}

\subsection{Proof of Lemma \ref{lemma:spectralbound}}
\label{sec:lem_spec}
\begin{proof}
    Following the Spectral Bound part of Theorem \ref{theorem:leftbound}, we have
    \begin{align*}
        \left\|\hat{\boldU}\boldH-\boldU\right\| &= \left\|\hat{\boldU}\hat{\boldU}^\top \boldU - \boldU\boldU^\top\boldU\right\|\\
        &\leq \left\|\hat{\boldU}\hat{\boldU}^\top - \boldU\boldU^\top\right\|\|\boldU\|\\
        &\leq \left\|\hat{\boldU}\mathcal{O}_0-\boldU\right\|\\
        &\lesssim \frac{\zeta_{\text{op}}(\boldR)+\alpha\zeta_{\text{op}}(\boldX)}{\sigma_K(\cG)},
    \end{align*}
    where the second to last inequality comes from \citet[Lemma 2.6]{Chen2021} and the last from (\ref{eqn:spectralbound}).

    Since $\boldU$ and $\hat{\boldU}$ are both orthogonal matrices, $\boldH$ can be written as
    \begin{equation*}
        \boldH = \mathbf{U}_l\cos\mathbf{\Psi}\mathbf{U}_r^\top,
    \end{equation*}
    where $\mathbf{\Psi}$ is the principal angles between $\hat{\boldU}$ and $\boldU$, and $\boldU_l$, $\boldU_r$ are orthonormal matrices. Since $\mathcal{O}_0=\text{sgn}(\boldH)=\mathbf{U}_l\mathbf{U}_r^\top$ \citep[Appendix D.2.1]{maImplicitRegularizationNonconvex2020}, one has
    \begin{align*}
        \|\boldH - \mathcal{O}_0\| &= \|\boldU_l(\cos\mathbf{\Psi}-\mathbf{I}_K)\boldU_r^\top\|\\
        &\leq \|\mathbf{I}_K-\cos\mathbf{\Psi}\|\\
        &\lesssim \|\sin\mathbf{\Psi}\|^2\\
        &= \|\hat{\boldU}\hat{\boldU}^\top - \boldU\boldU^\top\|^2\\
        &\leq \frac{(\zeta_{\text{op}}(\boldR)+\alpha\zeta_{\text{op}}(\boldX))^2}{\sigma_K(\cG)^2},
    \end{align*}
    where the second to last inequality stems from \citet[Lemma 2.1.2]{Chen2021}. Then, from Lemma \ref{lemma:combine},
    \begin{align*}
        \sigma_1(\boldH)&\leq \sigma_1(\mathcal{O}_0)+\|\boldH-\mathcal{O}_0\|\lesssim 1,\\
        \sigma_K(\boldH)&\geq \sigma_1(\mathcal{O}_0)-\|\boldH-\mathcal{O}_0\|\gtrsim 1.
    \end{align*}
    Therefore,
    \begin{equation*}
        \sigma_1(\boldH)\asymp \sigma_K(\boldH) \asymp 1.
    \end{equation*}
    Similarly,
    \begin{align*}
        \|\boldH^\top\boldH-\mathbf{I}_K\| &= \|\boldU_r (\cos^2\mathbf{\Psi} - \mathbf{I}_K)\mathbf{U}_r^\top\|\\
        &\leq \|\cos^2\mathbf{\Psi} - \mathbf{I}_K\|\\
        &\leq \|\sin\mathbf{\Psi}\|^2\\
        &\leq \frac{(\zeta_{\text{op}}(\boldR)+\alpha\zeta_{\text{op}}(\boldX))^2}{\sigma_K(\cG)^2}.
    \end{align*}
\end{proof}

\subsection{Proof of Lemma \ref{lemma:kappa}}
\label{sec:pf_lemma_kappa}
\begin{align*}
    \kappa(\mathcal{L}) &= \sqrt{\frac{\sigma_1(\mathcal{G})}{\sigma_K(\cG)}}\\
    &\leq \sqrt{\frac{\sigma_1(\cR)^2+\alpha \sigma_1(\cX)^2}{\sigma_K(\cR)^2+\alpha\sigma_K(\cX)^2}}\\
    &\leq \sqrt{\frac{\sigma_1(\cR)^2}{\sigma_K(\cR)^2}+ \frac{\alpha\sigma_1(\cX)^2}{\alpha\sigma_K(X)^2}}\\
    &=\sqrt{\kappa(\cR)^2+\kappa(\cX)^2} = \kappa_{\max},
\end{align*}
where the second equality comes from Lemma \ref{lemma:small_vals} and Weyl's inequality.

\subsection{Proof of Lemma \ref{lemma:29}}
\label{sec:lem_29}
    Note that from the definition,
    \begin{equation*}
        \cGm = \cGm_R + \alpha \cGm_X.
    \end{equation*}
    Therefore, one has
    \begin{align*}
        \left\|\cGm - \boldG \right\| &= \left\|\cGm_R+\alpha \cGm_X-\boldG_R - \alpha \boldG_X \right\|\\
        &\leq \left\|\cGm_R - \boldG_R\right\| + \alpha\left\|\cGm_X - \boldG_X \right\|.
    \end{align*}
    From \citet[Lemma 29]{Yan2021}, we have with probability at least $1-O(d_R^{-10})$
    \begin{equation*}
        \left\|\cGm_R - \boldG_R\right\| \lesssim \sigma(\boldR)^2\sqrt{NJ}\log d_R + \kappa(\cR)^2\sigma(\boldR)\sigma_1(\cR)\sqrt{\mu K\log d_R},
    \end{equation*}
    and similarly for $\|\cGm_X-\boldG_X \|$ with probability at least $1-O(d_X^{-10})$. Combining these inequalities yields, with probability at least $1-O(d^{-10})$
    \begin{align*}
        \left\|\cGm - \boldG \right\| &\lesssim \sigma(\boldR)^2\sqrt{NJ}\log d_R + \kappa(\cR)^2\sigma(\boldR)\sigma_1(\cR)\sqrt{\mu K\log d_R} \\
        &+ \alpha\left(\sigma(\boldX)^2\sqrt{NW}\log d_X + \kappa(\cX)^2\sigma(\boldX)\sigma_1(\cX)\sqrt{\mu K\log d_X}\right)\\
        &\leq \left[\sigma(\boldR)^2+\alpha\sigma(\boldX)^2 \right]\sqrt{N\max\{J, W \}}\log d \\
        &+ \left[\sigma(\boldR) \sigma_1(\cR) + \alpha \sigma(\boldX)\sigma_1(\cX) \right] \kappa_{\max}^2\sqrt{\mu K\log d}.
    \end{align*}

    Following the same reasoning as before, From \citet[Lemma 29]{Yan2021}, one has,
    \begin{equation*}
        \left\|\cGm_R - \cG_R\right\| \lesssim \zeta_{\text{op}}(\boldR),
    \end{equation*}
    and similarly for $\left\|\cGm_X-\cG_X\right\|$. Combining these two inequalities together yields,
    \begin{equation*}
        \left\|\cGm - \cG\right\| \lesssim \zeta_{\text{op}}(\boldR) + \alpha \zeta_{\text{op}}(\boldX),
    \end{equation*}
    with probability exceeding $1-O(d^{-10})$.

\subsection{Proof of Lemma \ref{lemma:leftbound1}}
\label{sec:lem_left1}
\begin{proof}
    From triangle inequality,
    \begin{align*}
        \|(\boldG-\cG)_{m,:}\boldU\|_2&\leq \left\|\left[\od(\boldG-\cG) \right]_{m,:}\boldU \right\|_2 + \|\left[\di(\boldG-\cG) \right]_{m,:}\boldU\|_2\\
        &= \underbrace{\left\|\left[\od(\boldG^0-\cG) \right]_{m,:}\boldU \right\|}_{\beta_1} + \underbrace{\|\left[\di(\boldG-\cG) \right]_{m,:}\boldU\|_2}_{\beta_2},
    \end{align*}
    where $\boldG^0$ is the diagonal-deleted version of the gram matrix, and the last equality hinges on the equivalence of off-diagonal components between diagonal deletion algorithm \citep{caiSubspaceEstimationUnbalanced2020} and heteroskedastic principal component analysis. 
    
    Following from \citet[Lemma 2]{caiSubspaceEstimationUnbalanced2020}, $\beta_1$ can be bounded by,
    \begin{equation*}
        \left\|\left[\od(\boldG^0-\cG) \right]_{m,:}\boldU \right\| \lesssim \left(\zeta_{\text{op}}(\boldR)+ \alpha\zeta_{\text{op}}(\boldX) \right)\sqrt{\frac{\mu K}{N}},
    \end{equation*}
    with probability at least $1-O(d^{-10})$. This hinges on the observation that
    \begin{align*}
        (\boldG-\mathcal{G})_{m,i} &= \langle \boldR_{m,:}, \boldR_{i,:} \rangle + \alpha \langle \boldX_{m,:}, \boldX_{i,:} \rangle - \langle \cR_{m,:}, \cR_{i,:} \rangle - \alpha \langle \cX_{m,:}, \cX_{i,:} \rangle \\
        &= \langle \textbf{E}_{m,:}^R,\textbf{E}_{i,:}^R \rangle + \langle \cR_{m,:}, \textbf{E}_{i,:}^R \rangle + \langle \textbf{E}_{m,:}^R, \cR_{i,:}\rangle\\
        &+\alpha\left(\langle \textbf{E}_{m,:}^X,\textbf{E}_{i,:}^X \rangle + \langle \cX_{m,:}, \textbf{E}_{i,:}^X \rangle + \langle \textbf{E}_{m,:}^X, \cX_{i,:}\rangle\right), \quad i\neq m;\\
        (\boldG^0-\mathcal{G})_{m,m} &= -\mathcal{\boldG}_{m,m} = - \|\cR_{m,:}\|_2^2 - \alpha \|\cX_{m,:}\|_2^2,
    \end{align*}
    and there is no interaction between $\boldR$ and $\boldX$, the same treatment can be applied. The slight difference of our bound lies in the diagonal term. One can also refer to \citet[Lemma 25]{Yan2021} for a detailed discussion.
    
    For $\beta_2$,
    \begin{align*}
        \|\left[\di(\boldG-\cG) \right]_{m,:}\boldU\|_2 &= |\boldG_{m,m}-\cG_{m,m}|\|\boldU_{m,:}\|_{2}\\
        &\lesssim \|\di(\boldG-\cG))\|\|\boldU\|_{2,\infty}\\
        &\lesssim \kappa_{\max}^2\sqrt{\frac{\mu K}{N}}\left( \zeta_{\text{op}}(\boldR) + \alpha\zeta_{\text{op}}(\boldX) \right) \sqrt{\frac{\mu K}{N}}\\
        &\lesssim \left( \zeta_{\text{op}}(\boldR) + \alpha\zeta_{\text{op}}(\boldX) \right) \sqrt{\frac{\mu K}{N}},
    \end{align*}
    where the second inequality is from Lemma \ref{lemma:wholebound} and the last inequality invokes the assumption that $N\gtrsim \kappa_{\max}^4\mu K$.

    Combining $\beta_1$ and $\beta_2$, we have
    \begin{equation*}
        \left\|(\boldG-\cG)_{m,:}\boldU \right\|_2 \lesssim \left( \zeta_{\text{op}}(\boldR)+\alpha\zeta_{\text{op}}(\boldX) \right)\sqrt{\frac{\mu K}{N}},
    \end{equation*}
    with probability at least $1-O(d^{-10})$.
\end{proof}

\subsection{Proof of Lemma \ref{lemma:32}}
\label{sec:lem_32}
\begin{proof}
    The proof is similar for $\boldR$ and $\boldX$. We here only present the proof for $\boldR$. The following proof is adapted from the proof of \citet[Lemma 32]{Yan2021} and \citet[Lemma 7]{caiSubspaceEstimationUnbalanced2020}. Here, we only show the proof of Eqn. (\ref{eqn:F18}). The proof of (\ref{eqn:F19}) is similar and actually simpler than the proof of (\ref{eqn:F18}).

    From matrix Bernstein inequality \citep[Theorem 6.1.1]{Tropp2015} and \citet[Lemma 32]{Yan2021}, one has
    \begin{align*}
        \left\|\mathbf{E}_{m,:}^R[\Pm(\boldR)]^\top\left( \Um\Hm-\boldU\right) \right\| &= \left\| \sum_{j=1}^J \boldE_{m,j}^R \left[[\Pm(\boldR)]^\top\left( \Um\Hm-\boldU\right)\right]_{j,:} \right\|\\
        &\lesssim \underbrace{\sigma(\boldR)\left\|\left[\Pm(\boldR)\right]^\top \left(\Um\Hm-\boldU\right) \right\|_{F}\sqrt{\log d_R}}_{\beta_1}\\
        & + \underbrace{B_R\left\|[\Pm(\boldR)]^\top\left( \Um\Hm-\boldU\right) \right\|_{2,\infty}\log d_R}_{\beta_2}. 
    \end{align*}

    First, for $\beta_1$, with probability at least $1-O(d_R^{-11})$, we have,
    \begin{align*}
        \left\|[\Pm(\boldR)]^\top\left(\Um\Hm-\boldU\right)  \right\|_F &\leq   \|\boldR\|\left\|\Um\Hm-\boldU\right\|_F \\
        &\leq \left( \|\cR\|+\|\boldE^R\|\right) \left\|\Um\Hm-\boldU\right\|_F\\
        &\lesssim \left(\sigma_1(\cR)+\sigma(\boldR)\sqrt{J}\right) \left\|\Um\Hm-\boldU\right\|_F.
    \end{align*}
    As a result, following the subsequent proof in \citet[Lemma 32]{Yan2021}, we obtain,
    \begin{equation*}
        \beta_1 \lesssim \zeta_{\text{op}}(\boldR)\left\|\Um\Hm-\boldU\right\|_{2,\infty}.
    \end{equation*}

    For $\beta_2$, first, by adapting \citet[Lemma 31]{Yan2021} in the aforementioned way, we can obtain
    \begin{equation}
        \label{eqn:F16}
        \begin{aligned}
            &\left\|\mathbf{e}^\top_l[\Pm(\boldR)]^\top \left(\hat{\boldU}^{(m)}\Hm-\boldU \right) \right\|\\ &\lesssim \underbrace{\frac{(\zeta_{\text{op}}(\boldR)+\alpha\zeta_{\text{op}}(\boldX))^2}{\sigma_K(\cG)^2}\|\cR^\top\|_{2,\infty}}_{\beta_{2,1}}  \\
            &+ \underbrace{\left(B_R\log d_R+ \sigma(\boldR) \sqrt{N\log d_R} \right)\|\Um\Hm-\boldU\|_{2,\infty}}_{\beta_{2,2}}\\
            &+ \underbrace{\left( \left\|\cR^\top \right\|_{2,\infty} +  B_R\log d_R+ \sigma(\boldR) \sqrt{N\log d_R} \right)\left\| \Um\UmT - \boldU^{(m,l)}\boldU^{(m,l)\top}\right\|}_{\beta_{2,3}},
        \end{aligned}
    \end{equation}
    and we have a similar inequality for the terms corresponding to $\mathbf{X}$. In the following, we bound $\beta_{2,1},\beta_{2,2},\beta_{2,3}$ respectively. 
    
    For $(B_R\log d_R)\|\cR^\top\|_{2,\infty}$, following the treatment in \citet[Lemma 32]{Yan2021}, one has
    \begin{equation*}
        (B_R\log d_R)\left\| \cR^\top \right\|_{2,\infty} \leq \sigma(\boldR)\sqrt{J\log d_R}\sqrt{\frac{\mu K}{J}}\sigma_1(\cR) \lesssim \sqrt{\frac{\mu K}{N}}\zeta_{\text{op}}(\boldR).
    \end{equation*}
    Then, the whole $(B_R\log d_R)\beta_{2,1}$ will be bounded by,
    \begin{equation*}
        (B_R\log d_R)\beta_{2,1} \lesssim \sqrt{\frac{\mu K}{N}}\frac{\zeta_{\text{op}}(\boldG)^2}{\sigma_K(\cG)^2}\zeta_{\text{op}}(\boldR).
    \end{equation*}

    Similarly for $\beta_{2,2}$, we have
    \begin{equation*}
        (B_R\log d_r)\beta_{2,2} \lesssim \zeta_{\text{op}}(\boldR)\left\| \Um\Hm - \boldU \right\|_{2,\infty}.
    \end{equation*}

    \textbf{As for $\beta_{2,3}$}: Following the treatment in \citet[Lemma 32]{Yan2021}, we have
    \begin{align*}
        &(B_R\log d_R)\left\| \Um\UmT - \Uml\UmlT \right\|\\ &\lesssim \frac{B_R\log d_R}{\sigma_K(\cG)}\left( B_R\log d_R+\sigma(\boldR)\sqrt{N\log d_R} \right)^2\left\| \Um\Hm \right\|_2\\
        &+ \frac{\sigma(\boldR)}{\sigma_K(\cG)}B_R\log d_R + \frac{B_R\log d_R}{\sigma_K(\cG)} \left(B_R\log d_R+\sigma(\boldR)\sqrt{N\log d_R}\right)\left\| \cR^\top \right\|_{2,\infty}\\
        &+ \alpha\left[ \frac{B_R\log d_R}{\sigma_K(\cG)}\left( B_X\log d_X+\sigma(\boldX)\sqrt{N\log d_X} \right)^2\left\| \Um\Hm \right\|_2 \right.\\
        &\left.+ \frac{\sigma(\boldX)}{\sigma_K(\cG)}B_R\log d_R + \left(B_X\log d_X+ \frac{B_R\log d_R}{\sigma_K(\cG)}\sigma(\boldX)\sqrt{N\log d_X}\right)\left\| \cX^\top \right\|_{2,\infty} \right]\\
        &\lesssim \frac{B_R\log d_R}{\sigma_K(\cG)}\zeta_{\text{op}}(\boldR) \left\|\Um\Hm\right\|_{2,\infty}\\
        &+ \left(B_R\log d_R + \sigma(\boldR) \sqrt{N\log d_R}\right)\sqrt{\frac{\mu K}{N}} \frac{\zeta_{\text{op}}(\boldR)}{\sigma_K(\cG)}\\
        &+ \alpha\left[ \frac{B_R\log d_R}{\sigma_K(\cG)}\zeta_{\text{op}}(\boldX) \left\|\Um\Hm\right\|_{2,\infty}\right.\\
        &\left.+ \sqrt{\frac{J\log d_R}{W\log d_X}} \frac{\sigma(\boldR)}{\sigma(\boldX)} \left(B_X\log d_X + \sigma(\boldX) \sqrt{N\log d_X}\right)\sqrt{\frac{\mu K}{N}} \frac{\zeta_{\text{op}}(\boldX)}{\sigma_K(\cG)} \right]\\
        &\stackrel{(a)}{\lesssim} \frac{B_R\log d_R}{\sigma_K(\cG)}\zeta_{\text{op}}(\boldR) \left\|\Um\Hm\right\|_{2,\infty}\\
        &+ \left(B_R\log d_R + \sigma(\boldR) \sqrt{N\log d_R}\right)\sqrt{\frac{\mu K}{N}} \frac{\zeta_{\text{op}}(\boldR)}{\sigma_K(\cG)}\\
        &+ \alpha\left[ \frac{B_R\log d_R}{\sigma_K(\cG)}\zeta_{\text{op}}(\boldX) \left\|\Um\Hm\right\|_{2,\infty}\right.\\
        &\left.+ \left(B_X\log d_X + \sigma(\boldX) \sqrt{N\log d_X}\right)\sqrt{\frac{\mu K}{N}} \frac{\zeta_{\text{op}}(\boldX)}{\sigma_K(\cG)} \right],
    \end{align*}
    where we use the assumption that $\sigma(\boldR)^2 J\log d_R\asymp \sigma(\boldX)^2 W\log d_X$ in (a).  

    Therefore, with probability at least $1-O(d_R^{-10})$, we have,
    \begin{align*}
        (B_R\log d_R)\beta_{2,3} &\lesssim \left(\left\|\cR^\top\right\|_{2,\infty} + B_R\log d_R + \sigma(\boldR)\sqrt{N\log d_R}\right)\left\{ \frac{B_R\log d_R}{\sigma_K(\cG)}\zeta_{\text{op}}(\boldR) \left\|\Um\Hm\right\|_{2,\infty}\right.\\
        &+ \left(B_R\log d_R + \sigma(\boldR) \sqrt{N\log d_R}\right)\sqrt{\frac{\mu K}{N}} \frac{\zeta_{\text{op}}(\boldR)}{\sigma_K(\cG)} + \alpha\left[ \frac{B_R\log d_R}{\sigma_K(\cG)}\zeta_{\text{op}}(\boldX) \left\|\Um\Hm\right\|_{2,\infty}\right.\\
        &\left.\left.+ \left(B_X\log d_X + \sigma(\boldX) \sqrt{N\log d_X}\right)\sqrt{\frac{\mu K}{N}} \frac{\zeta_{\text{op}}(\boldX)}{\sigma_K(\cG)} \right]\right\}\\
        &\lesssim \frac{\zeta_{\text{op}}(\boldR)^2}{\sigma_K(\cG)}\left\| \Um\Hm \right\|_{2,\infty} + \sqrt{\frac{\mu K}{N}}\frac{\zeta_{\text{op}}(\boldR)^2}{\sigma_K(\cG)} \\
        &+\alpha \left[\frac{\zeta_{\text{op}}(\boldR)\zeta_{\text{op}}(\boldX)}{\sigma_K(\cG)}\left\| \Um\Hm \right\|_{2,\infty} + \frac{\mu K}{\sqrt{NW}}\frac{\zetaop(\boldR)\zetaop(\boldX)}{\sigma_K(\cG)}+ \sqrt{\frac{\mu K}{N}}\frac{\zetaop(\boldG)^2}{\sigma_K(\cG)}\right],
    \end{align*}
    where the last step follows from
    \begin{align*}
        &\left(\left\|\cR^\top\right\|_{2,\infty} + B_R\log d_R + \sigma(\boldR)\sqrt{N\log d_R}\right) \left(B_X\log d_X + \sigma(\boldR) \sqrt{N\log d_X}\right)\sqrt{\frac{\mu K}{N}} \frac{\zeta_{\text{op}}(\boldX)}{\sigma_K(\cG)}\\
        &\lesssim \left\|\cR^\top\right\|_{2,\infty}\sigma(\boldX)\sqrt{\mu K\log d_X}\sqrt{\frac{\mu K}{N}}\frac{\zetaop(\boldX)}{\sigma_K(\cG)} + \sqrt{\frac{\mu K}{N}} \frac{\zetaop(\boldX)^{3/2}\zetaop(\boldR)^{1/2}}{\sigma_K(\cG)}\\
        &\lesssim \frac{\mu K}{\sqrt{NW}}\frac{\zetaop(\boldR)\zetaop(\boldX)}{\sigma_K(\cG)} + \sqrt{\frac{\mu K}{N}} \frac{\zetaop(\boldG)^2}{\sigma_K(\cG)},
    \end{align*}
    where the last inequality stems from both $\zetaop(\boldR)\lesssim \zetaop(\boldG)$ and $\zetaop(\boldX)\lesssim \zetaop(\boldG)$,
    \begin{equation*}
        \left\|\cR^\top \right\|_{2,\infty}\lesssim \sqrt{\frac{\mu K}{J}}\sigma_1(\cR),
    \end{equation*}
    and 
    \begin{equation*}
        \left(B_X\log d_X + \sigma(\boldX)\sqrt{N\log d_X}\right)^2 \lesssim \zetaop(\boldX),
    \end{equation*}
    following from \citet[F.49]{Yan2021}.

    Putting $\beta_1$, $\beta_{2,1}$, $\beta_{2,2}$, $\beta_{2,3}$ together, we can obtain,

    \begin{align*}
        &\left\|\mathbf{E}_{m,:}^R[\Pm(\boldR)]^\top\left( \Um\Hm-\boldU\right)\right\|_2\\
        &\lesssim \zeta_{\text{op}}(\boldR)\left\|\Um\Hm-\boldU\right\|_{2,\infty} + \frac{\zeta_{\text{op}}(\boldR)^2}{\sigma_K(\cG)}\left\| \Um\Hm \right\|_{2,\infty}\\
        &+ \sqrt{\frac{\mu K}{N}}\frac{\zeta_{\text{op}}(\boldR)^2}{\sigma_K(\cG)} +\alpha \left[\frac{\zeta_{\text{op}}(\boldR)\zeta_{\text{op}}(\boldX)}{\sigma_K(\cG)}\left\| \Um\Hm \right\|_{2,\infty}\right.\\
        &\left.+ \sqrt{\frac{\mu K}{N}}\frac{\zeta_{\text{op}}(\boldR)\zeta_{\text{op}}(\boldX)}{\sigma_K(\cG)} + \sqrt{\frac{\mu K}{N}} \frac{\zetaop(\boldG)^2}{\sigma_K(\cG)} \right].
    \end{align*}
    
\end{proof}

\subsection{Proof of Lemma \ref{lemma:GUbound}}
\label{sec:lem_GU}
\begin{proof}
    Define 
    \begin{align*}
        \left[\mathcal{P}_{-m,:}(\mathbf{A})\right]_{i,:} = \left\{
        \begin{aligned}
            &0,\quad  &\text{if } i=m,\\
            &\mathbf{A}_{i,:},\quad  &\text{if } i\neq m.
        \end{aligned}
        \right.
    \end{align*}
    This operator zeros out the $m$th row of any matrix $\mathbf{A}$. Recall the definition of $\boldG$, and we can obtain,
    \begin{align*}
        \boldG_{m,:}&=\boldR_{m,:}[\Pm(\boldR)]^\top + \boldG_{m,m}^R \mathbf{e}^\top_m + \alpha\boldX_{m,:}[\Pm(\boldX)]^\top + \alpha\boldG_{m,m}^X\mathbf{e}_m^\top\\
        &= \cR_{m,:}\cR^\top + \cR_{m,:}[\Pm(\boldE^R)]^\top + \boldE^R_{m,:}[\Pm(\boldR)]^\top + (\boldG_{m,m}^R-\cG^R_{m,m})\mathbf{e}_m^\top\\
        &+\alpha\left(\cX_{m,:}\cX^\top + \cX_{m,:}[\Pm(\boldE^X)]^\top + \boldE^X_{m,:}[\Pm(\boldX)]^\top + (\boldG_{m,m}^X-\cG^X_{m,m})\mathbf{e}_m^\top\right).
    \end{align*}
    Hence, 
    \begin{align*}
        \|\boldG_{m,:}(\hat{\boldU}\boldH-\boldU)\|_2&\leq \underbrace{\|\mathbf{e}_m^\top\cG (\hat{\boldU}\boldH-\boldU)\|_2}_{\beta_1} + \underbrace{\|\cR_{m,:}[\Pm(\boldE^R)]^\top (\hat{\boldU}\boldH-\boldU)\|_2}_{\beta_2}\\
        &+ \underbrace{\|\boldE^R_{m,:}[\Pm(\boldR)]^\top (\hat{\boldU}\boldH-\boldU)\|_2}_{\beta_3}  + \underbrace{\|(\boldG_{m,m}^R-\cG_{m,m}^R)\mathbf{e}_m^\top(\hat{\boldU}\boldH - \boldU)\|_2}_{\beta_4}\\
        &+\alpha\left(\|\cX_{m,:}[\Pm(\boldE^X)]^\top (\hat{\boldU}\boldH-\boldU)\|_2 + \|\boldE^X_{m,:}[\Pm(\boldX)]^\top(\hat{\boldU}\boldH-\boldU)\|_2\right.\\
        &+ \left.  \|(\boldG_{m,m}^X-\cG_{m,m}^X)\mathbf{e}_m^\top(\hat{\boldU}\boldH - \boldU)\|_2        \right).
    \end{align*}
    Since the proof of $\boldR$ terms and $\boldX$ terms are similar, for simplicity, in the following, we only prove the parts concerning $\boldR$.

    For $\beta_1$, one has
    \begin{align*}
        \|\mathbf{e}_{m}^\top\cG(\hat{\boldU}\boldH-\boldU)\|_2 &\leq \|\mathbf{e}_m^\top\boldU\boldSigma^2\boldU^\top(\hat{\boldU}\boldH - \boldU)\|_2\\
        &\leq \|\boldU\|_{2,\infty}\|\boldSigma\|^2\|\boldU^\top (\hat{\boldU}\hat{\boldU}^\top\boldU-\boldU)\|\\
        &\lesssim \sqrt{\frac{\mu K}{N}}\sigma_1(\cG) \|\boldH^\top\boldH - \mathbf{I}_K\|\\
        &\lesssim \kappa(\cG)\frac{(\zeta_{\text{op}}(\boldR)+ \alpha\zeta_{\text{op}}(\boldX))^2}{\sigma_K(\cG)} \sqrt{\frac{\mu K}{N}},
    \end{align*}
    where the last inequality evokes Lemma \ref{lemma:spectralbound}.

    For $\beta_2$, one can see that
    \begin{equation*}
        \|\cR_{m,:}[\Pm(\boldE^R)]^\top (\hat{\boldU}\boldH-\boldU)\|_2 \leq \left\|\cR_{m,:}[\Pm(\boldE^R)]^\top\right\|_2 \|\hat{\boldU}\boldH-\boldU\|.
    \end{equation*}
    These two terms can be bounded separately. By \citet[Lemma 26]{Yan2021},
    \begin{equation*}
        \left\|\cR_{m,:}[\Pm(\boldE^R)]^\top\right\|_2\lesssim \sqrt{\mu K} \sigma_1(\cR)\sigma(\boldR).
    \end{equation*}
    Combined with Lemma \ref{lemma:spectralbound}, we have,
    \begin{align*}
        \left\|\cR_{m,:}[\Pm(\boldE^R)]^\top (\hat{\boldU}\boldH-\boldU)\right\|_2 &\lesssim \sqrt{\mu K} \sigma_1(\cR)\sigma(\boldR) \frac{\zeta_{\text{op}}(\boldR)+\alpha\zeta_{\text{op}}(\boldX)}{\sigma_K(\cG)}\\
        &\lesssim \sqrt{\frac{\mu K}{N}} \frac{\zetaop(\boldR)\zetaop(\boldG)}{\sigma_K(\cG)}
    \end{align*}
    Note that from the definition of $\zeta_{\text{op}}(\boldR)$, one has $\zeta_{\text{op}}(\boldR)\geq \sigma(\boldR)\sigma_1(\cR)\sqrt{N\log d_R}$.

    For $\beta_3$, first,
    \begin{equation*}
        \beta_3 \leq \underbrace{\left\| \boldE_{m,\cdot}^R\left[\mathcal{P}_{-m,\cdot}(\boldR) \right]^\top \left( \Um\Hm-\boldU\right) \right\|_2}_{\beta_{3,1}}  + \underbrace{\left\| \boldE_{m,\cdot}^R\left[\mathcal{P}_{-m,\cdot}(\boldR) \right]^\top \left( \hat{\boldU}\boldH - \Um\Hm\right) \right\|_2}_{\beta_{3,2}}. 
    \end{equation*}
    For $\beta_{3,1}$, from Lemma \ref{lemma:32}, we have with probability at least $1-O(d_R^{-10})$,
    \begin{align*}
        &\left\|\mathbf{E}_{m,:}^R[\Pm(\boldR)]^\top\left( \Um\Hm-\boldU\right)\right\|_2 \\&\lesssim \zeta_{\text{op}}(\boldR)\left\|\Um\Hm-\boldU\right\|_{2,\infty} + \frac{\zeta_{\text{op}}(\boldR)^2}{\sigma_K(\cG)}\left\| \Um\Hm \right\|_{2,\infty}\\
        &+ \sqrt{\frac{\mu K}{N}}\frac{\zeta_{\text{op}}(\boldR)^2}{\sigma_K(\cG)} +\alpha \left[\frac{\zeta_{\text{op}}(\boldR)\zeta_{\text{op}}(\boldX)}{\sigma_K(\cG)}\left\| \Um\Hm \right\|_{2,\infty}\right.\\
        &\left.+ \sqrt{\frac{\mu K}{N}}\frac{\zeta_{\text{op}}(\boldR)\zeta_{\text{op}}(\boldX)}{\sigma_K(\cG)} + \sqrt{\frac{\mu K}{N}}\frac{\zetaop(\boldG)^{2}}{\sigma_K(\cG)} \right],
    \end{align*}
    Together with Lemma \ref{lemma:33}, 
    \begin{align*}
        \left\|\Um\Hm-\hat{\boldU}\boldH\right\| &\leq \left\|\Um\UmT-\hat{\boldU}\hat{\boldU}^\top\right\|\\
        &\lesssim \kappa_{\max}^2\frac{\zeta_{\text{op}}(\boldR) + \alpha \zeta_{\text{op}}(\boldX)}{\sigma_K(\cG)} \left( \left\| \hat{\boldU}\boldH\right\|_{2,\infty} + \sqrt{\frac{\mu K}{N}}\right)\\
        &\leq \kappa_{\max}^2\frac{\zeta_{\text{op}}(\boldR) + \alpha \zeta_{\text{op}}(\boldX)}{\sigma_K(\cG)} \left( \left\| \hat{\boldU}\boldH - \boldU\right\|_{2,\infty} + \left\|\boldU\right\|_{2,\infty} + \sqrt{\frac{\mu K}{N}}\right)\\
        &\asymp \kappa_{\max}^2\frac{\zeta_{\text{op}}(\boldR) + \alpha \zeta_{\text{op}}(\boldX)}{\sigma_K(\cG)} \left( \left\| \hat{\boldU}\boldH - \boldU\right\|_{2,\infty} + \sqrt{\frac{\mu K}{N}}\right),
    \end{align*}
    and thus
    \begin{align*}
        \left\|\Um\Hm-\boldU\right\|_{2,\infty} &\leq \left\|\Um\Hm-\hat{\boldU}\boldH\right\|_{2,\infty} + \left\| \hat{\boldU}\boldH - \boldU \right\|_{2,\infty}\\
        &\lesssim \kappa_{\max}^2\frac{\zeta_{\text{op}}(\boldR) + \alpha \zeta_{\text{op}}(\boldX)}{\sigma_K(\cG)} \left( \left\| \hat{\boldU}\boldH - \boldU\right\|_{2,\infty} + \sqrt{\frac{\mu K}{N}}\right) + \left\| \hat{\boldU}\boldH - \boldU \right\|_{2,\infty}\\
        &\lesssim \kappa_{\max}^2\frac{\zeta_{\text{op}}(\boldR) + \alpha \zeta_{\text{op}}(\boldX)}{\sigma_K(\cG)} \sqrt{\frac{\mu K}{N}}+ \left\| \hat{\boldU}\boldH - \boldU \right\|_{2,\infty},
    \end{align*}
    \begin{align*}
        \left\|\Um\Hm\right\|_{2,\infty} &\lesssim \|\boldU\|_{2,\infty} + \left\|\Um\Hm-\boldU\right\|_{2,\infty}\\
        &\lesssim \sqrt{\frac{\mu K}{N}} +  \kappa_{\max}^2\frac{\zeta_{\text{op}}(\boldR) + \alpha \zeta_{\text{op}}(\boldX)}{\sigma_K(\cG)} \sqrt{\frac{\mu K}{N}}+ \left\| \hat{\boldU}\boldH - \boldU \right\|_{2,\infty}\\
        &\asymp \sqrt{\frac{\mu K}{N}} + \left\| \hat{\boldU}\boldH - \boldU \right\|_{2,\infty},
    \end{align*}
    where in the last equality, we use $\zetaop(\boldG)/\kappa_{\max}^2\ll \sigma_K(\cG)$.

    Therefore,
    \begin{align*}
        \beta_{3,1} &\lesssim \kappa_{\max}^2\frac{\zeta_{\text{op}}(\boldR) \zeta_{\text{op}}(\boldG)}{\sigma_K(\cG)}\sqrt{\frac{\mu K}{N}} + \zeta_{\text{op}}(\boldR) \left\|\hat{\boldU}\boldH-\boldU\right\|_{2,\infty}+ \frac{\zeta_{\text{op}}(\boldR)^2}{\sigma_K(\cG)}\sqrt{\frac{\mu K}{N}}\\
        &+\frac{\zetaop(\boldR)^2}{\sigma_K(\cG)}\left\|\hat{\boldU}\boldH-\boldU\right\|_{2,\infty} + \sqrt{\frac{\mu K}{N}}\frac{\zetaop(\boldR)^2}{\sigma_K(\cG)} + \alpha\left[ \frac{\zeta_{\text{op}}(\boldR)\zetaop(\boldX)}{\sigma_K(\cG)}\sqrt{\frac{\mu K}{N}}\right.\\
        &\left.+ \frac{\zetaop(\boldR)\zetaop(\boldX)}{\sigma_K(\cG)} \left\|\hat{\boldU}\boldH-\boldU\right\|_{2,\infty} +  \sqrt{\frac{\mu K}{N}}\frac{\zetaop(\boldG)^{2}}{\sigma_K(\cG)}  \right]\\
        &\asymp \kappa_{\max}^2\frac{\zeta_{\text{op}}(\boldR) \zeta_{\text{op}}(\boldG)}{\sigma_K(\cG)}\sqrt{\frac{\mu K}{N}} + \zetaop(\boldR)\left\|\hat{\boldU}\boldH-\boldU\right\|_{2,\infty}\\
        &+\alpha\left[ \frac{\zeta_{\text{op}}(\boldR)\zetaop(\boldX)}{\sigma_K(\cG)}\sqrt{\frac{\mu K}{N}} + \sqrt{\frac{\mu K}{N}}\frac{\zetaop(\boldG)^{2}}{\sigma_K(\cG)} + \frac{\zetaop(\boldR)\zetaop(\boldX)}{\sigma_K(\cG)} \left\|\hat{\boldU}\boldH-\boldU\right\|_{2,\infty} \right].
    \end{align*}

    For $\beta_{3,2}$,
    \begin{align*}
        \left\| \boldE_{m,\cdot}^R\left[\mathcal{P}_{-m,\cdot}(\boldR) \right]^\top \left( \hat{\boldU}\boldH - \Um\Hm\right) \right\|_2 &\leq \left\| \boldE_{m,\cdot}^R\left[\mathcal{P}_{-m,\cdot}(\boldR) \right]^\top\right\|_2 \left\| \left( \hat{\boldU}\boldH - \Um\Hm\right) \right\|\\
        &\lesssim \zetaop(\boldR)\left\| \left( \hat{\boldU}\boldH - \Um\Hm\right) \right\|\\
        &\lesssim \kappa_{\max}^2\zetaop(\boldR)\frac{\zetaop(\boldG)}{\sigma_K(\cG)} \left(\left\|\Um\Hm\right\|_{2,\infty} + \sqrt{\frac{\mu K}{N}}\right)\\
        &\lesssim \kappa_{\max}^2\zetaop(\boldR)\frac{\zetaop(\boldG)}{\sigma_K(\cG)}\left(\left\|\hat{\boldU}\boldH-\boldU\right\|_{2,\infty} + \sqrt{\frac{\mu K}{N}}\right),
    \end{align*}
    where the second to last inequality applies Lemma \ref{lemma:33}.

    Combining $\beta_{3,1}$ and $\beta_{3,2}$ yield,
    \begin{align*}
        \beta_3 &\lesssim \kappa_{\max}^2\frac{\zeta_{\text{op}}(\boldR) \zeta_{\text{op}}(\boldG)}{\sigma_K(\cG)}\sqrt{\frac{\mu K}{N}} + \zetaop(\boldR)\left\|\hat{\boldU}\boldH-\boldU\right\|_{2,\infty}\\
        &+\alpha\left[ \frac{\zeta_{\text{op}}(\boldR)\zetaop(\boldX)}{\sigma_K(\cG)}\sqrt{\frac{\mu K}{N}} + \sqrt{\frac{\mu K}{N}}\frac{\zetaop(\boldG)^{2}}{\sigma_K(\cG)} + \frac{\zetaop(\boldR)\zetaop(\boldX)}{\sigma_K(\cG)} \left\|\hat{\boldU}\boldH-\boldU\right\|_{2,\infty} \right].
    \end{align*}

    For $\beta_4$, we have,
    \begin{equation*}
        \beta_4 \leq \left\| \od\left(\boldG^R-\cG^R\right) \right\| \left\|\hat{\boldU}\boldH - \boldU \right\|_{2,\infty}
        \lesssim \kappa_{\max}^2\sqrt{\frac{\mu K}{N}}\zeta_{\text{op}}(\boldR)\left\|\hat{\boldU}\boldH - \boldU \right\|_{2,\infty},
    \end{equation*}
    where the last inequality comes from \citet[Lemma 23]{Yan2021}.

    Then, combining $\beta_1,\cdots,\beta_4$, following the proof of \citet[Lemma 26]{Yan2021}, we can obtain,
    \begin{align*}
        \beta_1+\beta_2+\beta_3+\beta4 &\lesssim \kappa_{\max}^2\frac{\zeta_{\text{op}}(\boldR) \zeta_{\text{op}}(\boldG)}{\sigma_K(\cG)}\sqrt{\frac{\mu K}{N}} + \zetaop(\boldR)\left\|\hat{\boldU}\boldH-\boldU\right\|_{2,\infty}\\
        &+\alpha\left[ \frac{\zeta_{\text{op}}(\boldR)\zetaop(\boldX)}{\sigma_K(\cG)}\sqrt{\frac{\mu K}{N}} + \sqrt{\frac{\mu K}{N}}\frac{\zetaop(\boldG)^{2}}{\sigma_K(\cG)} + \frac{\zetaop(\boldR)\zetaop(\boldX)}{\sigma_K(\cG)} \left\|\hat{\boldU}\boldH-\boldU\right\|_{2,\infty} \right].
    \end{align*}

    Thus, by symmetry, we have,
    \begin{align*}
        \left\|\boldG_{m,:}(\hat{\boldU}\boldH-\boldU)\right\|_2 &\lesssim \kappa_{\max}^2\frac{\zeta_{\text{op}}(\boldR) \zeta_{\text{op}}(\boldG)}{\sigma_K(\cG)}\sqrt{\frac{\mu K}{N}} + \zetaop(\boldR)\left\|\hat{\boldU}\boldH-\boldU\right\|_{2,\infty}\\
        &+\alpha\left[ \frac{\zeta_{\text{op}}(\boldR)\zetaop(\boldX)}{\sigma_K(\cG)}\sqrt{\frac{\mu K}{N}} + \frac{\zetaop(\boldR)\zetaop(\boldX)}{\sigma_K(\cG)} \left\|\hat{\boldU}\boldH-\boldU\right\|_{2,\infty} + \sqrt{\frac{\mu K}{N}}\frac{\zetaop(\boldG)^2}{\sigma_K(\cG)} \right]\\
        &+\alpha \left\{ \kappa_{\max}^2\frac{\zeta_{\text{op}}(\boldX) \zeta_{\text{op}}(\boldG)}{\sigma_K(\cG)}\sqrt{\frac{\mu K}{N}} + \zetaop(\boldX)\left\|\hat{\boldU}\boldH-\boldU\right\|_{2,\infty}\right.\\
        &\left.+\alpha\left[ \frac{\zeta_{\text{op}}(\boldR)\zetaop(\boldX)}{\sigma_K(\cG)}\sqrt{\frac{\mu K}{N}} + \frac{\zetaop(\boldR)\zetaop(\boldX)}{\sigma_K(\cG)} \left\|\hat{\boldU}\boldH-\boldU\right\|_{2,\infty} + \sqrt{\frac{\mu K}{N}} \frac{\zetaop(\boldG)^2}{\sigma_k(\cG)} \right] \right\}\\
        &\lesssim \kappa_{\max}^2\frac{\zeta_{\text{op}}(\boldR) \zeta_{\text{op}}(\boldG)}{\sigma_K(\cG)}\sqrt{\frac{\mu K}{N}} + \zetaop(\boldR)\left\|\hat{\boldU}\boldH-\boldU\right\|_{2,\infty}\\
        &+\alpha\left[ \kappa_{\max}^2\frac{\zeta_{\text{op}}(\boldX) \zeta_{\text{op}}(\boldG)}{\sigma_K(\cG)}\sqrt{\frac{\mu K}{N}} + \zetaop(\boldX)\left\|\hat{\boldU}\boldH-\boldU\right\|_{2,\infty} + \sqrt{\frac{\mu K}{N}} \frac{\zetaop(\boldG)^2}{\sigma_k(\cG)} \right]\\
        &+\alpha^2\left[ \frac{\zeta_{\text{op}}(\boldR)\zetaop(\boldX)}{\sigma_K(\cG)}\sqrt{\frac{\mu K}{N}} + \frac{\zetaop(\boldR)\zetaop(\boldX)}{\sigma_K(\cG)} \left\|\hat{\boldU}\boldH-\boldU\right\|_{2,\infty} + \sqrt{\frac{\mu K}{N}} \frac{\zetaop(\boldG)^2}{\sigma_k(\cG)} \right],
    \end{align*}
    provided that $\zetaop(\boldR)\ll \sigma_K(\cG)^2$.

\end{proof}

\subsection{Proof of Lemma \ref{lemma:l1_l2_U}}
\label{sec:lem_l1l2}
    Similarly as in \citet[Lemma 30]{Yan2021} and \citet[Lemma 9]{caiSubspaceEstimationUnbalanced2020}, to prove the result, we only need to show the upper bound on $\left\|\boldG^{(m)}-\boldG^{(m,l)}  \right\|$ and $\left\| \left( \boldG^{(m)} - \boldG^{(m,l)} \right)\Uml \right\|$ respectively. The bound on $\left\|\boldG^{(m)}-\boldG^{(m,l)}  \right\|$ is just a simple application of triangle inequality. For $\left\| \left( \boldG^{(m)} - \boldG^{(m,l)} \right)\Uml \right\|$, first, note that,
    \begin{equation*}
        \left( \boldG^{(m)} - \boldG^{(m,l)} \right)_{m,\cdot} = \cR\boldE^{R\top}_{:,l} - \cR_{m,l}E_{m,l}^R\mathbf{e}_m^\top + \alpha \left( \cX\boldE^{X\top}_{:,l} - \cX_{m,l}E_{m,l}^X\mathbf{e}_m^\top \right).
    \end{equation*}

    Following the treatment in \citet[Lemma 9]{caiSubspaceEstimationUnbalanced2020}, one has
    \begin{align*}
        \left\|\left( \boldG^{(m)} - \boldG^{(m,l)} \right)\Uml \right\| &\lesssim \underbrace{\left\| \Pm\left(\boldG^{(m)} - \boldG^{(m,l)}\right)\Uml\Hml \right\|}_{\beta_1}\\
        &+ \underbrace{\left\| \mathcal{P}_{m,\cdot} \left(\boldG^{(m)} - \boldG^{(m,l)}\right)\Uml\Hml \right\|}_{\beta_2}.
    \end{align*}

    For $\beta_2$, we have
    \begin{align*}
        \beta_2 &\leq \left\| \left( \boldG^{(m)}_R-\boldG^{(m,l)}_R \right)_{m,\cdot} \Uml\Hml \right\|_2 + \left\| \left( \boldG^{(m)}_R - \boldG^{(m,l)}_R \right)_{m,\cdot} \right\|_2\left\| \Uml\Hml \right\|_{2,\infty}\\
        &+ \alpha\left(\left\| \left( \boldG^{(m)}_X-\boldG^{(m,l)}_X \right)_{m,\cdot} \Uml\Hml \right\|_2 + \left\| \left( \boldG^{(m)}_X - \boldG^{(m,l)}_X \right)_{m,\cdot} \right\|_2\left\| \Uml\Hml \right\|_{2,\infty}\right)\\
        &\leq \|\cR\|_{\infty}\left\|(\boldE_{\cdot,l}^R)\Uml\Hml\right\|_2 + \|\cR\|_{\infty}\left\|\boldE_{\cdot,l}^R\right\|_2 \left\|\Uml\Hml\right\|_{2,\infty}\\
        &+ \alpha \left( \|\cX\|_{\infty}\left\|(\boldE_{\cdot,l}^X)\Uml\Hml\right\|_2 + \|\cX\|_{\infty}\left\|\boldE_{\cdot,l}^X\right\|_2 \left\|\Uml\Hml\right\|_{2,\infty} \right)\\
        &\lesssim \left[\left( B_R\log d_R+\sigma(\boldR)\sqrt{N\log d_R} \right)\|\cR\|_{\infty}+\alpha\left( B_X\log d_X+ \sigma(\boldX) \sqrt{N\log d_X} \right)\|\cX\|_{\infty} \right] \left\| \Uml\Hml \right\|_{2,\infty}.
    \end{align*}

    For $\beta_1$, similarly we have,
    \begin{align*}
        \left\| \Pm\left(\boldG^{(m)} - \boldG^{(m,l)}\right)\Uml\Hml \right\| &\leq \underbrace{\left\| \left(\boldE_{\cdot,l}^R\boldE_{\cdot,l}^{R\top} - \mathbf{D}_l^R\right)\Uml\Hml \right\|}_{\beta_{1,1}}\\
        &+ \underbrace{\left\| \left(\cR_{:,l}\boldE_{\cdot,l}^{R\top} + \boldE_{\cdot,l}^R\cR_{:,l}^{\top} - 2\hat{\mathbf{D}}_l^R\right)\Uml\Hml \right\|}_{\beta_{1,2}}\\
        &+\alpha\left[ \left\| \left(\boldE_{\cdot,l}^X\boldE_{\cdot,l}^{X\top} - \mathbf{D}_l^X\right)\Uml\Hml \right\|\right.\\
        &\left.+  \left\| \left(\cX_{:,l}\boldE_{\cdot,l}^{X\top} + \boldE_{\cdot,l}^X\cX_{:,l}^{\top} - 2\hat{\mathbf{D}}_l^X\right)\Uml\Hml \right\| \right],
    \end{align*}
    where $\mathbf{D}^R_l$ and $\hat{\mathbf{D}}^R_l$ are diagonal matrices defined as,
    \begin{align*}
        \left(\mathbf{D}_l^R\right)_{i,i} &= E_{i,l}^{R2},\\
        \left(\hat{\mathbf{D}_l^R}\right)_{i,i} &= \cR_{i,l}E_{i,l}^R,
    \end{align*}
    and $\mathbf{D}^X_l$ and $\hat{\mathbf{D}}_l^X$ are defined accordingly. For simplicity, we only show the bound for $\beta_{1,1}$ and $\beta_{1,2}$. The terms with respect to $X$ can be bounded in the same manner.

    For $\beta_{1,1}$, following the same treatment, one has
    \begin{equation*}
        \left\| \mathbf{E}_{\cdot,l}^R \right\|_2\left\| \left(\boldE_{\cdot,l}^R\right)^{\top}\Uml\Hml \right\|_2 \lesssim \left(B_R\log d_R + \sigma(\boldR)\sqrt{N\log d_R}\right)^2 \left\|\Uml\Hml \right\|_{2,\infty},
    \end{equation*}
    and
    \begin{equation*}
        \left\| \left( \mathbf{D}_l^R- \mathbb{E}\left[\mathbf{D}_l^R\right] \right) \Uml\Hml \right\| \lesssim (B_R^2\log d_R + \sigma(\boldR)^2N\sqrt{\log d_R})\left\| \Uml\Hml \right\|_{2,\infty},
    \end{equation*}
    with
    \begin{equation*}
        \left\| \mathbb{E}\left[\mathbf{D}_l^R\right]\Uml\Hml \right\| \leq \sigma(\boldR)^{2}.
    \end{equation*}
    Therefore, combining these yields,
    \begin{align*}
        \beta_{1,1}&\leq \left\| \mathbf{E}_{\cdot,l}^R \right\|_2\left\| \left(\boldE_{:,l}^R\right)^{\top}\Uml\Hml \right\|_2 + \left\| \mathbf{D}_l^R \Uml\Hml \right\|\\
        &\lesssim \left(B_R\log d_R + \sigma(\boldR)\sqrt{N\log d_R}\right)^2 \left\|\Uml\Hml \right\|_{2,\infty} + \sigma(\boldR)^{2}.
    \end{align*}

    As for $\beta_{1,2}$, with probability at least $1-O(d_R^{-13})$, we have
    \begin{align*}
        \left\| \cR_{\cdot,l}\boldE_{\cdot,l}^{R\top}\Uml\Hml \right\| &\lesssim \left( B_R\log d_R + \sigma(\boldR)\sqrt{N\log d_R} \right)\left\|\cR^{\top}\right\|_{2,\infty}\left\| \Uml\Hml \right\|_{2,\infty},\\
        \left\| \boldE_{\cdot,l}^{R}\cR_{\cdot,l}^\top\Uml\Hml \right\| &\lesssim \left(B_R\sqrt{\log d_R} + \sigma(\boldR)\sqrt{N}\right)\left\| \cR^\top \right\|,\\
        \left\| \hat{\mathbf{D}}_l^R\Uml\Hml \right\| &\lesssim \left( B_R\log d_R + \sigma(\boldR)\sqrt{N\log d_R} \right)\|\cR\|_{\infty}\left\|\Uml\Hml \right\|_{2,\infty}.
    \end{align*}
    Combining these results yields,
    \begin{equation*}
        \beta_{1,2} \lesssim \left( B_R\log d_R + \sigma(\boldR)\sqrt{N\log d_R} \right)\|\cR^\top\|_{2,\infty}
    \end{equation*}

    Then, combining $\beta_{1,1}$ and $\beta_{1,2}$, we can obtain,
    \begin{align*}
        \beta_1 &\lesssim \left(B_R\log d_R + \sigma(R)\sqrt{N\log d_R}\right)^2 \left\|\Uml\Hml \right\|_{2,\infty} + \sigma(\boldR)^{2}\\ 
        &+ \left( B_R\log d_R + \sigma(\boldR)\sqrt{N\log d_R} \right)\|\cR^\top\|_{2,\infty}\\
        &+ \alpha\left[ \left(B_X\log d_X + \sigma(\boldX)\sqrt{N\log d_X}\right)^2 \left\|\Uml\Hml \right\|_{2,\infty} + \sigma(\boldX)^{2}\right.\\ &+\left. \left( B_X\log d_X + \sigma(\boldX)\sqrt{N\log d_X} \right)\|\cX^\top\|_{2,\infty} \right].
    \end{align*}

    Combining $\beta_1$ and $\beta_2$, we have
    \begin{align*}
        \left\|\left( \boldG^{(m)} - \boldG^{(m,l)} \right)\Uml \right\| &\lesssim  \left(B_R\log d_R + \sigma(\boldR)\sqrt{N\log d_R}\right)^2 \left\|\Uml\Hml \right\|_{2,\infty} + \sigma(\boldR)^{2}\\
        &+ \left( B_R\log d_R + \sigma(\boldR)\sqrt{\log d_R} \right)\|\cR^\top\|_{2,\infty}\\
        &+ \alpha\left[ \left(B_X\log d_X + \sigma_(\boldX)\sqrt{N\log d_X}\right)^2 \left\|\Uml\Hml \right\|_{2,\infty}+ \sigma(\boldX)^{2}\right.\\
        &\left. + \left( B_X\log d_X + \sigma(\boldX)\sqrt{N\log d_X} \right)\|\cX^\top\|_{2,\infty} \right].
    \end{align*}
    
    Therefore, by applying Davis-Kahan $\sin \Theta$ Theorem, with probability no less than $1-O(d^{-13})$,
    \begin{align*}
        \left\| \Um\UmT-\Uml\UmlT \right\| &\leq \frac{\left\|\left( \boldG^{(m)} - \boldG^{(m,l)} \right)\Uml \right\|}{\sigma_K(\boldG^{(m)})-\sigma_{K+1}(\boldG^{(m)}) - \left\|\boldG^{(m)}-\boldG^{(m,l)}\right\|}\\
        &\lesssim \frac{1}{\sigma_K(\cG)}\left\|\left( \boldG^{(m)} - \boldG^{(m,l)} \right)\Uml \right\|\\
        &\lesssim \frac{1}{\sigma_K(\cG)}\left(B_R\log d_R + \sigma(\boldR)\sqrt{N\log d_R}\right)^2 \left\|\Um\Hm\right\|_{2,\infty}\\
        &+ \frac{\sigma(\boldR)^{2}}{\sigma_K(\cG)} + \frac{1}{\sigma_K(\cG)}\left( B_R\log d_R + \sigma(\boldR)\sqrt{N\log d_R} \right)\|\cR^\top\|_{2,\infty}\\
        &+ \alpha\left[ \frac{1}{\sigma_K(\cG)}\left(B_X\log d_X + \sigma(\boldX)\sqrt{N\log d_X}\right)^2 \left\|\Um\Hm\right\|_{2,\infty} \right.\\
        &+\left. \frac{\sigma(\boldX)^{2}}{\sigma_K(\cG)} + \frac{1}{\sigma_K(\cG)}\left( B_X\log d_X + \sigma(\boldX)\sqrt{N\log d_X} \right)\|\cX^\top\|_{2,\infty} \right].
    \end{align*}

\subsection{Proof of Lemma \ref{lemma:33}}
\label{sec:lem_33}
    By Davis-Kahan $\sin \boldTheta$ Theorem \citep[Theorem 2.2.1]{Chen2021}, we have
    \begin{align*}
        \left\| \Um\UmT - \boldU\boldU^\top\right\| &\leq \frac{\left\|(\boldG-\boldG^{(m)})\Um\right\|}{\sigma_K(\boldG^{(m)})-\sigma_{K+1}(\boldG)} \leq \frac{2\left\|(\boldG-\boldG^{(m)})\Um\right\|}{\sigma_K(\cG)}\\
        &\leq \underbrace{\frac{2\left\| \di\left(\boldG - \boldG^{(m)} \right)\boldU^{(m)} \right\|}{\sigma_K(\cG)}}_{\alpha_1} + \underbrace{\frac{2\left\| \od\left(\boldG - \boldG^{(m)} \right)\boldU^{(m)} \right\|}{\sigma_K(\cG)}}_{\alpha_2} ,
    \end{align*}
    where the second inequality comes from Weyl's inequality,
    \begin{align*}
        \sigma_K\left(\boldG^{(m)}\right) &\geq \sigma_K\left(\cG\right) - \left\|\boldG^{(m)} - \cG\right\| \stackrel{(a)}{\geq} \sigma_K(\cG)-\Tilde{C}\zeta_{\text{op}}(\boldG) \stackrel{(b)}{\geq} \frac{3}{4}\sigma_K(\cG),\\
        \sigma_{K+1}(\boldG) &\leq \sigma_{K+1}(\cG)+\|\boldG - \cG\| \stackrel{(c)}{\leq} \Tilde{C}\zeta_{\text{op}}(\boldG) \stackrel{(d)}{\leq} \frac{1}{4}\sigma_K(\cG),
    \end{align*}
    with $\Tilde{C}\geq 0$ a constant. (a) comes from Lemma \ref{lemma:29}; (c) comes from Lemma \ref{lemma:individualbound}; (b) and (d) come from $\zeta_{\text{op}}(\boldG) \ll \sigma_K(\cG)$.

    \textbf{Bounding} $\alpha_1$. Since $\di\left(\boldG^{(m)}\right)=\di(\cG)$, we have
    \begin{equation*}
        \alpha_1 \leq \frac{2\left\| \di\left(\boldG - \boldG^{(m)} \right)\right\| \left\|\Um \right\|}{\sigma_K(\cG)} \lesssim \kappa_{\max}^2\sqrt{\frac{\mu K}{N}}\zeta_{\text{op}}(\boldR) + \alpha \kappa_{\max}^2\sqrt{\frac{\mu K}{N}}\zeta_{\text{op}}(\boldX),
    \end{equation*}
    where the second inequality comes from Lemma \ref{lemma:individualbound}.

    \textbf{Bounding} $\alpha_2$. Since $\od(\boldG-\boldG^{(m)})$ is supported on the $m$th row and $m$th column, we have,
    \begin{align*}
        \left\| \od\left(\boldG - \boldG^{(m)} \right) \Um\right\| &\leq \left\| \od\left(\boldG - \boldG^{(m)} \right) \Um\right\|_F \stackrel{(a)}{\lesssim} \left\| \od\left(\boldG - \boldG^{(m)} \right) \Um\boldH^{(m)}\right\|_F\\
        &\leq \left\|\mathcal{P}_{m,\cdot}\left(\boldG - \boldG^{(m)} \right)\Um\Hm \right\| + \left\|\mathcal{P}_{\cdot,m}\left(\boldG - \boldG^{(m)} \right)\Um\Hm \right\|\\
        &\stackrel{(b)}{=} \underbrace{\left\|\boldE_{m,\cdot}\left[\mathcal{P}_{-m,\cdot}(\mathcal{L})\right]^\top \Um\Hm \right\|_F}_{\alpha_{2,1}} + \underbrace{\left\|\left(\boldG - \boldG^{(m)} \right)_{m,\cdot} \right\|_2 \left\|\Um_{m,\cdot} \Hm \right\|_2}_{\alpha_{2,2}},
    \end{align*}
    where (a) is from Lemma \ref{lemma:spectralbound} and (b) comes from the definition of $\boldG$ and $\boldG^{(m)}$.

    For $\alpha_{2,1}$, note that
    \begin{equation*}
        \mathbf{E} = \begin{bmatrix}
            \mathbf{E}^{R} & \alpha \mathbf{E}^{X}
        \end{bmatrix}, 
    \end{equation*}
    and the similar property holds for $\mathbf{M}$. Therefore, $\alpha_{2,1}$ can be bounded by,
    \begin{align*}
        \alpha_{2,1} &\leq \left\|\boldE_{m,\cdot}^{R}\left[\mathcal{P}_{-m,\cdot}(\boldR)\right]^\top \Um\Hm \right\|_F + \alpha \left\|\boldE_{m,\cdot}^{X}\left[\mathcal{P}_{-m,\cdot}(\boldX)\right]^\top \Um\Hm \right\|_F\\
        &\lesssim \left(\zeta_{\text{op}}(\boldR) + \alpha \zeta_{\text{op}}(\boldX)\right) \left( \left\| \Um\Hm\right\|_{2,\infty} + \sqrt{\frac{\mu K}{N}} \right),
    \end{align*}
    with probability exceeding $1-O(d^{-11})$, where the second inequality comes from Lemma \ref{lemma:32}.

    As for $\alpha_{2,2}$, by Lemma \ref{lemma:29}, we have
    \begin{align*}
        \alpha_{2,2} \leq \left\| \boldG - \boldG^{(m)} \right\|_2\left\|\Um\Hm \right\|_{2,\infty} &\leq \left(\left[\sigma(\boldR)^2+\alpha\sigma(\boldX)^2 \right]\sqrt{N\max\{J, W \}}\log d + \right. \\  &\left.\left[\sigma(\boldR) \sigma_1(\cR) + \alpha \sigma(\boldX)\sigma_1(\cX) \right] \kappa_{\max}^2\sqrt{\mu K\log d}\right) \left\|\Um\Hm \right\|_{2,\infty}.
    \end{align*}

    {with the assumption $N\gtrsim \kappa_{\max}^4\mu K$}, one can see that $\alpha_{2,1}\lesssim \alpha_{1,1}$. Therefore,
    \begin{align*}
        \left\| \od\left(\boldG - \boldG^{(m)} \right) \Um\right\| \lesssim \left(\zeta_{\text{op}}(\boldR) + \alpha \zeta_{\text{op}}(\boldX)\right) \left( \left\| \Um\Hm\right\|_{2,\infty} + \sqrt{\frac{\mu K}{N}} \right),
    \end{align*}
    and
    \begin{align*}
        \alpha_2 \lesssim \frac{\zeta_{\text{op}}(\boldR) + \alpha \zeta_{\text{op}}(\boldX)}{\sigma_K(\cG)} \left( \left\| \Um\Hm\right\|_{2,\infty} + \sqrt{\frac{\mu K}{N}} \right).
    \end{align*}

    Then, combining $\alpha_1$ and $\alpha_2$, we have,
    \begin{align*}
        \left\|\Um\UmT - \boldU\boldU^\top \right\| &\leq \alpha_1 + \alpha_2\\
        &\lesssim \kappa_{\max}^2\sqrt{\frac{\mu K}{N}}\zeta_{\text{op}}(\boldR) + \alpha \kappa_{\max}^2\sqrt{\frac{\mu K}{N}}\zeta_{\text{op}}(\boldX)\\
        &+ \frac{\zeta_{\text{op}}(\boldR) + \alpha \zeta_{\text{op}}(\boldX)}{\sigma_K(\cG)} \left( \left\| \Um\Hm\right\|_{2,\infty} + \sqrt{\frac{\mu K}{N}} \right)\\
        &\lesssim \kappa_{\max}^2\frac{\zetaop(\cG)}{\sigma_K(\cG)} \left( \left\| \Um\Hm\right\|_{2,\infty} + \sqrt{\frac{\mu K}{N}}\right),
    \end{align*}
    with probability at least $1-O(d^{-11})$

\section{Discussion on incoherence parameter}
\label{sec:dis_incoherence}

\subsection{The bound on incoherence parameters}
\label{sec:incoherence bound}

First, we present the following lemma for completeness.

\begin{lemma}[Lemma 3.6 in \citet{Mao2019}]
    \label{lemma:pi_bound}
    If $\bm{\pi}_i^{*}\stackrel{i.i.d}{\sim} \text{Dirichlet}(\bm{\beta})$ for every $i\in [N]$, let $\beta_{\max}=\max_k{\beta}_k$, $\beta_{\min}=\min_k\beta_k$, $\beta_0=\sum_{k=1}^K\beta_k$ and $\nu = \beta_0/\beta_{\min}$, then
    \begin{align*}
        \rP\left(\sigma_K(\boldPi)^2\geq \frac{N}{2\nu(1+\beta_0)}\right)  &\geq 1 - K\exp\left(-\frac{N}{36\nu^2(1+\beta_0)^2}\right),\\
        \rP\left( \sigma_1(\boldPi)^2 \leq \frac{3N(\beta_{\max}+\|\bm{\beta}\|^2)}{2\beta_0(1+\beta_0)} \right) &\geq 1 - K\exp\left(-\frac{N}{36\nu^2(1+\beta_0)^2}\right),\\
        \rP\left( \kappa(\boldPi)^2 \leq \frac{3(\beta_{\max}+\|\bm{\beta}\|^2)}{\beta_{\min}} \right) &\geq 1- K\exp\left(-\frac{N}{36\nu^2(1+\beta_0)^2}\right).
    \end{align*}
\end{lemma}

The following lemma is adapted from \citet[Lemma 3.6]{Mao2019}, since Beta distribution is a special case of Dirichlet distribution.
\begin{lemma}
    \label{lemma:theta_bound}
    If $\theta_{jk}\stackrel{i.i.d.}{\sim} \text{Beta}(a,b)$ for every $(j,k)\in [J]\times [K]$, let $c_1=ab/((a+b)^2(a+b+1))$, $c_2=Ka^2/(a+b)^2 + c_1$, then,
    \begin{align*}
        \rP\left( \sigma_K(\boldTheta)^2\geq \frac{c_1J}{2} \right) &\geq 1-K\exp\left(-\frac{c_1^2J}{80K^2}\right),\\
        \rP\left( \sigma_1(\boldTheta)^2\leq \frac{3c_2J}{2} \right) &\geq 1 - K\exp\left(-\frac{c_1^2J}{80K^2}\right),\\
        \rP\left(\kappa(\boldTheta)^2\leq \frac{3c_2}{c_1}\right) &\geq 1-2K\exp\left(-\frac{c_1^2J}{80K^2}\right).
    \end{align*}
\end{lemma}
\begin{proof}
    Let $\mathbf{A}=\boldTheta^\top \boldTheta\in \mathbb{R}^{K\times K}$, and $\mathbf{C}_j=\bm{\theta}_{j,:} \bm{\theta}_{j,:}^{\top} - E\left[ \bm{\theta}_{j,:}\bm{\theta}_{j,:}^{\top} \right]\in \mathbb{R}^{K\times K} $. From simple matrix multiplication, one can see $\mathbf{A} - E[\mathbf{A}] = \sum_{j=1}^J\mathbf{C}_j$.

    By properties of Beta distribution,
    \begin{equation*}
        E\left[ \bm{\theta}_{j,:}\bm{\theta}_{j,:}^{\top} \right] = \frac{a^2}{(a+b)^2}\mathbf{1}\mathbf{1}^{\top} + \di \left( \frac{ab}{(a+b)^2(a+b+1)} \right).
    \end{equation*}
    From Weyl's inequality,
    \begin{align*}
        \lambda_K(E[\mathbf{A}]) &\geq \frac{ab}{(a+b)^2(a+b+1)} = c_1,\\
        \lambda_1(E[\mathbf{A}]) &\leq \frac{Ka^2}{(a+b)^2} + \frac{ab}{(a+b)^2(a+b+1)} = c_2.
    \end{align*}
    From the definition, $c_1\leq 1$ and $c_2 \leq K+1$.
    
    Since $|\theta_{jk}|\leq 1$ for all $(j,k)\in [J]\times [K]$, $\| \bm{\theta}_{j,:}\bm{\theta}_{j,:}^{\top} \| = \|\bm{\theta}_{j,:}\|^2 \leq K $. Hence, by triangle inequality,
    \begin{equation*}
        \left\|\mathbf{C}_j\right\| \leq \|\mathbf{A}\|+\left\| E[\mathbf{A}] \right\| \leq K+ c_2\leq 2K+1.
    \end{equation*}
    By Jensen's inequality,
    \begin{equation*}
        \left\|E\left[ \mathbf{C}_j^2 \right] \right\| \leq E\left[ \left\|\mathbf{C}_j^2\right\| \right] \leq 9K^2.
    \end{equation*}

    Then, by matrix Bernstein inequality \citep{Tropp2015}, 
    \begin{equation*}
        \rP\left(\sigma_1\left( \mathbf{A}-E[\mathbf{A}] \right) \geq t\right) \leq K\exp\left\{- \frac{t^2}{18K^2J + Kt}\right\}.
    \end{equation*}

    By Weyl's inequality,
    \begin{align*}
        \sigma_K(E[\mathbf{A}]) = \sigma_K\left( \sum_{j=1}^JE\left[\bm{\theta}_{j,:}\bm{\theta}_{j,:}^\top\right] \right) &\geq Jc_1,\\
        \sigma_1(E[\mathbf{A}]) = \sigma_1\left( \sum_{j=1}^JE\left[\bm{\theta}_{j,:}\bm{\theta}_{j,:}^\top\right] \right) &\leq Jc_2.
    \end{align*}

    Hence, let $t=c_1J/2$, with probability at least
    \begin{equation*}
        1 - K\exp\left(-\frac{c_1^2J^2}{72K^2J + 4Kc_1J}\right) \geq 1 - K\exp\left(-\frac{c_1^2J}{80K^2}\right), 
    \end{equation*}
    we have
    \begin{equation*}
        \sigma_K(\boldTheta)^2 = \sigma_K(\mathbf{A}) \geq \sigma_K(E[\mathbf{A}]) - |\sigma_K(\mathbf{A}) - \sigma_K(E[\mathbf{A}])| \geq \sigma_K(E[\mathbf{A}]) - \sigma_1\left(\mathbf{A} - E[\mathbf{A}]\right) \geq  \frac{c_1J}{2}.
    \end{equation*}

    Similarly, let $t=c_2J/2$, with probability at least $1- K\exp\left( -c_2^2J/(80K^2) \right)$,
    \begin{equation*}
        \sigma_1(\boldTheta)^2 = \sigma_1(\mathbf{A}) \leq \sigma_1(E[\mathbf{A}]) + |\sigma_1(\mathbf{A}) - \sigma_1(E[\mathbf{A}])| \geq \sigma_1(E[\mathbf{A}]) + \sigma_1\left(\mathbf{A} - E[\mathbf{A}]\right) \geq  \frac{3c_2J}{2}.
    \end{equation*}

    Therefore, since $c_1<c_2$, by union bound, with probability exceeding $1- 2K\exp\left( -c_1^2J/(80K^2) \right)$,
    \begin{equation*}
        \kappa(\boldTheta)^2 = \frac{\sigma_1(\boldTheta)^2}{\sigma_K(\boldTheta)^2}\leq \frac{3c_2}{c_1}.
    \end{equation*}
\end{proof}

\begin{lemma}
    \label{lemma:m_bound}
    If $m_{jk}\stackrel{i.i.d.}{\sim} N(\mu,\sigma^2)$ truncated to be in $[-\xi, \xi]$, for every $(j,k)\in [J]\times [K]$. For notational simplicity, let $E[m_{jk}]=\tilde{\mu}$, $\text{var}(m_{jk})=\tilde{\sigma}^2$. Let $\tilde{c}_1=\tilde{\sigma}^2$, $\tilde{c}_2=K\tilde{\mu}^2 + \tilde{\sigma}^2$, then,
    \begin{align*}
        \rP\left( \sigma_K(\boldM)^2\geq \frac{c_1W}{2} \right) &\geq 1-K\exp\left(-\frac{\tilde{c}_1^2W}{80K^2\xi^4}\right),\\
        \rP\left( \sigma_1(\boldTheta)^2\leq \frac{3c_2W}{2} \right) &\geq 1 - K\exp\left(-\frac{\tilde{c}_1^2W}{80K^2\xi^4}\right),\\
        \rP\left(\kappa(\boldTheta)^2\leq \frac{3c_2}{c_1}\right) &\geq 1-2K\exp\left(-\frac{\tilde{c}_1^2W}{80K^2\xi^4}\right).
    \end{align*}
\end{lemma}
\begin{proof}
    The proof is similar to the proof of Lemma \ref{lemma:theta_bound}, and thus omitted here. Notably, for any distribution in the interval $[-\xi,\xi]$, the largest variance it can reach comes from the point masses at $\{-\xi,\xi\}$ with equal $1/2$ probability. This gives an upper bound on $\tilde{\sigma}^2$ by,
    \begin{equation*}
        \tilde{\sigma}^2\leq \frac{1}{2}\xi^2 + \frac{1}{2}\xi^2 = \xi^2.
    \end{equation*}
    Similarly, $|\tilde{\mu}|\leq \xi$. Therefore, $c_1$ and $c_2$ can be upper bounded by $\xi$. Following the same procedures as the proof in Lemma \ref{lemma:theta_bound}, we are able to obtain the result.
\end{proof}

\begin{lemma}
    Under Assumption \ref{assump:high_prob_cons}, if all $\boldPi$, $\boldTheta$ and $\boldM$ are full rank, then
    \begin{equation*}
        \mu \leq \min\left\{ \frac{N}{K\sigma_K(\boldPi)}, \frac{J\kappa(\boldPi)}{\sqrt{K}\sigma_K(\boldTheta)}, \frac{NJ}{\sqrt{K}\sigma_K(\boldPi)\sigma_K(\boldTheta)}, \frac{W\xi\kappa(\boldPi)}{\sqrt{K}\sigma_K(\boldM)}, \frac{\xi NW}{ \sqrt{K}\sigma_K(\boldsymbol{\Pi}) \sigma_K(\mathbf{M})}\right\},
    \end{equation*}
    with probability exceeding $1- O\left(K(\min\{N,J,K\})^{-3}\right)$.
\end{lemma}
\begin{proof}

    First, we here consider the event
    \begin{align*}
        \mathbf{\Omega} = \biggl\{(\boldPi, \boldTheta, \boldM): &\sigma_K(\boldPi)\gtrsim \frac{\sqrt{N}}{K}, \sigma_1(\boldPi)\lesssim\sqrt{\frac{N}{K}}, \kappa(\boldPi)\lesssim \sqrt{K}, \\
        &\sigma_K(\boldTheta)\gtrsim \frac{\sqrt{J}}{K}, \sigma_1(\boldTheta)\lesssim \sqrt{\frac{J}{K}}, \kappa(\boldTheta)\lesssim \sqrt{K}, \\
        &\sigma_K(\boldM)\gtrsim \frac{\sqrt{W}}{K}, \sigma_1(\boldM)\lesssim \sqrt{\frac{W}{K}}, \kappa(\boldM)\lesssim \sqrt{K} \biggr\}
    \end{align*}
    By Assumption \ref{assump:high_prob_cons} and Lemma \ref{lemma:pi_bound}, \ref{lemma:theta_bound}, and \ref{lemma:m_bound},
    \begin{equation*}
        \rP(\mathbf{\Omega}) \geq 1- O\left(K(\min\{N,J,K\})^{-3}\right).
    \end{equation*}
    Hence, we here further discuss the bound for $\mu$ under our model assumptions under $\mathbf{\Omega}$. Recall that $\mu=\max\{\mu(\cR), \mu(\cX)\}$,
\begin{equation*}
    \mu(\mathcal{R}) = \max\left\{\frac{NJ\|\cR\|_{\infty}^2}{\|\cR\|_F^2}, \frac{N\|\boldU_\mathcal{R}\|_{2,\infty}^2}{K}, \frac{J\|\mathbf{V}_\mathcal{R}\|_{2,\infty}^2}{K} \right\},
\end{equation*}
and $\mu(\cX)$ is defined similarly. We bound the three terms respectively as follows.

Since the set of pure subjects $S$ is universal across $\cR$ and $\cX$, one has
\begin{equation*}
    \|\boldU_{\cR}\|_{2,\infty} = \|\boldU_{S,:}(\cR)\|_{2,\infty} \leq \|\boldU_{S,:}(\cR)\| = \frac{1}{\sigma_K(\boldsymbol{\Pi})},
\end{equation*}
where $\boldU_{S,:}(\cR)$ denotes the $S$ rows of left singular subspace of $\cR$, and the last equality comes from similar reasoning as in Lemma \ref{lemma:eigenvaluePiU}.

\indent For $\mathbf{V}_{\cR}$, one has
\begin{align*}
    \|\mathbf{V}_{\cR}\|_{2,\infty} &= \|\boldTheta \boldU_{S,:}(\cR)^{-\top} \boldsymbol{\Sigma}_{\cR}^{-1}\|_{2,\infty}\\
    &\leq \|\boldTheta\|_{2,\infty}\|\boldU_{S,:}(\cR)^{-1}\| \|\boldsymbol{\Sigma}_{\cR}^{-1}\|\\
    &\leq \|\boldTheta\|_{2,\infty}\frac{\sigma_1(\boldsymbol{\Pi})}{\sigma_K(\boldsymbol{\Pi}\boldsymbol{\Theta}^\top)}\\
    &\leq \frac{\sqrt{K}\kappa(\boldsymbol{\Pi})}{\sigma_K(\boldsymbol{\Theta})},
\end{align*}
where the first inequality comes from \citet[Proposition 6.5]{Cape2018}, the second inequality comes from similar reasoning as in Lemma \ref{lemma:keigenvaluebound}, the last inequality comes from $\|\boldTheta\|_{\infty}\leq 1$.

For $\|\cR\|_{\infty}/\|\cR\|_F$, one has
\begin{align*}
    \frac{\|\cR\|_{\infty}}{\|\cR\|_F} &\leq \frac{1}{\left(\sum_{i=1}^K \sigma_i^2(\cR) \right)^{1/2}}\\
    &\leq \frac{1}{\sqrt{K}\sigma_K(\boldsymbol{\Pi} \boldsymbol{\Theta}^\top)}\\
    &\leq \frac{1}{\sqrt{K}\sigma_K(\boldsymbol{\Pi}) \sigma_K(\boldsymbol{\Theta})}.
\end{align*}

Similarly, for $\mu(\cX)$, similar bounds can be derived,
\begin{align*}
    \|\boldU_{\cX}\|_{2,\infty} &\leq \frac{1}{\sigma_K(\boldsymbol{\Pi})},\\
    \|\mathbf V_{\cX}\|_{2,\infty} &\leq \frac{\sqrt{K}\xi \kappa(\boldsymbol{\Pi})}{\sigma_K(\mathbf{M})},\\
    \frac{\|\cX\|_{\infty}}{\|\cX\|_F} &\leq \frac{\xi}{ \sqrt{K}\sigma_K(\boldsymbol{\Pi}) \sigma_K(\mathbf{M})},
\end{align*}
where the last inequality applies $\|\cX\|_{\infty} \leq \|\boldsymbol{\Pi}\|_{2,\infty} \|\mathbf{M}\|_{2,\infty}\leq \sqrt{K}\xi $.

Therefore, the joint incoherence parameter is bounded by,
\begin{align*}
    \mu \leq \min\left\{ \frac{N}{K\sigma_K(\boldPi)}, \frac{J\kappa(\boldPi)}{\sqrt{K}\sigma_K(\boldTheta)}, \frac{NJ}{\sqrt{K}\sigma_K(\boldPi)\sigma_K(\boldTheta)}, \frac{W\xi\kappa(\boldPi)}{\sqrt{K}\sigma_K(\boldM)}, \frac{\xi NW}{ \sqrt{K}\sigma_K(\boldsymbol{\Pi}) \sigma_K(\mathbf{M})}\right\}.
\end{align*}
\end{proof}

\subsection{Discussions on the Assumptions}
Under the assumptions of Corollary \ref{coro:incoherence bound}, $\sigma_K\sigma_K(\cR)\geq \sigma_K(\Pi)\sigma_K(\Theta) \gtrsim \sqrt{NJ}$. Therefore, $\sigma_K(\cR)^2\gg \kappa(\cR)^2\zetaop(\boldR)$, and Assumption \ref{assump:information} is satisfied. Hence, the assumptions on the signal strength is achieved under these commonly applied priors.

\end{document}